\documentclass[11pt,a4paper]{article}
\usepackage[T1]{fontenc}
\usepackage{lmodern}
\usepackage[margin=27mm]{geometry}
\usepackage{amsmath,amssymb,amsthm,mathtools,booktabs,array,microtype}
\usepackage{tikz}
\usepackage{enumitem}
\usepackage{float}
\usepackage[hidelinks]{hyperref}
\usepackage{algorithm,algorithmic}
\setlist{itemsep=3pt,topsep=5pt}
\newtheorem{theorem}{Theorem}[section]
\newtheorem{lemma}[theorem]{Lemma}
\newtheorem{proposition}[theorem]{Proposition}
\newtheorem{corollary}[theorem]{Corollary}
\theoremstyle{definition}
\newtheorem{definition}[theorem]{Definition}

\newtheorem{target}[theorem]{Desired theorem}
\theoremstyle{remark}

\newcommand{\Ostar}{O^{*}}
\newcommand{\E}{\mathbb E}
\newcommand{\Prb}{\mathbb P}
\newcommand{\OPT}{\operatorname{OPT}}

\newcommand{\ind}{\mathbf 1}
\newcommand{\bits}{\{0,1\}}

\DeclareMathOperator{\rank}{rank}
\hypersetup{pdftitle={Accumulating Pruning Constraints: New Exact Polynomial-Space Algorithms for Max 2-SAT}}
\title{Accumulating Pruning Constraints: \\
	New Exact Polynomial-Space Algorithms for Max 2-SAT\\[3pt]}
\author{Konstantin Kutzkov \\\url{kutzkov@gmail.com}}
\date{}

\begin{document}
\maketitle

\begin{abstract}
We present deterministic polynomial-space algorithms that solve Max 2-SAT
on $n$ variables exactly in time $\Ostar(2^{(1-\gamma)n})$, for a fixed
$\gamma>0$, on several classes of instances whose clause density may grow
with $n$. Our main tool extends the classic DPLL branching algorithm. When
the search gives a variable its less promising value, comparing the two
values of the variable yields a pruning constraint that is satisfied by
some optimal assignment. Accumulating these constraints along each branch
prunes all but an exponentially small fraction of the assignments. The classes are formulas with many biased variables,
semirandom 2-CNF formulas, and Max-Cut on regular expander graphs and on
graphs generated by a stochastic block model. In a semirandom formula, the
pairs of variables that may share a clause are arbitrary, and the clauses
on these pairs are random. For semirandom formulas and block-model graphs,
the algorithms are still exact on every input, and the time bound holds
with high probability over the input.
To our knowledge, these are the first polynomial-space algorithms with such
a running time for nontrivial classes of dense Max 2-SAT instances. Beyond these
classes, the pruning constraints can be added to any branching algorithm
for Max 2-SAT, and we expect them to be of practical relevance for general
formulas, in particular for sparse ones, where each constraint involves only
a few variables.
\end{abstract}

\section{Introduction}

Max-SAT is the optimization version of the fundamental NP-complete SAT problem. Instead of checking whether a Boolean formula over $n$ variables in conjunctive normal form (CNF) is satisfiable, that is, whether some assignment satisfies all its clauses, Max-SAT asks for an assignment satisfying the largest number of
clauses. Exhaustive search solves both problems in time $\Ostar(2^n)$, where the notation $\Ostar$ suppresses factors polynomial in the input size. For formulas whose clauses may be arbitrarily long, no algorithm with running time $\Ostar(2^{(1-\gamma)n})$, for a constant $\gamma>0$, is known for either problem. The $k$-SAT problem, where each clause has at most $k$ literals, is easily solved in time $\Ostar(c_k^n)$ for a constant $c_k<2$. For Max-SAT, no comparable polynomial-space result is known even for $k=2$; note that Max 2-SAT is NP-hard, although 2-SAT is solvable in polynomial time. The only known algorithm that improves on the $\Ostar(2^n)$ bound on all Max 2-SAT instances is the elegant algorithm of Williams~\cite{Williams}, with running time $\Ostar(2^{\omega n/3})$, where $\omega$
is the matrix multiplication exponent. However, it uses space $\Ostar(2^{2n/3})$, which makes it only of theoretical interest. Designing an algorithm with running time $\Ostar(2^{(1-\gamma)n})$ that uses only polynomial space has been an open problem since 2004~\cite{Williams,WOEGINGER}.

\subsection{Previous work}\label{sec:previous}

Exact algorithms for Max 2-SAT, which contains the Max-Cut problem on graphs as a special case, have been analyzed in two ways:
\begin{itemize}
	\item In terms of the number $m$ of clauses, with running time $\Ostar(2^{m/c})$ for a constant $c$. Such bounds improve on $\Ostar(2^n)$ only for sparse instances with $m<cn$. There is a long list of incremental improvements~\cite{SS07,KulikovK07,EdwardsM15,GaspersS17}; the currently best known bound is $\Ostar(2^{m/6.321})$~\cite{GaspersS12}.
	\item In terms of the maximum or average degree $d$ of the formula, with running time $\Ostar(2^{(1-f(d))n})$, where $f(d)\in(0,1)$ and $f(d)\to0$ as $d\to\infty$. Several works give $f(d) = c/d$ for a small constant $c$~\cite{SS07,Marcin07,best_sparse_maxsat}, and Golovnev and Kutzkov~\cite{GK14} and Hastings~\cite{Hastings} give the asymptotically better $f(d) = \log d/d$. These bounds yield $\Ostar(2^{(1-\gamma)n})$ for a fixed $\gamma >0$ only when the number of clauses is linear in $n$.
\end{itemize}

\paragraph{The exponential time hypothesis.}
The exponential time hypothesis (ETH) and its strong version
(SETH)~\cite{impagliazzo2001problems,SETH} assert that 3-SAT requires time
$2^{\Omega(n)}$ and that the savings for $k$-SAT vanish as $k$ grows. Two
recent preprints from frontier AI labs~\cite{Anthropic26ETH,OpenAI26ETH} claim to refute
both hypotheses with a polynomial-space algorithm that solves $k$-SAT in
subexponential time for every fixed $k$. If confirmed, this is a theoretical
breakthrough, but it has little bearing on the results of this paper.
Applied to Max 2-SAT through standard reductions, the algorithm improves on
exhaustive search only at very low density, when a variable occurs on
average in $O(\sqrt{\log n})$ clauses, and the dense instances in our
classes are not affected. Moreover, the current version of the
algorithm has no practical value, because of astronomical constants hidden in the big-Oh notation which make the new upper bound purely asymptotic.

\subsection{Our contribution}
We present a technique that extends the classic DPLL branching algorithm of Davis, Putnam, Logemann, and Loveland~\cite{davis1962machine,davis1960computing}. Consider a branching variable $v$, and suppose that we first explore $v=1$ and then backtrack to $v=0$. Setting $v$ affects only the clauses containing $v$, so for a given assignment of the other variables, which of the two values is better depends only on the values of the neighbors of $v$. After backtracking, the branch $v=0$ therefore only needs the assignments of the neighbors for which $v=0$ is strictly better. We describe them by a linear inequality over the literals of the neighbors, a hard \emph{pruning constraint}. Since the constraint depends only on the neighborhood of $v$, we can efficiently decide which of the two branches would receive the more restrictive constraint. Along a branch, the search creates many such constraints. Each of them alone prunes only part of the search space, but for each of our classes, only an exponentially small fraction of all assignments satisfies the sum of the constraints of its own branch, and listing only these assignments yields the saving.

Our contributions are of two kinds:
\begin{enumerate}
	\item We present new algorithms for several classes of instances, including dense ones with clause density $m/n=\omega(1)$, that run in time  $\Ostar(2^{(1-\gamma)n})$ using only polynomial space.  To the best of our knowledge these are the first results that show such running time bounds for nontrivial classes of dense Max 2-SAT instances, thus making an important step towards resolving the open problem from~\cite{Williams,WOEGINGER}.
	\item The key idea of creating pruning constraints during search is easy to implement and likely to find practical applications. Some of our classes model practical inputs; for example, the stochastic block model is the standard model of networks with communities. The savings that we prove are small, since they result from a conservative analysis of a deliberately simple version of the algorithm. On practical instances we expect much larger savings: after some branching steps, the remaining variables have lower degrees and, relative to their degrees, larger differences between their positive and negative occurrences, which yields more restrictive constraints. Similarly, SAT solvers routinely solve instances with thousands of variables, although their worst-case running time is exponential. In future work, we plan to obtain such practical gains with heuristics based on our constraints and to evaluate them experimentally. The present paper focuses on the theoretical results.
\end{enumerate}

\subsection{Main results}\label{sec:results}

We give polynomial-space algorithms with running time
$\Ostar(2^{(1-\gamma)n})$, for a constant $\gamma>0$, on four classes of
instances. In every class the number of clauses may grow faster than
linearly in $n$, and every class contains instances with a clause on a
constant fraction of all pairs of variables. To our knowledge, no
polynomial-space algorithm with such a running time was previously known
for any of these classes once the density grows with $n$.

All our algorithms are deterministic and return an optimal assignment on
every input. They do not check the conditions that define a class; these
conditions are used only to bound the running time. (The two Max-Cut
algorithms take the constants of their class as parameters.) Two of the
classes consist of semirandom or random instances. For them, the
randomness lies in the input, and the time bound holds in expectation over
the input or with high probability (w.h.p.), that is, with probability
tending to $1$ as $n$ grows. As in the average-case analysis of QuickSort
with a fixed pivot rule, the algorithm is always correct and fast on most
inputs.

\begin{enumerate}
\item \textbf{Formulas with linearly many biased variables.}
A variable is \emph{biased} if the numbers of its positive and negative
occurrences differ by at least a constant fraction of its number of
occurrences, for example if it occurs positively in $3/4$ of its
clauses. Making its more frequent literal true satisfies more clauses on
average over the values of the other variables, although an optimal
assignment may still choose the other value. In the branch where a biased
variable takes its other value, the pruning constraint requires its
neighbors to deviate from their typical values by a constant fraction of
its degree.

Our first result applies when a constant fraction of the variables are
biased. Its simplest form is Corollary~\ref{cor:lineargap}. Consider a
formula in which no clause occurs twice and no clause contains a variable
twice. If, for some constants $\varepsilon,\eta>0$, at least $\eta n$ of
its variables each occur at least $\varepsilon n$ more often with one sign
than with the other, then an optimal assignment can be found in time
$\Ostar(2^{(1-\gamma)n})$ for a constant $\gamma>0$.
Theorem~\ref{thm:bias} covers formulas of every density. It requires
linearly many biased variables and, in addition, a bound on the sum of the
squared numbers of occurrences, so that a few variables cannot occur much
more often than the others.

\item \textbf{Semirandom formulas.}
In a \emph{semirandom} formula, the pairs of variables that may share a
clause are arbitrary, and the clauses on these pairs are random. Such
formulas lie between worst-case and fully random instances;
Feige~\cite{feige2007refuting} formalized this model for 3-SAT. We use two
semirandom models. Fix any graph $G$ whose vertices are the variables. In
the first model, every edge $\{u,v\}$ of $G$ receives one of the four
clauses $(u\vee v)$, $(u\vee\neg v)$, $(\neg u\vee v)$, $(\neg u\vee\neg v)$,
chosen uniformly at random, independently for different edges. In the
second model, each of these four clauses is included independently with
probability $1/4$. In both models, unit clauses may be added arbitrarily,
as long as they do not depend on the random choices.

Theorem~\ref{thm:random_main} shows that on such formulas our algorithm
runs in expected time $\Ostar(2^{(1-\gamma)n})$, and within a bound of
the same form with high probability, for a constant $\gamma>0$ that does
not depend on $G$. If $G$ is itself random, with
each pair of variables joined independently with probability $p$, we
obtain random 2-CNF formulas with about $pn^2/2$ clauses, for every $p$.
Random formulas are the standard test case for the average-case behavior
of satisfiability algorithms, and this result covers all densities at once.
With random signs, no variable is biased and each pruning constraint is
weak. The algorithm adds up many weak constraints after scaling them to a
common size, and the random signs keep a few shared variables from
satisfying many constraints at once.

\item \textbf{Max-Cut on regular expander graphs.}
The Max-Cut problem asks to divide the vertices of a graph into two sides
so that as many edges as possible join vertices on different sides. It is
a special case of Max 2-SAT: introduce a variable for each vertex,
indicating its side, and the two clauses $(u\vee v)$ and
$(\neg u\vee\neg v)$ for each edge $\{u,v\}$. Both clauses are satisfied if
$u$ and $v$ are on different sides; otherwise exactly one is. A graph is
\emph{$d$-regular} if every vertex has exactly $d$ neighbors. Its
\emph{adjacency matrix} $A$ is the symmetric $n\times n$ matrix with
$A_{uv}=1$ if $\{u,v\}$ is an edge and $A_{uv}=0$ otherwise. For a
$d$-regular graph, the largest eigenvalue of $A$ is $d$. We consider
graphs in which all other eigenvalues have absolute value at most
$\lambda_0\sqrt d$, for a constant $\lambda_0$. Such graphs are called
\emph{spectral expanders}. Informally, they look random: any two large
sets of vertices are joined by roughly as many edges as in a random graph
with the same number of edges.

Corollary~\ref{cor:maxcut} shows that if, in addition,
$d\le(1-\varepsilon)n$ for a constant $\varepsilon>0$, then Max-Cut can be
solved in polynomial space and time $\Ostar(2^{(1-\gamma)n})$ for a
constant $\gamma>0$. This is a worst-case
guarantee: it holds for every graph with these properties. Explicit
examples are the Paley graphs, in which every vertex is adjacent to half
of the other vertices, and non-bipartite Ramanujan graphs with
$d\le(1-\varepsilon)n$. Random regular graphs qualify as well: by known
spectral-gap theorems, a uniformly random $d$-regular graph satisfies the
eigenvalue condition with high probability for every
$d\le(1-\varepsilon)n$ (Corollary~\ref{cor:random_regular}). In Max-Cut,
no variable is biased and no sign is random. Instead, all pruning
constraints depend on one global number, the number of vertices on one
side. The algorithm guesses this number, and the expander property makes
the remaining parts of the constraints nearly independent.

\item \textbf{Max-Cut on random graphs with communities.}
The stochastic block model is the most widely studied random model of
networks with community structure~\cite{Abbe18}. The vertices are divided
into $k$ groups, called communities. Every pair of vertices is joined by an
edge independently, with a probability $\pi_{ab}$ that depends only on the
communities $a$ and $b$ of its endpoints. For $k=1$ this is the
Erd\H{o}s--R\'enyi random graph $G(n,p)$, in which every edge is present
independently with probability $p$.

Theorem~\ref{thm:sbm} assumes that the number $k$ of communities is a
constant, that the variances $\pi_{ab}(1-\pi_{ab})$ of the edge indicators
differ by at most a constant factor $\chi$, and that
$n\pi_{ab}(1-\pi_{ab})\ge C\log n$ for a sufficiently large constant $C$.
Under these assumptions, with high probability over the random graph, our
algorithm solves Max-Cut in polynomial space and time
$\Ostar(2^{(1-\gamma)n})$, for a constant $\gamma>0$. The algorithm needs
to know only the constants $k$ and $\chi$, not the communities or the edge
probabilities. For $G(n,p)$ the condition is $np(1-p)\ge C\log n$; roughly,
both the graph and its complement must have expected degree at least a
large multiple of $\log n$. This includes dense random graphs such as
$G(n,1/2)$. Instead of one number, the algorithm guesses the projections
of the assignment onto a few approximate eigenvectors of the graph. It
never needs to find the communities, which need not be easy to detect:
the probabilities inside and between communities may even be equal.
\end{enumerate}

Both Max-Cut results remain valid with arbitrary additional unit clauses,
which give each vertex an integer bonus for one of the two sides. In the
stochastic block model, the unit clauses may even be chosen after the graph
is drawn. In the language of statistical
physics, this is the problem of finding a ground state of an
antiferromagnetic Ising model on the graph, with an external field at
every vertex. Table~\ref{tab:results} summarizes the four results.

\begin{table}[ht]
\centering
\small
\setlength{\tabcolsep}{4pt}
\begin{tabular}{@{}llll@{}}
\toprule
Instances & Density & Time bound holds & Result\\
\midrule
Many biased variables & any & for every instance
  & Thm.~\ref{thm:bias}, Cor.~\ref{cor:lineargap}\\
Semirandom 2-CNF & any & in expectation and w.h.p.
  & Thm.~\ref{thm:random_main}\\
Max-Cut, regular expanders & $d\le(1-\varepsilon)n$ & for every instance
  & Cor.~\ref{cor:maxcut}\\
Max-Cut, block model & $n\pi_{ab}(1-\pi_{ab})\ge C\log n$ & w.h.p.\ over the graph
  & Thm.~\ref{thm:sbm}\\
\bottomrule
\end{tabular}
\caption{The four classes on which our algorithms run in polynomial space
and time $\Ostar(2^{(1-\gamma)n})$.}
\label{tab:results}
\end{table}

\paragraph{Organization.}
Section~\ref{sec:framework} introduces the notation and the simplification
rules, and it explains how pruning constraints arise during branching and
how our algorithms accumulate them.
Sections~\ref{sec:bias}--\ref{sec:sbm} treat the four classes in the order
above. Each of these sections starts with an informal overview, followed
by the algorithm, its correctness proof, and its running-time analysis.
Section~\ref{sec:verification} summarizes a Lean~4 formalization of the
main arguments~\footnote{See \url{https://github.com/konstantinkutzkov/max-2-sat-arxiv}}. The
appendices contain the longer proofs and estimates, the spectral
preprocessing used when the communities are unknown, additional
simplification rules, details of the formal verification, and a proposed
extension of the first result.

\subsection{Relation to earlier methods}

Our pruning constraints continue the clause learning approach of Kulikov
and Kutzkov~\cite{KulikovK07}, which rests on the same observation: one of
the two branches on a variable only needs to keep the assignments that are
strictly better than their counterparts in the other branch. Kulikov and
Kutzkov derived new clauses from this observation and used them to
simplify the formula. This gave polynomial-space algorithms with a fixed
exponential saving for every fixed clause density, a saving that tends to
$0$ as the density grows. We keep the comparisons as linear constraints,
add up the constraints created along a branch, and enumerate only the
completions that satisfy the sum. For the classes of
Section~\ref{sec:results}, only an exponentially small fraction of all
assignments satisfies the sum of its own branch, whatever the density.

Hastings~\cite{Hastings} bounded the number of local maxima of Max 2-SAT
instances, that is, of assignments that no single flip improves. He used
the technique behind these bounds to obtain a faster polynomial-space
algorithm for formulas of large degree. Our approach is related to his
only in spirit. Every optimal assignment is a local maximum, and our
pruning constraints express local optimality at some of the branching
variables. However, we do not count local maxima, and we
use the local conditions only through their sum. Local optimality alone
cannot give an exponential saving on all dense instances: for Max-Cut on
the complete graph, every balanced cut is a local maximum, and there are
about $2^n/\sqrt n$ of them.

Our running-time analyses follow the path of a uniformly random
assignment through the branching tree of a deterministic algorithm. The
number of enumerated assignments equals $2^n$ times the probability that
a random assignment satisfies the summed constraint of its own branch.
Arguments of this type were used by Santhanam~\cite{shrinkage} for the
satisfiability of de Morgan formulas, and for sparse Max-SAT by Sakai,
Seto, and Tamaki~\cite{max_sat_width_reduction} and by Chen and
Santhanam~\cite{best_sparse_maxsat}. There, a concentration argument
along the random path shows that greedily restricting variables shrinks
the instance with high probability. In our setting, the random path
determines which pruning constraints are created, and a concentration
argument shows that the random assignment satisfies their sum only with
exponentially small probability.

\section{Branching with pruning constraints}\label{sec:framework}

Let $F$ be a Max 2-SAT formula on a set $V$ of $n$ variables. In this
section and in Section~\ref{sec:bias}, the formula is \emph{simple}: no
clause occurs twice, and no variable occurs twice in a clause. We write
$\mathbf v$ for a variable and $v,\neg v$ for its two literals. The
\emph{degree} $d^F(v)$ of a literal $v$ is the number of clauses containing
$v$; it is the sum of the numbers $d_1^F(v)$ and $d_2^F(v)$ of unit and
binary clauses containing $v$. We omit the superscript when the formula is
clear. The \emph{neighbors} of a literal $v$ are the literals $x$ with
$(v\vee x)\in F$; they form the multiset $N(v)$, of size $d_2(v)$. Two
variables are \emph{neighbors} if some binary clause contains both. The
\emph{degree} of a variable $\mathbf v$ is $d(v)+d(\neg v)$, and its
\emph{binary degree} is $d_2(v)+d_2(\neg v)$. The number of binary clauses
is $m_2$, and their density is $m_2/n$. The notation $\Ostar$ suppresses
factors polynomial in the input size.

A \emph{partial assignment} $\alpha$ assigns values to a subset of $V$. A
\emph{complete} assignment assigns all variables, and a complete assignment
extending $\alpha$ is a \emph{completion} of $\alpha$. The \emph{score}
$F(x)$ of a complete assignment $x$ is its number of satisfied clauses, and
$\OPT(F)=\max_xF(x)$. Setting a literal $v$ to $1$ satisfies the $d(v)$
clauses containing $v$ and removes $\neg v$ from the clauses containing it:
a binary clause $(\neg v\vee y)$ becomes the unit clause $(y)$, and a unit
clause $(\neg v)$ becomes empty and stays unsatisfied. We write $F[v=1]$
for the resulting formula and include the $d(v)$ satisfied clauses in its
score; $F[v=0]$ is defined in the same way, with $\neg v$ set to $1$. Then
\[
 \OPT(F)=\max\bigl(\OPT(F[v=0]),\OPT(F[v=1])\bigr).
\]
More generally, the \emph{residual formula} $F|_\alpha$ of a partial
assignment $\alpha$ results from setting all its literals. We record the
number of clauses satisfied by $\alpha$ in a separate score $s$, so that the
score of every completion is $s$ plus its score in $F|_\alpha$. The
residual formula keeps all unassigned variables, including those that no
longer occur in any clause. Later sections allow repeated unit clauses, and
Appendix~\ref{app:random_signs} boundedly many binary clauses on a pair;
all statements below hold for each copy.

\paragraph{Soft clauses and hard pruning constraints.}
Our algorithms are DPLL branching
algorithms~\cite{davis1960computing,davis1962machine}: they select a
variable $\mathbf v$ and recursively solve $F[v=0]$ and $F[v=1]$. The
clauses of $F$ are soft constraints: an optimal assignment need not satisfy
all of them. In addition, a branch may carry \emph{pruning constraints},
linear inequalities $\sum_i\ell_i\le b$ over literals $\ell_i$, with
repeated occurrences counted separately. They are hard: the branch
considers only assignments that satisfy all of its pruning constraints,
and this prunes the search space. We write $\mathcal C$ for a collection of
pruning constraints. (Sections~\ref{sec:maxcut} and~\ref{sec:sbm} also use
inequalities with rational coefficients.)

\subsection{Simplification rules and preprocessing}

The first two rules below preserve the score of every assignment up to a
recorded constant; our algorithms apply them once, before branching. The
third rule, unit clause domination, is the starting point of our pruning
constraints. Appendix~\ref{app:simplify} shows how such rules can be
combined with pruning constraints; our algorithms do not need this.

\begin{lemma}[Unit clause cancellation]\label{lem:unit_canc}
Removing one copy of $(v)$ and one copy of $(\neg v)$ decreases the score
of every complete assignment by exactly $1$.
\end{lemma}
\begin{proof}
Exactly one of the two unit clauses is satisfied by every assignment.
\end{proof}

\begin{lemma}[Simplifying two clauses on the same variable pair]\label{lem:simplify}
Let $a,b$ be literals on distinct variables. Replacing
$(a\vee b),(a\vee\neg b)$ by the unit clause $(a)$ decreases the score of
every complete assignment by exactly $1$.
\end{lemma}
\begin{proof}
If $a=1$, the two original clauses contribute $2$ and the new unit contributes
$1$. If $a=0$, exactly one original clause is satisfied and the new unit
contributes $0$. In both cases the difference is $1$.
\end{proof}

The initial procedure \textnormal{\textsc{Preprocess}} makes no assignments.
A pair of clauses is \emph{eligible} if one of the two reductions applies
to it: two opposite units, or two binary clauses that share one literal
and are opposite in the other. As long as an eligible pair exists, the
procedure cancels the two opposite units or replaces the two binary
clauses by their shared unit. Each reduction lowers the
score of every assignment by $1$, which is recorded in $s$. The procedure returns the
reduced formula and $s$. It applies the reductions in a fixed order, so
its output is deterministic; the order does not affect the score
identities.\footnote{Store the clauses as an ordered list, retaining
multiplicities. At each step, scan pairs of clause positions in
lexicographic order and apply the first available rule. When two binary
clauses are replaced, place the new unit at the later clause's position;
otherwise retain the order of surviving clauses.}

\begin{algorithm}[tp]
\caption{\textnormal{\textsc{Preprocess}}$(F)$}\label{alg:preprocess}
\begin{algorithmic}[1]
\STATE $s\gets0$; retain the clauses of $F$ in their input order.
\WHILE{an eligible pair of clauses exists}
  \STATE Select the first eligible pair in the order described above.
  \IF{the pair is $(v),(\neg v)$}
    \STATE Remove both clauses using Lemma~\ref{lem:unit_canc}.
  \ELSE
    \STATE Replace $(a\vee b),(a\vee\neg b)$ by $(a)$ using Lemma~\ref{lem:simplify}.
  \ENDIF
  \STATE $s\gets s+1$.
\ENDWHILE
\STATE \textbf{return} $(F,s)$.
\end{algorithmic}
\end{algorithm}

The score identities justify \textnormal{\textsc{Preprocess}} for every assignment,
not only for an optimum. For distinct input binary clauses, after it terminates a variable pair occurs in
at most two binary clauses. If there are two, their literals are
complementary in both positions: otherwise the two clauses would share a
literal and the loop could apply again. We call this form \emph{canonical}.
Assigning variables creates only units, so the binary clauses do not need
to be preprocessed again at later nodes.

\begin{lemma}[Unit clause domination]\label{lem:unit_dom}
Suppose $d_1(\neg v)=0$ and $d_1(v)\ge d_2(\neg v)$. For every fixed
assignment of the other variables, setting $v=1$ satisfies at least as many
clauses as setting $v=0$.
\end{lemma}
\begin{proof}
Changing $v$ from $0$ to $1$ gains all $d_1(v)$ positive units. Clauses
containing $v$ cannot lose satisfaction. The only possible losses are among
the $d_2(\neg v)$ binary clauses containing $\neg v$, since there are no
units $(\neg v)$. The net gain is therefore at least
$d_1(v)-d_2(\neg v)\ge0$.
\end{proof}

Under the assumptions of the lemma, $\OPT(F)=\OPT(F[v=1])$. In
particular, a pure literal $v$, whose negation does not occur, can be set
to $1$.

\subsection{Pruning constraints}\label{sec:pruning}

\paragraph{What we learn when backtracking.}
Suppose the algorithm branches on $\mathbf v$, explores the branch $v=1$
first, and then backtracks to the branch $v=0$. Setting $\mathbf v$ changes
only the clauses that contain $\mathbf v$; all other clauses are the same in
both branches. Fix the values of all other variables. A clause $(v\vee u)$
with $u$ true, or $(\neg v\vee w)$ with $w$ true, is then satisfied whatever
value $\mathbf v$ takes. Every other clause containing $\mathbf v$ acts as a
unit clause: $(v\vee u)$ with $u$ false acts as $(v)$, and $(\neg v\vee w)$
with $w$ false acts as $(\neg v)$. The two branches therefore differ only by
a contest between these units and the original unit clauses on
$\mathbf v$. If the units $(v)$ are at least as many as the units
$(\neg v)$, the value $v=1$ is at least as good, as in unit clause
domination (Lemma~\ref{lem:unit_dom}), and the branch $v=1$ already covers
this assignment. After backtracking, the branch $v=0$ therefore only needs
the assignments for which the units $(\neg v)$ win, and whether they win
depends only on the values of the neighbors of $v$.

For example, let $\mathbf v$ occur only in the clauses $(v\vee a)$,
$(v\vee b)$, $(v\vee c)$, $(\neg v\vee d)$, and $(\neg v\vee e)$. If $a$ is
false and $d,e$ are true, then $(v\vee a)$ acts as $(v)$, while
$(\neg v\vee d)$ and $(\neg v\vee e)$ are satisfied anyway. The value $v=1$
wins whatever $b$ and $c$ are, so the branch $v=0$ can skip every assignment
with $a=0$ and $d=e=1$. In general, the branch $v=0$ needs only the
assignments in which more of $d,e$ than of $a,b,c$ are false.

\paragraph{The pruning lemma.}
To state this observation precisely, fix the values of all variables other
than $\mathbf v$. The number of clauses satisfied with $v=1$ minus the
number satisfied with $v=0$, the \emph{gain} of $v$, is the number of units
$(v)$ minus the number of units $(\neg v)$ in the contest above:
\begin{equation}
 \mathrm{gain}_v
 =\bigl(d_1(v)+\#\{u\in N(v):u=0\}\bigr)
 -\bigl(d_1(\neg v)+\#\{w\in N(\neg v):w=0\}\bigr).
 \label{eq:gain}
\end{equation}
For a complete assignment $x$ we write $\mathrm{gain}_v(x)$; it does not
depend on the value of $\mathbf v$ in $x$.

\begin{lemma}[Pruning constraints from almost-domination]\label{lem:pruning}
Let $v$ be a literal of an unassigned variable of the residual formula $F$,
and define
\begin{equation}
 C(v):\qquad \sum_{u\in N(v)}\neg u+\sum_{w\in N(\neg v)}w
 \;\le\;d(\neg v)-d_1(v)-1.
 \label{eq:gap_pruning}
\end{equation}
Consider an assignment to the variables of the neighbors of $v$. If it
satisfies $C(v)$, then setting $v=0$ satisfies strictly more clauses than
setting $v=1$, whatever the values of the remaining variables. If it
violates $C(v)$, then setting $v=1$ satisfies at least as many clauses.
Consequently,
\[
 \OPT(F)=\max\bigl(\OPT(F[v=1]),\,\OPT(F[v=0]\wedge C(v))\bigr),
\]
where $F[v=0]\wedge C(v)$ restricts the branch $v=0$ to the assignments that
satisfy $C(v)$.
\end{lemma}
\begin{proof}
By~\eqref{eq:gain}, setting $v=0$ is strictly better if and only if
\[
 d_1(\neg v)+\#\{w\in N(\neg v):w=0\}\ \ge\ d_1(v)+\#\{u\in N(v):u=0\}+1 .
\]
Substituting $\#\{w\in N(\neg v):w=0\}=d_2(\neg v)-\sum_{w\in N(\neg v)}w$,
$\#\{u\in N(v):u=0\}=\sum_{u\in N(v)}\neg u$, and
$d_1(\neg v)+d_2(\neg v)=d(\neg v)$ turns this inequality into $C(v)$. It
involves only the neighbors of $v$, which proves the first two claims. For
the identity, every assignment of the branch $v=0$ that violates $C(v)$ is
matched by the same assignment with $v=1$, which is at least as good.
\end{proof}

\paragraph{A compact representation.}
The lemma says which partial assignments the branch $v=0$ does not need:
the assignments to the neighbors of $v$ that violate $C(v)$. There can be
exponentially many of them, and the single linear inequality $C(v)$
describes them all; the algorithms store and combine this inequality
instead of a list of partial assignments. Its left-hand side counts the neighbors that favor
$v=1$: a false neighbor $u$ of $v$ turns $(v\vee u)$ into a unit $(v)$, and
a true neighbor $w$ of $\neg v$ satisfies $(\neg v\vee w)$, so this clause
cannot become a unit $(\neg v)$. The constraint allows only few of them. In
the example, $C(v)$ reads
\[
 \neg a+\neg b+\neg c+d+e\le1 .
\]
The simplification rules are special cases: if $\neg v$ does not occur, or
if $d_1(\neg v)=0$ and $d_1(v)\ge d_2(\neg v)$, the bound of $C(v)$ is
negative, no assignment satisfies $C(v)$, and the branch $v=0$ can be
skipped altogether.

\paragraph{Using the constraint during branching.}
When the algorithm branches on $\mathbf v$, it substitutes into $C(v)$ the
values of the neighbors that are already assigned: a false literal is
removed, and a true literal is removed while the bound decreases by $1$. A
complementary pair $z,\neg z$ can also be removed while decreasing the
bound by $1$, since $z+\neg z=1$. The remaining literals belong to
unassigned variables, which the algorithm assigns later in the subtree of
the branch $v=0$, so the constraint restricts these later choices. For
example, a negative bound proves that the subtree contains no assignment
satisfying the constraint, and a literal that occurs more often than the
bound must be false in every such assignment. Our algorithms enforce the
constraints of a branch through their sum (Section~\ref{sec:accumulate}).

\subsection{Which branch receives the constraint}\label{sec:preference}

The lemma applies to either literal of the branching variable: $C(v)$
restricts the branch $v=0$, and $C(\neg v)$ restricts the branch $v=1$. A
constraint prunes much if the other value is the better one for most
assignments of the neighbors. In the example, $v$ occurs three times and
$\neg v$ twice. If the neighbors take uniformly random values, the
left-hand side of $C(v)$ has mean $5/2$ and must stay at most $1$, which
only $6$ of the $32$ assignments to $a,\ldots,e$ achieve. The opposite
constraint $C(\neg v)$ reads $a+b+c+\neg d+\neg e\le2$ and keeps $16$ of
them. Since both constraints depend only on the neighborhood of
$\mathbf v$, the better choice can be computed by counting occurrences.

We call the mean of the left-hand side under uniform values minus the bound
the \emph{requirement} of a constraint: how far the remaining assignment
must push the left-hand side below its typical value. For $C(v)$ it equals
$g_{\mathbf v}(F)/2+1$, and for $C(\neg v)$ it equals
$-g_{\mathbf v}(F)/2+1$, where the \emph{effective gap} of a variable
$\mathbf v$ in a formula $H$ is
\begin{equation}
 g_{\mathbf v}(H)=d_2^H(v)-d_2^H(\neg v)
                   +2\bigl(d_1^H(v)-d_1^H(\neg v)\bigr).
 \label{eq:effectivegap}
\end{equation}
Equivalently, $g_{\mathbf v}(H)/2$ is the expected gain of $v$ when the
other variables are uniform: in~\eqref{eq:gain}, an original unit clause
always counts, and a binary clause counts only if its other literal is
false, which happens with probability $1/2$. We call the literal with
nonnegative gap,
\begin{equation}
 \ell_{\mathbf v}(H)=
 \begin{cases}v,&g_{\mathbf v}(H)\ge0,\\
               \neg v,&g_{\mathbf v}(H)<0,
 \end{cases}
 \label{eq:bias_preference}
\end{equation}
the \emph{preferred literal}, and the value of $\mathbf v$ that makes it true
the \emph{preferred value}. Our algorithms add the constraint of the
preferred literal to the branch of the nonpreferred value; the order in
which the two branches are explored does not matter. We call this
constraint the \emph{comparison} of $\mathbf v$, since it compares the two
values of $\mathbf v$. Computing it takes time linear in the degree of
$\mathbf v$. The algorithm of Section~\ref{sec:random} computes the gaps in
its current residual formula, and the algorithm of Section~\ref{sec:bias}
computes them once, in the formula after preprocessing.
Sections~\ref{sec:maxcut} and~\ref{sec:sbm} use a related rule and a relaxed
comparison, defined there.

\paragraph{Max-Cut as a special case.}
Recall from Section~\ref{sec:results} that Max-Cut is the special case of
Max 2-SAT with the clauses $(u\vee v)$ and $(\neg u\vee\neg v)$ for every
edge $\{u,v\}$, where the value of a variable indicates the side of its
vertex. For this encoding, the pruning constraint of
Lemma~\ref{lem:pruning} is a statement about the neighbors of a vertex.
Moving a vertex $v$ to the other side changes the number of cut edges by
the number of its neighbors on its old side minus the number on its new
side. With the placed vertices fixed and the unplaced ones uniform, the
preferred side of $v$ is therefore the side opposite to most of its
already placed neighbors. Placing $v$ on the other side pays off only if,
in the end, at most half of its neighbors lie on its own side. Otherwise,
moving $v$ back cuts more edges, and that cut is examined in the branch on
the preferred side. The branch against the preference therefore keeps only
the completions in which most of the unplaced neighbors of $v$ go to the
other side. Figure~\ref{fig:maxcut} shows an example.

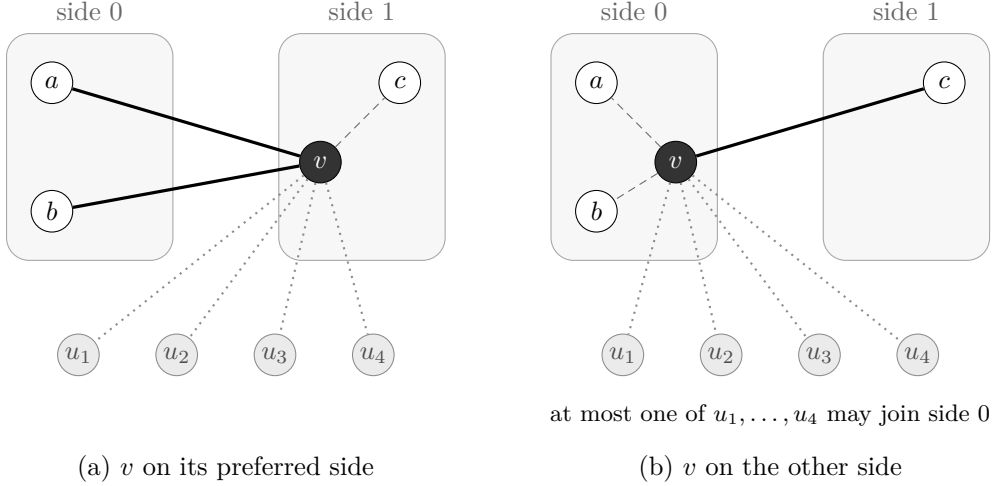
\begin{figure}[t]
\centering
\begin{tikzpicture}[
  vtx/.style={circle,draw,fill=white,inner sep=0pt,minimum size=5.5mm,font=\small},
  br/.style={vtx,fill=black!80,text=white},
  free/.style={vtx,draw=black!45,fill=black!8,text=black!70},
  side/.style={rounded corners=3mm,draw=black!35,fill=black!3},
  cut/.style={very thick},
  uncut/.style={densely dashed,black!60},
  open/.style={dotted,thick,black!45}]
\foreach \dx/\vx/\pref/\lab in {0/4.15/1/{(a) $v$ on its preferred side},
                                7.2/1.65/0/{(b) $v$ on the other side}}{
  \begin{scope}[xshift=\dx cm]
    \draw[side] (0,0) rectangle (2.2,3);
    \draw[side] (3.6,0) rectangle (5.8,3);
    \node[font=\small,black!60] at (1.1,3.3) {side 0};
    \node[font=\small,black!60] at (4.7,3.3) {side 1};
    \node[vtx] (a) at (0.6,2.35) {$a$};
    \node[vtx] (b) at (0.6,0.65) {$b$};
    \node[vtx] (c) at (5.2,2.35) {$c$};
    \node[br] (v) at (\vx,1.3) {$v$};
    \foreach \i/\x in {1/0.95,2/2.25,3/3.55,4/4.85}{
      \node[free] (u\i) at (\x,-1.25) {$u_{\i}$};
      \draw[open] (v) -- (u\i);
    }
    \ifnum\pref=1
      \draw[cut] (v) -- (a); \draw[cut] (v) -- (b); \draw[uncut] (v) -- (c);
    \else
      \draw[uncut] (v) -- (a); \draw[uncut] (v) -- (b); \draw[cut] (v) -- (c);
      \node[font=\footnotesize] at (2.9,-2.05)
        {at most one of $u_1,\ldots,u_4$ may join side 0};
    \fi
    \node[font=\small] at (2.9,-2.75) {\lab};
  \end{scope}
}
\end{tikzpicture}
\caption{Branching on a vertex $v$ of degree~$7$ in Max-Cut. The neighbors
$a,b,c$ of $v$ are already placed, and $u_1,\ldots,u_4$ are not. Thick
edges are cut, dashed edges are not, and dotted edges are still open.
(a)~On side~1, $v$ cuts two of its three placed edges. This is its
preferred side, and the branch keeps all placements of $u_1,\ldots,u_4$.
(b)~Placing $v$ on side~0 pays off only if at least four of its seven
neighbors end on side~1. Otherwise, moving $v$ to side~1 cuts more edges,
and that cut is examined in branch~(a). Since $a$ and $b$ are on side~0,
at most one of $u_1,\ldots,u_4$ may join them there, and only $5$ of the
$16$ placements of $u_1,\ldots,u_4$ remain.}
\label{fig:maxcut}
\end{figure}

\subsection{Accumulating constraints}\label{sec:accumulate}\label{sec:enumeration}

All four algorithms fix an order of the variables and branch on all
assignments to the first half of them, the \emph{prefix}. At every prefix
variable, the branch of the nonpreferred value receives the comparison.
Along a branch, the comparisons accumulate. Each one alone excludes only a
part of the assignments, as in the example, but together they can exclude
almost all of them. After the prefix, the algorithm adds up the comparisons
of the branch into a single inequality over the remaining variables, lists
exactly the assignments that satisfy it, and scores each of them. We call a
complete assignment a \emph{candidate} if it satisfies the summed
comparisons of its own branch. Lemma~\ref{lem:pruning} justifies each
comparison on its own. For a whole branch, the correctness proofs show that
some optimal assignment satisfies all comparisons on its own branch, hence
also their sum, so it is a candidate. The running time is governed by the
number of candidates. For each input class, the analysis shows that only an
exponentially small fraction of all assignments are candidates, which gives
running time $\Ostar(2^{(1-\gamma)n})$ for a constant $\gamma>0$; we call
$\gamma$ the \emph{saving}. The rest of this section provides the tools for
adding up constraints and listing the solutions of their sum.

\begin{lemma}[Infeasibility by summation]\label{lem:constraint_sum}
Let $\mathcal C$ consist of constraints $C:\sum_{i=1}^{r_C}\ell_i\le b_C$,
with literal occurrences counted separately. Add the constraints, and
repeatedly remove complementary literal pairs from the sum, decreasing
its bound by $1$ per pair. If the resulting bound is negative, no
assignment satisfies all constraints in $\mathcal C$.
\end{lemma}
\begin{proof}
An assignment satisfying every constraint also satisfies their sum.
Cancelling a pair preserves the inequality because its contribution is
exactly one for every assignment. The remaining left-hand side is
nonnegative, so it cannot satisfy a negative bound.
\end{proof}

This is a sufficient infeasibility test. Passing it does not imply that
the collection has a common satisfying assignment. The summed constraint
is temporary; it is not added to $\mathcal C$.

Several algorithms will branch on a fixed set $P$ and then complete the
remaining variables. A \emph{prefix assignment} $\alpha$ specifies all
variables of $P$; a \emph{suffix assignment} $y$ specifies $V\setminus P$.
These names refer to the order that lists $P$ first, not necessarily to
consecutive original labels. Within each set we use the original input
order. The complete assignment is $\alpha\cup y$.

For a collection $\mathcal C$ of pruning constraints, its ordinary
sum is
\begin{equation}
 C_\Sigma:\qquad
 \sum_{C\in\mathcal C}\sum_{i=1}^{r_C}\ell_{C,i}
       \le \sum_{C\in\mathcal C}b_C.
 \label{eq:ordinary_sum}
\end{equation}
Every assignment satisfying the individual constraints satisfies this sum.
The sum can also admit other assignments; we may enumerate and score them
without affecting correctness. An empty collection gives $0\le0$.

After substituting a prefix, collect the occurrences of each unassigned
variable. If its contribution is $aw+b\neg w$, replace it by
$\min(a,b)+|a-b|z$, where $z=w$ when $a\ge b$ and $z=\neg w$ otherwise.
Move the constant to the right-hand side. Thus the sum has the form
\begin{equation}
 \sum_{i=1}^k w_i z_i\le B,\qquad w_i\ge0,\quad z_i\in\bits,
 \label{eq:budget}
\end{equation}
with one oriented literal $z_i$ for each remaining variable. Include
variables of coefficient $0$, so every enumerated assignment is complete.
For occurrence counts, this is precisely complementary cancellation from
Lemma~\ref{lem:constraint_sum}. The same operation works with nonnegative
rational coefficients, which the later sections will use.

For this single inequality, the minimum left-hand side is $0$: set all
remaining $z_i$ to $0$. Hence a negative bound is an exact infeasibility
test for the sum. This does not decide joint feasibility of the individual
constraints. The distinction matters for the following enumeration bound:
every retained partial assignment has a completion satisfying the sum.
In the lemma, a node of the depth-first search is \emph{active} while its
subtree is being explored, and the \emph{prefix tree} is the binary tree of
all assignments to the branching variables, in the branching order.

\begin{lemma}[Enumerating the summed constraint]\label{lem:enumeration}
Let $\sum_{i=1}^k w_i z_i\le B$ have nonnegative rational weights and $N$
satisfying Boolean assignments. Depth-first enumeration in a fixed order,
rejecting a partial assignment as soon as its assigned contribution exceeds
$B$, lists each satisfying assignment once and visits at most $1+2kN$
recursion nodes, including rejected nodes. Its recursion depth is at most
$k+1$, and its time and space per active node are polynomial in the encoding
length of the inequality and the assignment.

Consequently, suppose branching assigns a fixed set of $T$ variables,
and $N(\alpha)$ counts the suffix assignments satisfying the sum for
prefix $\alpha$. If constructing each summed inequality takes polynomial
time and its encoding length is polynomial in the input size, the whole
search uses polynomial space and time
\begin{equation}
 \Ostar(2^T+W(F)),\qquad W(F)=\sum_{\alpha\in\bits^T}N(\alpha).
 \label{eq:enumerationcost}
\end{equation}
More precisely, the full prefix tree followed by suffix enumeration visits
at most $2^{T+1}-1+2kW(F)$ nodes when $k$ variables remain after each prefix.
\end{lemma}
\begin{proof}
If $B<0$, the root is rejected and is the only visited node. Otherwise,
every retained partial assignment extends to a solution by setting all
unassigned $z_i$ to $0$. At a fixed depth, distinct partial assignments
have disjoint sets of completions. There are therefore at most $N$
retained nodes at each depth. There are $k$ depths at which children can
be created. Apart from the root, every visited node is one of the at most
two children of a retained node. This proves the bound $1+2kN$, including
children that are immediately rejected. Each complete assignment follows
one path, so no output is repeated.

Each recursive call assigns one variable. The depth is therefore at most
$k+1$, and the algorithm stores the active path rather than all outputs.
Subtraction and comparison of the rational coefficients have polynomial
bit cost in the input encoding length.

The full prefix tree has $2^T-1$ internal nodes; its leaves are the roots
of the suffix enumerations. Summing their bounds gives
\[
 (2^T-1)+\sum_\alpha(1+2kN(\alpha))
 =2^{T+1}-1+2kW(F).
\]
Constructing and processing the sums adds polynomial work per prefix and
per suffix node. Process prefixes and outputs consecutively to obtain
the stated space bound.
\end{proof}

\section{Formulas with many biased variables}\label{sec:bias}

\subsection{Intuition}\label{sec:bias_intuition}

\paragraph{One biased variable.}
Consider a variable $\mathbf v$ that occurs in $D$ binary clauses, in
$3D/4$ of them as $v$ and in $D/4$ as $\neg v$, and in no unit clause. Its preferred literal is $v$, and the branch $v=0$ receives the
comparison $C(v)$ of Lemma~\ref{lem:pruning}. It allows at most $D/4-1$ of
the $D$ neighbor literals to favor $v=1$, that is, to be a false neighbor
of $v$ or a true neighbor of $\neg v$. If the neighbors took uniformly
random values, about $D/2$ of these literals would favor $v=1$. The
comparison therefore asks for a deviation of $D/4+1$ below the typical
value, a constant fraction of $D$. Only a fraction exponentially small in $D$
of the assignments to the neighbors satisfies it. We call a variable
\emph{biased} if, as here, the occurrence counts of its two literals differ
by a constant fraction of its degree.

\paragraph{Many biased variables.}
A single comparison restricts only the neighbors of its variable. A saving
on all $n$ variables must come from many comparisons together. The
algorithm branches on about half of the variables. On each branch, every
branching variable with its nonpreferred value contributes its comparison,
and after the branching the algorithm lists only the completions that
satisfy the sum of these comparisons. To count the listed assignments,
follow a uniformly random complete assignment through the branching. It
takes the nonpreferred value at about half of the branching variables. If
many of them are biased, the requirements of their comparisons add up to a
large total, which the random assignment would have to compensate.

\paragraph{Overlapping neighborhoods.}
The comparisons are not independent. Biased variables may share many
neighbors, and one could fear that a good assignment to a few shared
variables satisfies many comparisons at once. This cannot happen as long as
no variable occurs much more often than the others. A binary clause on
$\mathbf v$ and $\mathbf w$ contributes one literal of $\mathbf w$ to the
comparison of $\mathbf v$ and one literal of $\mathbf v$ to the comparison
of $\mathbf w$. Hence $\mathbf w$ occurs in all comparisons together at most
as often as in the formula, however the neighborhoods overlap, and changing
the value of $\mathbf w$ moves the sum of the comparisons by at most about
twice its number of occurrences. When these numbers are comparable, a
bounded-differences inequality shows that the random assignment satisfies
the summed comparisons only with exponentially small probability.

\paragraph{Fixed preferences.}
The algorithm computes all preferred values once, in the formula after
preprocessing, before any variable is assigned, and it branches on a set
of about $n/2$ variables that carries at least half of the total gap.
Fixed preferences make each comparison a fixed function of the assignment,
so the analysis needs a single concentration inequality.
Section~\ref{sec:random} handles preferences that change along the branch.

\paragraph{Formal setup.}
Let $F_0$ be the formula after preprocessing, let $D_{\mathbf v}$ be the
binary degree of a variable $\mathbf v$ in $F_0$, and let
$g_{\mathbf v}=g_{\mathbf v}(F_0)$ be its effective
gap~\eqref{eq:effectivegap}. We call $\mathbf v$
\emph{$\varepsilon$-biased} if $|g_{\mathbf v}|\ge\varepsilon D_{\mathbf v}$;
the binary degree need not be linear in $n$. The comparison of $\mathbf v$
has $D_{\mathbf v}$ literals, and by Section~\ref{sec:preference} its
requirement is $|g_{\mathbf v}|/2+1$, at least a fraction $\varepsilon/2$ of
its length. In the example above, $g_{\mathbf v}=D/2$, so $\mathbf v$ is
$1/2$-biased.

\subsection{The algorithm}\label{sec:bias_algorithm}

The literal-gap algorithm, \textnormal{\textsc{LiteralGap}}, follows the
common scheme of Section~\ref{sec:framework}.

\paragraph{Setup.}
Apply \textnormal{\textsc{Preprocess}} to the input, obtaining $(F_0,s_0)$; a
complete assignment $x$ then has original score $s_0+F_0(x)$. The formula
$F_0$ and its variable set stay fixed during the search. Compute the gaps
$g_{\mathbf v}=g_{\mathbf v}(F_0)$ by~\eqref{eq:effectivegap} and the
preferred values $p_{\mathbf v}$, which are $1$ exactly when
$g_{\mathbf v}\ge0$. Let $A$ contain the first
$\lfloor n/2\rfloor$ variables in the input order, let $A'=V\setminus A$, and
branch on the half with the larger total gap:
\begin{equation}
 P=\begin{cases}
 A',&\displaystyle\sum_{\mathbf v\in A}|g_{\mathbf v}|
          \le\sum_{\mathbf v\in A'}|g_{\mathbf v}|,\\
 A,&\text{otherwise}.
 \end{cases}
 \label{eq:bias_prefix}
\end{equation}
This deterministic choice gives
\begin{equation}
 |P|\le\lceil n/2\rceil,\qquad
 \sum_{\mathbf v\in P}|g_{\mathbf v}|
       \ge\frac12\sum_{\mathbf v\in V}|g_{\mathbf v}|.
 \label{eq:bias_half}
\end{equation}
For each $\mathbf v\in P$, let $C_{\mathbf v}=C(\ell_{\mathbf v}(F_0))$ be the
comparison of $\mathbf v$: the constraint of Lemma~\ref{lem:pruning} for
its preferred literal, constructed in $F_0$.

\paragraph{Search.}
Enumerate the assignments $\alpha$ to $P$ by depth-first branching in the
input order. For each prefix, collect the constraints of the variables that
take their nonpreferred values,
\[
 \mathcal C(\alpha)=
 \{C_{\mathbf v}:\mathbf v\in P,\ \alpha_{\mathbf v}\ne p_{\mathbf v}\},
\]
form their ordinary sum~\eqref{eq:ordinary_sum}, substitute $\alpha$, and
cancel complementary occurrences. If the remaining bound is negative, reject
the prefix. Otherwise enumerate all suffix assignments satisfying the sum by
Lemma~\ref{lem:enumeration}, score each complete assignment, and keep the
best one over all prefixes.

\begin{algorithm}[tp]
\caption{\textnormal{\textsc{LiteralGap}}$(F)$}
\label{alg:literal_gap}\label{alg:bias}
\begin{algorithmic}[1]
\STATE $(F_0,s_0)\gets\textnormal{\textsc{Preprocess}}(F)$.
\STATE For every $\mathbf v$, compute $g_{\mathbf v}(F_0)$; set $p_{\mathbf v}\gets1$ if $g_{\mathbf v}(F_0)\ge0$ and $p_{\mathbf v}\gets0$ otherwise.
\STATE Choose the branching set $P$ by~\eqref{eq:bias_prefix}.
\STATE For each $\mathbf v\in P$, construct $C_{\mathbf v}=C(\ell_{\mathbf v}(F_0))$ in $F_0$ by Lemma~\ref{lem:pruning}.
\STATE $M\gets-\infty$; $\alpha^\star\gets\emptyset$.
\FOR{each assignment $\alpha$ to $P$, by depth-first branching in the input order}
  \STATE $\mathcal C\gets\{C_{\mathbf v}:\mathbf v\in P,\ \alpha_{\mathbf v}\ne p_{\mathbf v}\}$.
  \STATE Form $C_\Sigma$ by~\eqref{eq:ordinary_sum}; substitute $\alpha$ and cancel complementary occurrences.
  \COMMENT{The remaining inequality concerns exactly the suffix variables.}
  \IF{the remaining bound is nonnegative}
    \FOR{each suffix $y$ satisfying $C_\Sigma$, by Lemma~\ref{lem:enumeration}}
      \STATE $q\gets s_0+F_0(\alpha\cup y)$.
      \STATE If $q>M$, set $M\gets q$ and $\alpha^\star\gets\alpha\cup y$.
    \ENDFOR
  \ENDIF
\ENDFOR
\STATE \textbf{return} $(M,\alpha^\star)$.
\end{algorithmic}
\end{algorithm}

The algorithm receives only $F$; it needs none of the parameters of
Theorem~\ref{thm:bias} below.

\subsection{Correctness}

\begin{lemma}[Correctness of the literal gap algorithm]\label{lem:gap_correctness}
For every simple input formula, \textnormal{\textsc{LiteralGap}} returns
its optimum score and an assignment attaining that score.
\end{lemma}
\begin{proof}
Preprocessing subtracts the same constant $s_0$ from every assignment's
score, so it preserves the optimal assignments. Among the optimal
assignments for $F_0$, choose $x^*$ with the largest number of variables
set to their fixed preferred values. Such an assignment exists because
there are finitely many complete assignments.

Suppose $x^*_{\mathbf v}\ne p_{\mathbf v}$. Changing this one variable
to $p_{\mathbf v}$ cannot increase the score, by optimality. It cannot
leave the score unchanged either: it would then give another optimum
with one more preferred value. Its score is therefore strictly smaller.
By Lemma~\ref{lem:pruning}, $x^*$ satisfies $C_{\mathbf v}$.

In particular, $x^*$ satisfies every constraint in
$\mathcal C(x^*|_P)$, and hence their sum. Substitution and complementary
cancellation preserve this property. Its prefix is not rejected and its
suffix is enumerated. The algorithm evaluates that complete assignment
with the offset $s_0$, so it records the optimum original score. Every
other score considered also belongs to a complete assignment of $F$ and
cannot exceed the optimum.
\end{proof}

\subsection{Running-time guarantee}

The theorem requires enough total initial gap and a bound on how strongly
a few variables can influence the comparisons. Both conditions concern
the formula immediately after preprocessing. The average degree may grow
with $n$; it does not enter the saving in the exponent.

\begin{theorem}[Many biased variables]\label{thm:bias}
Fix constants $\varepsilon,\eta\in(0,1]$ and $R\ge1$.
Let $F$ be a simple unweighted Max 2-SAT formula on a set $V$ of $n$ variables,
and let $(F_0,s_0)=\textnormal{\textsc{Preprocess}}(F)$. For each variable
$\mathbf v$, define
\[
 D_{\mathbf v}=d_2^{F_0}(v)+d_2^{F_0}(\neg v),\qquad
 u_{\mathbf v}=d_1^{F_0}(v)+d_1^{F_0}(\neg v),
\]
\[
 g_{\mathbf v}=d_2^{F_0}(v)-d_2^{F_0}(\neg v)
                    +2\bigl(d_1^{F_0}(v)-d_1^{F_0}(\neg v)\bigr).
\]
Suppose there is a set $S\subseteq V$ with $|S|\ge\eta n$ such that
$|g_{\mathbf v}|\ge\varepsilon D_{\mathbf v}$ for every $\mathbf v\in S$.
Writing $d=|S|^{-1}\sum_{\mathbf v\in S}D_{\mathbf v}$, suppose also that
\begin{equation}
 d\ge1,\qquad
 \sum_{\mathbf v\in V}(D_{\mathbf v}+u_{\mathbf v})^2\le Rnd^2.
 \label{eq:biasregularity}
\end{equation}
Then \textnormal{\textsc{LiteralGap}}$(F)$ returns an optimal assignment for
$F$ in polynomial space and time $\Ostar(2^{(1-\gamma)n})$, where
$\gamma=\varepsilon^2\eta^2/(3000R)$.
\end{theorem}

The first hypothesis asks for at least $\eta n$ variables that are
$\varepsilon$-biased in the sense of Section~\ref{sec:bias_intuition}. The
second hypothesis~\eqref{eq:biasregularity} asks that the occurrence counts
are not too uneven. It allows unequal degrees and overlapping
neighborhoods, and it holds with $R=c^2$ whenever every variable occurs in
at most $cd$ clauses of $F_0$. The proof gives the explicit bound
$(4+2n)\exp((\log2-\varepsilon^2\eta^2/(512R))n)$ on the number of search
nodes; the theorem states a smaller $\gamma$ in a simpler form.

\subsection{Proof of the running-time bound}

By Lemma~\ref{lem:enumeration}, it suffices to bound the number $W(F_0)$
of candidates, and $W(F_0)/2^n$ is the probability that a uniformly random
complete assignment is a candidate. The algorithm itself is deterministic;
the random assignment only serves to count. Throughout, scores refer to the
preprocessed formula $F_0$.

\paragraph{The statistic.}
Let $x[\mathbf v\gets p_{\mathbf v}]$ denote the assignment obtained from
$x$ by changing only $\mathbf v$ to its preferred value. Define
\begin{equation}
 H_{\mathbf v}(x)=F_0(x[\mathbf v\gets p_{\mathbf v}])-F_0(x),
 \qquad L_P(x)=\sum_{\mathbf v\in P}H_{\mathbf v}(x).
 \label{eq:bias_statistic}
\end{equation}
Thus $H_{\mathbf v}(x)$ is the change in score when $\mathbf v$ is switched
to its preferred value; it is $0$ when $x$ already takes that value. If
$x_{\mathbf v}\ne p_{\mathbf v}$, the comparison $C_{\mathbf v}$ is the
inequality $H_{\mathbf v}(x)\le-1$, written over literals, by
Lemma~\ref{lem:pruning}. We write
$B_{\mathbf v}=D_{\mathbf v}+u_{\mathbf v}$ for the total number of
occurrences of $\mathbf v$ in $F_0$. The set $P$ and all preferences are
fixed before $x$ is drawn.

\begin{lemma}[The combined comparison statistic]\label{lem:bias_statistics}
For any fixed set $P$ and preferences chosen from the gaps of $F_0$,
\[
 \E L_P(x)=\frac14\sum_{\mathbf v\in P}|g_{\mathbf v}|.
\]
Changing only variable $\mathbf w$ changes $L_P$ by at most
$B_{\mathbf w}+D_{\mathbf w}\le2B_{\mathbf w}$.
If $x$ satisfies the summed constraint associated with its prefix $x|_P$,
then $L_P(x)\le0$.
\end{lemma}

The three parts have simple reasons. A variable takes its nonpreferred
value with probability $1/2$, and switching it to the preferred value
then gains $|g_{\mathbf v}|/2$ on average; this gives the expectation. A binary
clause on $\mathbf v$ and $\mathbf w$ enters only the comparisons of
$\mathbf v$ and $\mathbf w$. Changing $\mathbf w$ therefore affects $L_P$
only through the occurrences of $\mathbf w$, however much the
neighborhoods of different variables overlap. Finally, the summed
constraint of the prefix of $x$ is equivalent to $L_P(x)+|I_P(x)|\le0$,
where $I_P(x)$ is the set of prefix variables with nonpreferred values.
Appendix~\ref{app:bias_statistics_proof} gives the details.

The lemma reduces the count of candidates to a concentration estimate:
$L_P$ has a large mean, no single variable can change it much, and every
candidate has $L_P\le0$.

\begin{lemma}[Counting candidates from the initial gaps]\label{lem:bias_count}
For the preprocessed formula $F_0$, put
\[
 G=\sum_{\mathbf v\in V}|g_{\mathbf v}|,\qquad
 Q=\sum_{\mathbf v\in V}B_{\mathbf v}^2.
\]
Suppose $Q>0$, and choose $P$ by~\eqref{eq:bias_prefix}. Let $W(F_0)$ be
the total number of complete assignments satisfying their prefix's
summed constraint. Then
\begin{equation}
 W(F_0)\le2^n\exp\left(-\frac{G^2}{512Q}\right).
 \label{eq:bias_count}
\end{equation}
\end{lemma}
\begin{proof}
Let $x$ be a uniformly random complete assignment. Every candidate
satisfies $L_P(x)\le0$ by Lemma~\ref{lem:bias_statistics}, so
$W(F_0)\le2^n\Prb(L_P(x)\le0)$. By~\eqref{eq:bias_half} and
Lemma~\ref{lem:bias_statistics}, the mean $\mu=\E L_P(x)$ satisfies
$\mu\ge G/8$, and changing coordinate $\mathbf w$ changes $L_P$ by at most
$c_{\mathbf w}=2B_{\mathbf w}$, so $\sum_{\mathbf w}c_{\mathbf w}^2=4Q>0$.
The bounded-differences inequality of
Lemma~\ref{lem:bias_concentration}, applied with deviation $a=\mu$, gives
\[
 \Prb(L_P(x)\le0)\le\exp\left(-\frac{\mu^2}{8Q}\right)
 \le\exp\left(-\frac{G^2}{512Q}\right).\qedhere
\]
\end{proof}

\begin{proof}[Proof of Theorem~\ref{thm:bias}]
The algorithm is correct by Lemma~\ref{lem:gap_correctness}; we bound its
running time. The empty variable set is trivial, so assume $n\ge1$, and
put $\theta=\varepsilon\eta$. The hypotheses give
\[
 G\ge\sum_{\mathbf v\in S}|g_{\mathbf v}|
   \ge\varepsilon\sum_{\mathbf v\in S}D_{\mathbf v}
   =\varepsilon|S|d\ge\theta nd,
 \qquad
 Q=\sum_{\mathbf v\in V}B_{\mathbf v}^2\le Rnd^2,
\]
and $Q>0$ because $d\ge1$. Lemma~\ref{lem:bias_count} therefore gives
$W(F_0)\le2^n e^{-\theta^2n/(512R)}$. Since $|P|\le\lceil n/2\rceil$,
Lemma~\ref{lem:enumeration} bounds the number of search nodes, including
rejected ones, by
\[
 2^{|P|+1}-1+2nW(F_0)
 \le(4+2n)\exp\bigl((\log2-\theta^2/(512R))n\bigr),
\]
where we used $\theta^2/(512R)\le1/512<(\log2)/2$. Finally,
$\gamma=\theta^2/(3000R)$ satisfies $\gamma\log2\le\theta^2/(512R)$, so
the bound is at most $(4+2n)2^{(1-\gamma)n}$. Each search node takes
polynomial time, and the search uses polynomial space; see
Appendix~\ref{app:bias_implementation}.
\end{proof}

\subsection{Linear gaps in simple formulas}

In a simple formula, every occurrence count is $O(n)$. Linearly many
variables with gaps linear in $n$ therefore suffice, without a condition on
the degrees.

\begin{corollary}[Linear effective gaps in a simple formula]\label{cor:lineargap}
Fix $\varepsilon,\eta\in(0,1]$. Let $F$ be a simple unweighted Max 2-SAT
formula on $n$ variables. Suppose at least $\eta n$ variables satisfy
\[
 |d_2^F(v)-d_2^F(\neg v)+2(d_1^F(v)-d_1^F(\neg v))|
 \ge\varepsilon n.
\]
Then \textnormal{\textsc{LiteralGap}}$(F)$ returns an optimal assignment for
$F$ in polynomial space and time
$\Ostar(2^{(1-\varepsilon^2\eta^2/48000)n})$.
\end{corollary}

The proof applies Lemma~\ref{lem:bias_count} directly. In a simple
formula, every variable has at most $4(n-1)$ binary and two unit
occurrences, so $Q\le16n^3$, while the $\eta n$ variables with linear
gaps give $G\ge\varepsilon\eta n^2$. Preprocessing keeps both bounds,
because it preserves effective gaps and never increases occurrence
counts. In particular, a variable keeps its gap even if preprocessing
turns some of its binary clauses into unit clauses.
Appendix~\ref{app:lineargap_proof} gives the details.

Counting unit clauses once instead of twice changes the gap of a simple
formula by at most $1$, so the corollary also holds for that count with a
slightly smaller $\varepsilon$.

\section{Semirandom formulas}\label{sec:random}

\subsection{Intuition}\label{sec:random_intuition}

In a \emph{semirandom} formula~\cite{feige2007refuting}, a set $E$ of pairs
of variables and the unit clauses are fixed arbitrarily, and the binary
clauses, all on pairs in $E$, are random. We use two models. In the
\emph{first model}, each pair in $E$ receives one of its four possible
clauses, chosen uniformly and independently. In the \emph{second model},
each of the four clauses on each pair in $E$ is included independently
with probability $1/4$. In both models the signs of the literals are
random.

\paragraph{Why Section~\ref{sec:bias} does not apply.}
With random signs, a variable is typically not biased. A variable with $d$
binary clauses occurs positively in about $d/2$ of them, with a deviation
of order $\sqrt d$, so its gap is of order $\sqrt d$ rather than a constant
fraction of $d$. For instance, a variable with $100$ binary clauses may
occur $55$ times as $v$ and $45$ times as $\neg v$. Its comparison $C(v)$
then allows at most $44$ of its $100$ literals to favor $v=1$, while about
$50$ do so on average, with standard deviation $5$. About $14\%$ of the
assignments to the neighbors still satisfy it. A single comparison now
excludes only a constant fraction of these assignments, however large $d$
is, and a fixed saving must come from many such weak comparisons together.

\paragraph{Weighting comparisons of different lengths.}
When weak comparisons are added up, their lengths matter. Consider the two
constraints
\[
 \sum_{i=1}^{100}z_i\le40,\qquad
 \sum_{i=1}^{400}z_i\le180,
\]
with distinct Boolean variables in each sum. Their requirements are $10$
and $20$, and the standard deviations of their left-hand sides under
uniform values are $5$ and $10$. Each asks for a deviation of two standard
deviations, so the two constraints are about equally restrictive. Adding
them unscaled would nevertheless give the second twice the weight of the
first. Dividing its two sides by $2$ puts both on the same scale without
changing its satisfying assignments. The algorithm therefore multiplies
the comparison of a variable with $d$ neighbors by
$1/\lceil\sqrt d\rceil$ before adding it, since with random signs both its
requirement and the standard deviation of its left-hand side are typically
of order $\sqrt d$. The weights are fixed from the input, and they are all
equal when all variables have the same number of neighbors.

\paragraph{Overlapping neighborhoods.}
In Section~\ref{sec:bias}, a bound on the occurrence counts kept a few
shared variables from satisfying many comparisons at once. Here the degrees
are arbitrary, and the random signs take over this role. A shared variable
occurs positively in some comparisons and negatively in others, and
opposite occurrences cancel in the sum: for sums $S,S'$ of other literals,
the constraints $y+S\le b$ and $\neg y+S'\le c$ add up to $S+S'\le b+c-1$, so
$y$ no longer influences the sum. This holds even when two comparisons
contain exactly the same variables.

\paragraph{Preferences along the branch.}
As in Section~\ref{sec:bias}, the algorithm branches on half of the
variables, here the first half in the input order, and records the
comparison of every branching variable that takes its nonpreferred value.
Unlike there, it computes the preferred value of the next variable in the
current residual formula, from the clause counts that remain after the
earlier branching variables have been assigned. The reason lies in the
analysis. Follow a uniformly random complete assignment. When the search
reaches a variable, its requirement is measured against its expected gain
given the values already fixed. The later values are still uniform, so
they move the gain away from this expectation by a random amount with mean
$0$. The total compensation supplied by the later values is therefore a sum
of terms with conditional mean $0$, and a martingale inequality bounds it.
The proof reveals the random signs at a variable only when the search
reaches it, so it never needs the comparisons to be independent.

\paragraph{The resulting bound.}
On the branch of a uniformly random assignment, about half of the
branching variables take their nonpreferred values. Except with
exponentially small probability, their weighted requirements add up to a
constant times the number of branching steps, while the later values
compensate only a smaller amount. The random assignment then violates the
summed comparisons. This holds for all but an exponentially small fraction
of formula--assignment pairs, which bounds the expected number of
candidates.

\subsection{The algorithm}\label{sec:random_algorithm}\label{sec:random_notation}

The algorithm for semirandom formulas, \textnormal{\textsc{RandomFormula}},
follows the common scheme of Section~\ref{sec:framework}. Unlike
\textnormal{\textsc{LiteralGap}}, it branches on the first half of its
variables in the input order, computes preferred values in the current
residual formula, and weights the comparisons before adding them.

\paragraph{Weighted sums.}
Give each constraint
$C:\sum_i\ell_i\le b_C$ a positive rational weight $\omega_C$, fixed when
it is created, multiply both sides by it, and add:
\begin{equation}
 C_\Sigma:\qquad
 \sum_{C\in\mathcal C}\omega_C\sum_i\ell_i
 \le \sum_{C\in\mathcal C}\omega_C b_C .
\label{eq:weighted}
\end{equation}
Every assignment satisfying all individual constraints satisfies this sum,
and with all weights equal to $1$ it is the ordinary
sum~\eqref{eq:ordinary_sum}. Collecting complementary occurrences as in
Section~\ref{sec:enumeration} brings it to the form~\eqref{eq:budget}, so
Lemma~\ref{lem:enumeration} enumerates precisely the completions that
satisfy it.

\paragraph{Setup.}
The algorithm first constructs the \emph{observed graph} $G_F=(V_F,E_F)$ of
the binary clauses, before preprocessing:
\begin{equation}
 E_F=\bigl\{\{u,v\}:F\text{ contains a binary clause on }u,v\bigr\},
\label{eq:observed_support}
\end{equation}
and $V_F$ is the set of endpoints of these edges. A
variable outside $V_F$ occurs only in unit clauses. It is set to the value
that satisfies more of its units, with ties broken towards $0$, and removed
from the search. We write $\alpha_0$ for these assignments, $s_0$ for the
number of clauses they satisfy, and $F_{\rm core}$ for the clauses on the
remaining variables. If $V_F$ is empty, this solves the instance. Otherwise,
let $\mathbf u_1,\ldots,\mathbf u_r$ list the variables of $V_F$ in their
input order, write $x_t$ for the value of $\mathbf u_t$, put
$T=\lceil r/2\rceil$, and record the weights
\begin{equation}
 \omega_t=\frac{1}{\lceil\sqrt{\deg_{G_F}(\mathbf u_t)}\rceil}.
 \label{eq:random_weights}
\end{equation}
Finally, $F_{\rm core}$ is preprocessed.

\paragraph{Search.}
The search branches on $\mathbf u_1,\ldots,\mathbf u_T$ in this order, depth
first. Before branching on $\mathbf u_t$, it computes the preferred literal
of $\mathbf u_t$ in the current residual formula
by~\eqref{eq:bias_preference} and constructs the constraint of
Lemma~\ref{lem:pruning} for it. This constraint, with weight $\omega_t$, is
added only in the child that takes the nonpreferred value. At every node the
search forms the weighted sum~\eqref{eq:weighted} of the constraints created
so far, brings it to the form~\eqref{eq:budget}, and rejects the node if the
bound is negative. At depth $T$ it enumerates the assignments to
$\mathbf u_{T+1},\ldots,\mathbf u_r$ that satisfy the weighted sum, using
Lemma~\ref{lem:enumeration}. It scores each completion exactly and keeps the
best one. Finally, the best assignment is extended by $\alpha_0$, and $s_0$
is added to its score. Appendix~\ref{app:random_implementation} specifies
how the procedures keep track of scores, constraints, and the best
solution.

\begin{algorithm}[tp]
\caption{\textnormal{\textsc{RandomFormula}}$(F)$}\label{alg:weighted}
\begin{algorithmic}[1]
\STATE Construct $G_F=(V_F,E_F)$ from $F$ by~\eqref{eq:observed_support}, before preprocessing.
\STATE Optimize variables outside $V_F$, obtaining $F_{\rm core},s_0,\alpha_0$ as above.
\IF{$V_F=\emptyset$}
  \STATE \textbf{return} $(s_0,\alpha_0)$.
\ENDIF
\STATE Initialize shared $U=(\mathbf u_1,\ldots,\mathbf u_r)$ in the original input order, where $r=|V_F|$.
\STATE $T\gets\lceil r/2\rceil$.
\COMMENT{Branch on the first $T$ variables, then enumerate completions.}
\STATE Initialize the shared weights $\omega_t$ from $G_F$ by~\eqref{eq:random_weights}.
\STATE Initialize shared $M\gets-\infty$ and $\alpha^\star\gets\emptyset$.
\STATE $(F,s)\gets\textnormal{\textsc{Preprocess}}(F_{\rm core})$.
\STATE $\textnormal{\textsc{RandomSearch}}(F,\emptyset,\emptyset,s,T)$.
\STATE \textbf{return} $(M+s_0,\alpha^\star\cup\alpha_0)$.
\end{algorithmic}
\end{algorithm}

\begin{algorithm}[tp]
\caption{\textnormal{\textsc{RandomSearch}}$(F,\mathcal C,\alpha,s,T)$}
\label{alg:random_search}
\begin{algorithmic}[1]
\REQUIRE Fixed shared data $U,(\omega_t)_{t=1}^r$ and shared best solution $M,\alpha^\star$, initialized above.
\STATE $t\gets|\alpha|$.
\COMMENT{The variables $\mathbf u_1,\ldots,\mathbf u_t$ have been assigned.}
\STATE Form $C_\Sigma$ by~\eqref{eq:weighted} and simplify it to~\eqref{eq:budget}.
\IF{its bound $B$ is negative}
  \STATE \textbf{return}
\ENDIF
\IF{$t=T$}
  \FOR{each suffix assignment $y$ satisfying $C_\Sigma$, by Lemma~\ref{lem:enumeration}}
    \STATE Let $q$ be $s$ plus the number of clauses of $F$ satisfied by $y$.
    \STATE If $q>M$, update shared $M\gets q$ and $\alpha^\star\gets\alpha\cup y$.
  \ENDFOR
  \STATE \textbf{return}
\ENDIF
\STATE Let $\mathbf u\gets\mathbf u_{t+1}$.
\COMMENT{The next variable in the fixed list $U$.}
\STATE $\ell\gets\ell_{\mathbf u}(F)$, the preferred literal in the current $F$, by~\eqref{eq:bias_preference}.
\STATE Construct $C$ from $F$ and $\ell$ by Lemma~\ref{lem:pruning}; give it weight $\omega_{t+1}$.
\FOR{$b=1,0$}
  \STATE Make separate copies $(F_b,\mathcal C_b,\alpha_b,s_b)$ of the current state.
  \IF{$b=0$}
    \STATE $\mathcal C_b\gets\mathcal C_b\cup\{C\}$.
  \ENDIF
  \STATE Assign $\ell=b$, updating the copied formula, constraints, assignment, and score.
  \STATE $\textnormal{\textsc{RandomSearch}}(F_b,\mathcal C_b,\alpha_b,s_b,T)$.
\ENDFOR
\end{algorithmic}
\end{algorithm}

\subsection{Correctness}

Since the preferred values now depend on the path, the proof breaks ties
between optima by their position in the search tree. It does not use
randomness.

\begin{lemma}[Correctness of weighted enumeration]\label{lem:exactA}
For every formula $F$, including formulas with repeated
clauses, \textnormal{\textsc{RandomFormula}} returns an optimum score and
an assignment attaining it.
\end{lemma}
\begin{proof}
Variables outside $V_F$ can be optimized independently from their units.
It therefore suffices to prove that the recursive search finds an optimum
for $F_{\rm core}$. Preprocessing changes the score of every assignment
by the same recorded constant. It therefore preserves both optimal assignments and all
comparisons between the two values of a variable.

Consider the full decision tree in the fixed variable order, restoring
any prefixes rejected by the algorithm. At each of the first $T$ levels,
place the preferred child first. Complete the rest of the tree in any
fixed order, and choose the earliest optimal complete assignment $x^*$.
If its path takes a nonpreferred child at variable $t$, changing only
that value to the preferred one produces an earlier assignment. Its
score must be strictly smaller: a larger score contradicts optimality,
and an equal score contradicts the choice of $x^*$.
Lemma~\ref{lem:pruning} therefore implies that $x^*$ satisfies the
constraint added at that step.

Thus $x^*$ satisfies every constraint created on its prefix, and hence
their positively weighted sum at every earlier depth. The sum test never
rejects that prefix, and enumeration at depth $T$ visits its suffix.
Every enumerated candidate is scored in $F_{\rm core}$, including
the preprocessing and assignment offsets. Adding the separately optimized
score $s_0$ and extending by $\alpha_0$ therefore returns an optimum for $F$.
\end{proof}

\subsection{Semirandom inputs and the running-time guarantee}\label{sec:random_inputs}

The following theorem states the two models of
Section~\ref{sec:random_intuition} formally. The set $E$ is arbitrary and
is not given to the algorithm. In the second model, a pair in $E$ receives
no clause with probability $(3/4)^4$, so the observed graph $G_F$ may be a
proper subgraph of $(V,E)$.

\begin{theorem}[Semirandom formulas]\label{thm:random_main}
There is an absolute constant $c_0\in(0,(\ln2)/2)$ with the following
property. Fix $n$ variables and any set $E$ of unordered variable pairs.
Generate the binary clauses by either of the following two models:
\begin{enumerate}
\item On each pair in $E$, choose exactly one of the four possible
clauses uniformly, independently across pairs.
\item Include each of the four possible clauses on each pair in $E$
independently with probability $1/4$, with all choices independent.
\end{enumerate}
Pairs outside $E$ have no binary clauses.
Additional unit clauses may be any multiset fixed independently of
these draws.

Then \textnormal{\textsc{RandomFormula}}, applied to the resulting formula
alone, uses polynomial space and has expected running time
$\Ostar(2^n e^{-c_0n})$. With probability at least $1-3e^{-c_0n/2}$,
its running time is $\Ostar(2^n e^{-c_0n/2})$.
It returns an optimum on every input.
\end{theorem}

Since $c_0$ does not depend on $E$, the bound holds at every clause
density. By averaging, it also holds when $E$ is random and drawn
independently of the clauses, for example a block-model graph as in
Section~\ref{sec:sbm}. Max-Cut on such a graph is a different problem: its
two clauses on each edge are fixed, not random, and Section~\ref{sec:sbm}
treats it.

\subsection{Proof of the running-time bound}

We first condition on the observed graph $G_F$. By
Lemma~\ref{lem:observed_pairs}, the clauses on its edges keep the
independence and sign symmetry that the proof uses, so the graph can be
treated as fixed although the algorithm constructs it from $F$. The proof
then has three steps. Step~1 writes the summed constraint of a uniformly
random assignment as the requirements created along its path plus the
compensation supplied by the later values. Step~2 shows that the
requirements are large and the compensation is small, except with
exponentially small probability. Step~3 turns this into a time bound.
Appendix~\ref{app:random_signs} proves the supporting lemmas.

\begin{proof}[Proof of Theorem~\ref{thm:random_main}]
Both models satisfy the assumptions of Definition~\ref{def:invariant}
in Appendix~\ref{app:random_signs} with $m_0=4$ and $\kappa=3/4$: a pair
carries at most four clauses, and the signed occurrence count at an
endpoint has second moment at least $3/4$. We use the constants
$c_1,c_0>0$ of Theorem~\ref{thm:weighted} for these values;
the end of Appendix~\ref{app:random_signs} checks the assumptions and
shows that $c_0<(\ln2)/2$.

Use \textnormal{\textsc{RandomFormula}} on $F$ alone.
Correctness is Lemma~\ref{lem:exactA}. Depth-first
search and suffix enumeration use polynomial space. Constructing $G_F$
and optimizing variables outside $V_F$ take polynomial time.

Condition on $G_F=H$, where $H$ is any possible nonempty observed graph,
and write $r=|V(H)|$. By Lemma~\ref{lem:observed_pairs}, its pair laws
satisfy the assumptions of Definition~\ref{def:invariant} in Appendix~\ref{app:random_signs},
with the same $m_0,\kappa$. Its vertex order, inherited from the input indices, and its
degrees are now fixed. All probabilities in Steps~1 and~2 refer to this
conditional model. The weights are rational with polynomial encoding
length, and satisfy
\begin{equation}
 \frac{1}{2\sqrt{\deg_H(\mathbf u_t)}}\le\omega_t
 \le\frac{1}{\sqrt{\deg_H(\mathbf u_t)}}.
 \label{eq:random_weight_bounds}
\end{equation}
Thus all arithmetic per search state has polynomial cost.

\par\medskip\noindent\textbf{1. Express the summed requirement.}
Choose a uniform assignment $x\in\bits^r$, independently of the remaining
random choices in $F$. For the analysis only, construct its first $T=\lceil r/2\rceil$
comparisons even if an earlier sum test would have rejected its prefix.
The gain $\mathrm{gain}_t(x)$ of $\mathbf u_t$, defined in
Section~\ref{sec:framework}, does not depend on $x_t$. Let $\Delta_t$ be its
expected value when $x_1,\ldots,x_{t-1}$ are fixed and the later values are
uniform. Then $p_t=1$ if $\Delta_t\ge0$ and $p_t=0$ otherwise is the
preferred value of $\mathbf u_t$ computed by the algorithm,
by~\eqref{eq:bias_preference} in the current residual formula.

Put $I=\{t\le T:x_t\ne p_t\}$, the indices at which a constraint is
created. Lemma~\ref{lem:pruning} expresses each such constraint as
$(2p_t-1)\mathrm{gain}_t(x)+1\le0$. Hence its weighted sum is
\begin{equation}
 A_T+Q_T\le0,
 \label{eq:random_sum}
\end{equation}
where
\begin{align*}
 A_T&=\sum_{t\in I}\omega_t(|\Delta_t|+1),\\
 Q_T&=\sum_{t\in I}\omega_t(2p_t-1)
                       \bigl(\mathrm{gain}_t(x)-\Delta_t\bigr).
\end{align*}
Here $A_T$ records the requirements at creation, while $Q_T$ records the
deviation supplied by the subsequent values. The identity follows from
$(2p_t-1)\Delta_t=|\Delta_t|$.
Preprocessing preserves each score difference, so these expressions
computed from the original clauses are exactly the ones used by the search.

\par\medskip\noindent\textbf{2. Bound the fraction that satisfies the sum.}
Lemma~\ref{lem:random_requirements} shows that the requirements are
large on most paths: $A_T\ge c_1T/52$ except with probability
$e^{-\Omega(r)}$. Lemma~\ref{lem:random_compensation} shows that the
later values rarely compensate much: $Q_T\ge-c_1T/104$ except with
probability $e^{-\Omega(r)}$. On the remaining event,
$A_T+Q_T\ge c_1T/104>0$, so the summed constraint fails. The calculation
in Appendix~\ref{app:random_main_completion} makes the exceptional
probabilities explicit and gives
\[
 \Prb_{F,x}(A_T+Q_T\le0\mid G_F=H)\le3e^{-c_0r}.
\]

\par\medskip\noindent\textbf{3. Convert the count to running time.}
Summing over the $2^r$ assignments to the searched variables, the
expected number of candidates on $H$ is at most $3\cdot2^re^{-c_0r}$,
which is at most $3\cdot2^ne^{-c_0n}$ because $r\le n$ and $c_0<\ln2$.
Lemma~\ref{lem:enumeration} turns this count into the expected time bound,
and Markov's inequality gives the high-probability bound.
Appendix~\ref{app:random_main_completion} gives the details.
\end{proof}

\section{Max-Cut on regular expander graphs}\label{sec:maxcut}

\subsection{Intuition}\label{sec:maxcut_intuition}

Recall from Section~\ref{sec:framework} that Max-Cut encodes every edge
$\{u,v\}$ of a graph $G$ by $(u\vee v)$ and $(\neg u\vee\neg v)$. The pair
contributes $1+[x_u\ne x_v]$, so maximizing the number of satisfied
clauses maximizes the cut. Without unit clauses, every gap is $0$ and no
sign is random, so neither Section~\ref{sec:bias} nor
Section~\ref{sec:random} applies. Unit clauses only add a constant to each
gain; we omit them in this overview. Throughout, $G$ is $d$-regular.

\paragraph{The comparison of a vertex.}
An assignment $x$ is a cut. We call the vertices $v$ with $x_v=1$ the
\emph{ones}; they form side~$1$. For a vertex $t$, let $k_t(x)$ be the
number of its neighbors among the ones. Moving $t$ from side~$0$ to
side~$1$ cuts its $d-k_t(x)$ edges to side~$0$ and uncuts its $k_t(x)$
edges to side~$1$, so
\[
 \mathrm{gain}_t(x)=d-2k_t(x).
\]
When a branch places $t$ on the side opposite to its preferred value, we
keep, as in Lemma~\ref{lem:pruning} but including ties, only the
assignments in which moving $t$ to the preferred side does not increase
the score. We call this constraint the \emph{standard comparison} of $t$.
For $t$ on side~$0$ it reads $k_t(x)\ge d/2$, and for $t$ on side~$1$ it
reads $k_t(x)\le d/2$: at least half of the neighbors of $t$ lie on the
other side. This is the constraint of Figure~\ref{fig:maxcut}.

\paragraph{One count moves all comparisons.}
Each standard comparison involves only $t$ and its neighbors, yet all of
them depend on one global number, the count of ones. If a completion puts
more vertices on side~$1$ than usual, every vertex sees more neighbors on
side~$1$, and the comparisons of all vertices placed on side~$0$ become
easier at the same time. For example, in a graph of edge density $1/2$, a
completion with $20$ more ones than average gives every vertex about $10$
more neighbors among the ones. In the complete graph, $k_t(x)$ is the
number of ones other than $t$, so every comparison depends only on this
count, just as the size of the cut does, and the comparisons cannot
accumulate: one suitable count satisfies all of them.

To separate this global effect, put $p=d/(n-1)$, the edge density of $G$,
and let $m_t(x)=\sum_{w\ne t}x_w$ be the number of ones other than $t$. On
average, $m_t$ other vertices contain $p\,m_t$ neighbors of $t$, so we split
\[
 k_t(x)=p\,m_t(x)+e_t(x),\qquad
 e_t(x)=\sum_{w\ne t}\bigl((A_G)_{tw}-p\bigr)\,x_w,
\]
where $A_G$ is the adjacency matrix of $G$. The first term, the
\emph{common term}, rises for all vertices whenever a one is added
anywhere. The second term, the \emph{excess} of $t$, has coefficients
$1-p$ on the $d$ neighbors of $t$ and $-p$ on the $n-1-d$ other vertices,
which sum to $0$. In the complete graph, $p=1$ and every excess vanishes.
This is why our result assumes $d\le(1-\varepsilon)n$ for a constant
$\varepsilon>0$.

\paragraph{Why expanders help.}
Apart from the common term, the comparisons depend on the excesses, and
how strongly the excesses are correlated depends on the graph. If the
neighborhoods of different vertices overlap roughly like random sets, the
excesses behave almost independently. We capture this by the eigenvalues
of $A_G$. List them as
$d=\lambda_1(A_G)\ge\lambda_2(A_G)\ge\dots\ge\lambda_n(A_G)$, and let
$\lambda(G)=\max_{i\ge2}|\lambda_i(A_G)|$ (and $\lambda(G)=0$ if $n\le1$).
Our result assumes the \emph{eigenvalue condition}
$\lambda(G)\le\lambda_0\sqrt d$ for a constant $\lambda_0$. Graphs that
satisfy it are called \emph{spectral expanders}. Under this condition, for
every unit vector $c$, the variance of $\sum_tc_te_t(x)$ under a uniformly
random assignment is at most $O(\lambda_0^2/\varepsilon)$ times the
variance of a single excess: no combination of excesses fluctuates much
more than one excess. The common term escapes this control. It corresponds
to the all-ones eigenvector of $A_G$, whose eigenvalue $d$ is excluded from
$\lambda(G)$.

\paragraph{Guessing the count of ones.}
Suppose we knew the number $h$ of ones in some optimal cut. Searching only
the \emph{slice} of assignments with exactly $h$ ones, with the standard
comparisons, would be exact: an optimum with $h$ ones lies in the slice
and satisfies every standard comparison. We do not know $h$, but it has
only $n+1$ possible values, so the algorithm makes one \emph{run} for each
guess $h\in\{0,\ldots,n\}$. Runs with a wrong guess need not retain an
optimum; they only cost time, which the analysis bounds for every run.

On the slice the common term no longer varies. Once $t$ is placed,
$m_t(x)=h-x_t$, and the standard comparison becomes a condition on the
excess of $t$ alone, with a threshold fixed by the run:
\begin{align*}
 t\text{ on side }0:&\qquad k_t(x)\ge\tfrac d2
   \iff e_t(x)\ge\tfrac d2-p\,h,\\
 t\text{ on side }1:&\qquad k_t(x)\le\tfrac d2
   \iff e_t(x)\le\tfrac d2-p\,h+p.
\end{align*}
The algorithm uses these \emph{count comparisons}
(Section~\ref{sec:countalgorithm}), with the first threshold relaxed by an
\emph{allowance} $p\le1$.

\paragraph{Leaving the slice.}
The count comparisons are linear inequalities, and the algorithm applies
them to all of $\bits^n$ without checking that a candidate has $h$ ones.
This only adds candidates, all of which are scored, and it simplifies the
analysis, because the bits of a uniformly random assignment are
independent on the whole cube. The rewriting makes this harmless: since
the coefficients of the excess sum to $0$, the count comparisons are
uncorrelated with the number of ones. The standard comparisons
$k_t(x)\ge d/2$ on the whole cube would instead bring the common term back.

\paragraph{Rounds and preferences.}
As in the previous sections, the algorithm branches on the first half of
the vertices, the prefix. It processes the prefix in \emph{rounds}: a first
round of about $n/4$ vertices and a second round with the rest of the
prefix. The preferred side of $t$ is given by the sign of the expected
\emph{adjusted gain} $d-2(p\,h+e_t(x))$, the gain with the number $m_t(x)$
of other ones replaced by $h$, when the values in earlier rounds are fixed
and all other values are uniform. Thus the preferences within a round are
fixed before the values of that round are examined, which keeps the
analysis in terms of a symmetric coefficient matrix. In run $h$, the
expected adjusted gain estimates the gain of $t$ in a cut with $h$ ones,
and a comparison asks the completion to reverse its sign. The expected
original gain would instead describe completions in which about half of
the remaining vertices are ones; if this disagreed with $h$, the comparison
of $t$ could already hold on average and impose no requirement.

\paragraph{The picture with a guessed count.}
Compared with Figure~\ref{fig:maxcut}, a run with guess $h$ changes the
preferred side: $t$ prefers side~$1$ when its expected adjusted gain is
nonnegative, not when most of its placed neighbors are on side~$0$. On cuts
with exactly $h$ ones, the recorded constraint is that of the figure, up to
the allowance: at least half of the neighbors of $t$ lie on the other side.
The algorithm only states it in terms of the excess of $t$, the number of
neighbors of $t$ among the ones minus the average $p\,h$. Adding ones at
random positions raises the neighbor counts of all vertices together but
leaves the excesses unchanged on average, so the constraints of different
vertices no longer move together.

\paragraph{Accumulating requirements.}
For a vertex $t$ in the second round, the expected adjusted gain depends on
the values in the first round: on how many of its first-round neighbors
are ones, compared with an average vertex. The eigenvalue condition gives
many edges between the rounds, so these expectations are typically of
order $\sqrt d$. When $t$ is placed against its preference, its comparison
requires the later values to compensate this amount. The compensation is a
linear function of the later values with coefficients
$-2((A_G)_{tw}-p)$, the deviations of the pair coefficients from their
common part (Section~\ref{sec:countalgorithm} defines both for general
formulas), and the eigenvalue condition bounds the largest absolute
eigenvalue of this matrix. A uniformly random completion therefore
satisfies the summed comparisons only with exponentially small
probability, and enumerating the solutions of the sum turns this into a
running-time saving.

\subsection{Notation and the count adjustment}\label{sec:countalgorithm}

Expanding the truth indicator of each clause writes the score as
\begin{equation}
 F(x)=c+\sum_v a_vx_v+\sum_{u<v}J_{uv}x_ux_v,
 \qquad x\in\bits^n,
 \label{eq:quadratic}
\end{equation}
where $J$ is symmetric with zero diagonal. We call the $J_{uv}$ the
\emph{pair coefficients}. A binary clause on $u,v$ with equal literal
polarities contributes $-1$ to $J_{uv}$, one with opposite polarities
contributes $+1$, and unit clauses change only $c$ and the $a_v$. The gain
of $t$ is therefore
\[
 \mathrm{gain}_t(x)=a_t+\sum_{w\ne t}J_{tw}x_w,
\]
so $J_{tw}$ is the change of $\mathrm{gain}_t$ when $x_w$ goes from $0$
to~$1$.

Fix a rational number $\beta$, the \emph{common part} that is subtracted
from every pair coefficient; for Max-Cut, $\beta$ is the average pair
coefficient. For a guess $h\in\{0,\ldots,n\}$, define the \emph{adjusted
gain}
\begin{equation}
 H_t^{h,\beta}(x)=\mathrm{gain}_t(x)
          +\beta\left(h-\sum_{w\ne t}x_w\right)
       =a_t+\beta h+\sum_{w\ne t}(J_{tw}-\beta)x_w.
 \label{eq:countgain}
\end{equation}
It replaces the term $\beta\sum_{w\ne t}x_w$ of the gain by its value
$\beta h$ on assignments with $h$ ones. The remaining coefficients are the
\emph{deviations} $J_{tw}-\beta$.

\begin{lemma}[Gains on assignments with $h$ ones]\label{lem:slice}
Let $x\in\bits^n$ satisfy $\sum_wx_w=h$. Then for every variable $t$,
\[
 H_t^{h,\beta}(x)=\mathrm{gain}_t(x)+\beta x_t.
\]
Consequently, for every variable $t$ and every $p_t\in\bits$,
\[
 (2p_t-1)H_t^{h,\beta}(x)-|\beta|\le(2p_t-1)\,\mathrm{gain}_t(x).
\]
\end{lemma}
\begin{proof}
Substitute $\sum_{w\ne t}x_w=h-x_t$ into the first expression
in~\eqref{eq:countgain}. The inequality follows since
$(2p_t-1)\beta x_t\le|\beta|$ for $x_t\in\bits$.
\end{proof}

When $t$ is placed against its preferred value $p_t$, its \emph{standard
comparison} is $(2p_t-1)\mathrm{gain}_t(x)\le0$: moving $t$ to $p_t$ does not
increase the score. Algorithm~\ref{alg:counts} uses instead the
\emph{count comparison} $(2p_t-1)H_t^{h,\beta}(x)\le|\beta|$, with
\emph{allowance} $|\beta|$. By the lemma, on assignments with exactly $h$
ones, every standard comparison implies the corresponding count
comparison, so run $h$ keeps every optimum with $h$ ones
(Lemma~\ref{lem:count_correctness}).

For the Max-Cut encoding of a $d$-regular graph without unit clauses,
$J_{uv}=-2(A_G)_{uv}$, and we take $\beta=-2p$, the average pair
coefficient, with $p=d/(n-1)$. In the notation of
Section~\ref{sec:maxcut_intuition},
\[
 \mathrm{gain}_t(x)=d-2\bigl(p\,m_t(x)+e_t(x)\bigr),\qquad
 H_t^{h,\beta}(x)=d-2\bigl(p\,h+e_t(x)\bigr),
\]
so the adjusted gain and the count comparisons are those of
Section~\ref{sec:maxcut_intuition}. The allowance $|\beta|=2p$ corresponds
to the allowance $p$ on the excess there, since
$\mathrm{gain}_t=d-2k_t$.

\subsection{The algorithm}

The procedure \textnormal{\textsc{CountComparisons}} follows the common scheme
of Section~\ref{sec:framework}, with three changes: it makes one run for
every guess $h$ of the number of ones, it uses count comparisons, and it
computes preferred values from the expected adjusted gain in rounds, so
that the preferred value of a variable depends only on the values in
earlier rounds.

\paragraph{Input.}
The procedure receives $F$, the rational common part $\beta$, and ordered,
nonempty rounds $B_1,\ldots,B_\ell$ that partition a prefix of $T$ variables
in a fixed input order; the remaining variables form the suffix. For Max-Cut
on a $d$-regular graph, Section~\ref{sec:maxcut_result} uses
$\beta=-2d/(n-1)$ and two rounds of sizes $\lfloor n/4\rfloor$ and
$\lfloor n/2\rfloor-\lfloor n/4\rfloor$.

\paragraph{Setup.}
Apply \textnormal{\textsc{Preprocess}}, record its score offset $s_0$, and
compute the coefficients $a,J$ of the preprocessed formula
by~\eqref{eq:quadratic}. Preprocessing changes only the constant part of the
objective, so it preserves all gains.

\paragraph{Search.}
For each guess $h=0,\ldots,n$, enumerate all assignments $\alpha$ to the
prefix and process its rounds in order. (When $\beta=0$, the comparisons do
not depend on $h$, and the single guess $h=0$ suffices.) For a variable $t$
in round $B_j$, compute the expected adjusted gain $\Delta_t$, that is,
$H_t^{h,\beta}$ with the variables of $B_{<j}=B_1\cup\dots\cup B_{j-1}$
fixed to $\alpha$ and every other variable replaced by $1/2$
in~\eqref{eq:countgain}. Set $p_t=1$ if $\Delta_t\ge0$ and $p_t=0$
otherwise. If $\alpha_t\ne p_t$, add the count comparison
\[
 (2p_t-1)H_t^{h,\beta}(x)\le|\beta|
\]
to a running sum $L$, which records left-hand sides minus right-hand sides.
After the prefix, enumerate the suffix assignments $y$ with
$L(\alpha,y)\le0$ by Lemma~\ref{lem:enumeration}, score each candidate in the
actual formula, and keep the best one over all guesses and prefixes.

\begin{algorithm}[tp]
\caption{\textnormal{\textsc{CountComparisons}}$(F,\beta,B_1,\ldots,B_\ell)$}
\label{alg:counts}
\begin{algorithmic}[1]
\STATE $M\gets-\infty$; $\alpha^\star\gets\emptyset$.
\STATE $(F,s_0)\gets\textnormal{\textsc{Preprocess}}(F)$.
\STATE Compute $a,J$ from $F$ by~\eqref{eq:quadratic}.
\STATE Set $\mathcal H\gets\{0,\ldots,n\}$ if $\beta\ne0$, and $\mathcal H\gets\{0\}$ otherwise.
\COMMENT{One run suffices when the comparison does not depend on $h$.}
\FOR{each $h\in\mathcal H$}
  \FOR{each prefix assignment $\alpha$ on $B_1\cup\cdots\cup B_\ell$}
    \STATE $L(x)\gets0$.
    \FOR{each round $B_j$, and each variable $t\in B_j$}
      \STATE Compute the average $\Delta_t$ of $H_t^{h,\beta}$ with $x_{B_{<j}}=\alpha_{B_{<j}}$ and all other $x_w=1/2$.
      \STATE Set $p_t\gets1$ if $\Delta_t\ge0$, and $p_t\gets0$ otherwise.
      \IF{$\alpha_t\ne p_t$}
        \STATE $L(x)\gets L(x)+(2p_t-1)H_t^{h,\beta}(x)-|\beta|$.
      \ENDIF
    \ENDFOR
    \FOR{each suffix $y$ with $L(\alpha,y)\le0$, by Lemma~\ref{lem:enumeration}}
      \STATE Let $q\gets s_0+F(\alpha,y)$.
      \STATE If $q>M$, set $M\gets q$ and $\alpha^\star\gets\alpha\cup y$.
    \ENDFOR
  \ENDFOR
\ENDFOR
\STATE \textbf{return} $(M,\alpha^\star)$.
\end{algorithmic}
\end{algorithm}

\subsection{Correctness}

\begin{lemma}[Correctness of count comparisons]\label{lem:count_correctness}
For every $F$, every rational $\beta$, and every ordered prefix partition,
\textnormal{\textsc{CountComparisons}} returns an optimum score and an
assignment attaining it.
\end{lemma}
\begin{proof}
Fix an optimal complete assignment $x^*$ and its correct count
$h=\sum_wx_w^*$. If $x_t^*\ne p_t$, changing just this variable to $p_t$
cannot increase the score, so
$(2p_t-1)\mathrm{gain}_t(x^*)\le0$.
For the correct count,~\eqref{eq:countgain} therefore gives
\[
 (2p_t-1)H_t^{h,\beta}(x^*)
 =(2p_t-1)\mathrm{gain}_t(x^*)+(2p_t-1)\beta x_t^*
 \le|\beta|.
\]
Every comparison on its prefix holds, and hence so does their sum.
The corresponding suffix is enumerated in this run.
When $\beta=0$, the comparisons do not depend on $h$, so the single run
also retains $x^*$. Preprocessing preserves all gains and records its
constant score change. Scoring candidates with that constant included
proves the claim.
\end{proof}

\subsection{Running-time guarantee}\label{sec:maxcut_result}

We use $A_G$ and $\lambda(G)$ as defined in Section~\ref{sec:maxcut_intuition}.

\begin{corollary}[Max-Cut on regular expanders]\label{cor:maxcut}
Fix $\lambda_0\ge1$ and $0<\varepsilon\le1/2$. There are
$c=c(\lambda_0,\varepsilon)<2$ and a deterministic polynomial-space
algorithm that solves Max-Cut exactly on every graph and runs in time
$\Ostar(c^n)$ on every $d$-regular graph $G$ with
\[
 d\le(1-\varepsilon)n,\qquad \lambda(G)\le\lambda_0\sqrt d.
\]
The same holds after adding any multiset of unit clauses to the Max 2-SAT
encoding with $(u\lor v)$ and $(\neg u\lor\neg v)$ on each edge.
\end{corollary}

The algorithm is as follows. Fix an integer degree threshold
$D\ge64\lambda_0^2$ and an integer size threshold $N_0\ge n_0(\varepsilon)$,
where $n_0(\varepsilon)$ is the threshold of
Lemma~\ref{lem:expander_conditions} below. If $n<N_0$, enumerate all
assignments. Otherwise, let $d$ be the maximum degree of the input graph.
If $d<D$, greedily construct an independent set, branch on its complement,
and optimize the remaining vertices separately, including any unit clauses.
If $d\ge D$, use Algorithm~\ref{alg:counts} with $\beta=-2d/(n-1)$,
$T=\lfloor n/2\rfloor$, and two prefix rounds of sizes $\lfloor n/4\rfloor$
and $T-\lfloor n/4\rfloor$ in a fixed arbitrary vertex order. The algorithm
computes no spectral quantities; the eigenvalue condition is used only in
the analysis.

\subsection{Proof of the running-time bound}

The proof combines three lemmas. Lemma~\ref{lem:count_accumulation} gives
three conditions on the deviations $J_{uv}-\beta$ under which
\textnormal{\textsc{CountComparisons}} lists only an exponentially small
fraction of all assignments. It is a special case of the general
accumulation theorem proved in Appendix~\ref{app:accumulation}.
Lemma~\ref{lem:expander_conditions} shows that the Max-Cut encoding of a
regular graph with small nontrivial eigenvalues satisfies these
conditions. Lemma~\ref{lem:bounded_degree} handles graphs of bounded
degree, for which a simpler argument suffices.

For the chosen $\beta$, let $N$ be the matrix of deviations,
$N_{uv}=J_{uv}-\beta$ for $u\ne v$ and $N_{uu}=0$. Let $D_N$ be the mean
squared row norm of $N$, and $W_t$ the squared norm of the part of row $t$
in the first round:
\[
 D_N=\frac1n\sum_{u\ne v}N_{uv}^2,\qquad
 W_t=\sum_{w\in B_1}N_{tw}^2\quad(t\in B_2).
\]
We write $\|N\|=\sup_{z\ne0}\|Nz\|_2/\|z\|_2$ for the operator norm,
and use the same notation for other matrices; for a symmetric matrix it is
the largest absolute eigenvalue. We also use the Frobenius norm
$\|N\|_F=(\sum_{u,v}N_{uv}^2)^{1/2}$. The symbols $\ind$ and $I$
denote the all-ones vector and the identity matrix, respectively.

\begin{lemma}[Accumulation with one guessed count]\label{lem:count_accumulation}
Fix $g,K>0$. Use two prefix rounds of sizes $\lfloor n/4\rfloor$ and
$T-\lfloor n/4\rfloor$, where $T=\lfloor n/2\rfloor$. Suppose $D_N>0$ and
\[
 \|N\|^2\le KD_N,\qquad
 \sum_{t\in B_2}\sqrt{W_t}\ge4\sqrt3\,gn\sqrt{D_N},\qquad
 T|\beta|\le\frac18gn\sqrt{D_N}.
\]
For sufficiently large $n$ in terms of $g,K$,
\textnormal{\textsc{CountComparisons}} takes polynomial space and
$\Ostar((n+1)(2^{n/2}+2^ne^{-cn}))$ time, for a constant $c=c(g,K)>0$.
The linear coefficients of the objective may be arbitrary.
\end{lemma}

The three conditions have a direct meaning. The first says that no
direction dominates the deviations: $\|N\|\le\sqrt{KD_N}$, where
$\sqrt{D_N}$ is the root-mean-square row norm. A common coefficient shared
by all pairs would create such a direction, which is why $\beta$ is removed
first. The second says that the first-round parts of the rows of the
second round have total norm $\Omega(n\sqrt{D_N})$. A uniformly random
first round then gives the vertices of the second round expected adjusted
gains of total absolute value $\Omega(n\sqrt{D_N})$, and hence large
requirements whenever they take their nonpreferred values. The third says
that the total allowance is a small fraction of this total.

\begin{proof}
Apply the two-round case of Theorem~\ref{res:counts} with $\tau=1/2$,
$P=\beta\ind\ind^{\mathsf T}$, and the exact guesses
$h\in\{0,\ldots,n\}$.
The comparison allowance there is $s_t=|\beta|$.
For vertices in the first round, the earlier-round quantity $W_t$ is $0$,
so the sum in that theorem's gap hypothesis is precisely the sum over
$B_2$ displayed here. The other two displayed assumptions give its norm
and allowance bounds. There are at most $n+1$ guesses. Its enumeration
and space conclusions prove the lemma.
\end{proof}

\begin{lemma}[Regular expanders satisfy the conditions]\label{lem:expander_conditions}
Fix $\lambda_0\ge1$ and $0<\varepsilon\le1/2$, and put
$K=9\lambda_0^2/(2\varepsilon)$ and $g=\varepsilon/(320\sqrt3)$. There is
$n_0=n_0(\varepsilon)$ with the following property. Let $G$ be a
$d$-regular graph on $n\ge n_0$ vertices with
\[
 64\lambda_0^2\le d\le(1-\varepsilon)n,\qquad \lambda(G)\le\lambda_0\sqrt d.
\]
Consider its Max-Cut encoding, possibly with unit clauses, with
$\beta=-2d/(n-1)$ and any two disjoint prefix rounds of sizes
$\lfloor n/4\rfloor$ and $\lfloor n/2\rfloor-\lfloor n/4\rfloor$. Then
$D_N>0$, and the three conditions of Lemma~\ref{lem:count_accumulation}
hold with these $K$ and $g$.
\end{lemma}

Each condition reflects a property of $G$. On vectors orthogonal to
$\ind$, the matrix $N$ acts as $-2A_G$ shifted by $\beta$, and $N\ind=0$.
The eigenvalue bound on $G$ therefore gives $\|N\|\le3\lambda_0\sqrt d$.
Every row of $N$ has squared length $D_N=4d(1-d/(n-1))$, which the density
restriction keeps above $2\varepsilon d$. By the expander mixing
inequality, the two rounds are joined by at least $dn/40$ edges, and each
such edge contributes to $W_t$. Finally, the total allowance $T|\beta|$ is
at most $2d$, which is of lower order than $n\sqrt{D_N}$.
Appendix~\ref{app:expander_conditions} gives the details.

Graphs of small degree are excluded by
Lemma~\ref{lem:expander_conditions}, but for them a direct argument gives
a fixed saving.

\begin{lemma}[Graphs of bounded degree]\label{lem:bounded_degree}
For every graph $G$ of maximum degree at most $d$, branching on the
complement of a greedy independent set solves Max-Cut on $G$, also after
adding any multiset of unit clauses to its Max 2-SAT encoding, in
polynomial space and time $\Ostar(2^{(1-1/(d+1))n})$.
\end{lemma}
\begin{proof}
Greedily choose an independent set $S$: repeatedly select a vertex and
remove it together with its neighbors. Each selection removes at most
$d+1$ vertices, so $|S|\ge n/(d+1)$. Branch on all assignments to the
vertices outside $S$. Since $S$ is independent, every remaining clause
then contains a single vertex of $S$, so each vertex of $S$ can be set
optimally on its own. The search has $2^{n-|S|+1}-1$ recursive calls.
\end{proof}

\begin{proof}[Proof of Corollary~\ref{cor:maxcut}]
The algorithm described after the corollary returns an optimum on every
input: exhaustive search and the bounded-degree branch are exact, and
\textnormal{\textsc{CountComparisons}} is exact by
Lemma~\ref{lem:count_correctness}. Unit clauses change only the linear
coefficients, which all three lemmas allow to be arbitrary. It remains to
bound the running time. Let $D\ge64\lambda_0^2$ be the degree threshold of
the algorithm. If $d<D$, Lemma~\ref{lem:bounded_degree} gives time
$\Ostar(2^{(1-1/D)n})$. If $d\ge D$ and $n\ge n_0(\varepsilon)$,
Lemma~\ref{lem:expander_conditions} verifies the conditions of
Lemma~\ref{lem:count_accumulation} with $K=9\lambda_0^2/(2\varepsilon)$
and $g=\varepsilon/(320\sqrt3)$. For all $n$ above a threshold depending
on $g$ and $K$, that lemma bounds the time by
$\Ostar((n+1)(2^{n/2}+2^ne^{-c_1n}))$ for a constant
$c_1=c_1(\lambda_0,\varepsilon)>0$. The finitely many remaining sizes
contribute a constant factor. Altogether the running time is
$\Ostar(c^n)$ with $c=\max\{2^{1-1/D},\sqrt2,2e^{-c_1}\}<2$, which depends
only on $\lambda_0$ and $\varepsilon$, and all procedures use polynomial
space.
\end{proof}

\subsection{Concrete families and the transition to random graphs}
Paley graphs on $q\equiv1\pmod4$ vertices, for prime powers $q$, have
degree $(q-1)/2$ and nontrivial eigenvalues $(-1\pm\sqrt q)/2$
\cite[Section~9.1.2]{BH12}. Thus $\lambda_0=1$, $\varepsilon=1/2$
suffice for all sufficiently large $q$. Connected non-bipartite Ramanujan
graphs satisfy $\lambda(G)\le2\sqrt{d-1}$, so they give another family
when the density restriction holds. Bipartite Ramanujan graphs have the
additional eigenvalue $-d$ and do not satisfy this hypothesis at growing
degree; ordinary Max-Cut on bipartite graphs is of course immediate.

Random regular graphs also satisfy the eigenvalue condition with high
probability.

\begin{corollary}[Max-Cut on random regular graphs]\label{cor:random_regular}
Fix $0<\varepsilon\le1/2$. There are $c=c(\varepsilon)<2$ and a
deterministic polynomial-space algorithm that solves Max-Cut exactly on
every graph, with the following property. Let $d=d(n)$ satisfy
$1\le d\le(1-\varepsilon)n$, with $dn$ even, and let $G$ be a uniformly
random $d$-regular graph on $n$ vertices. With probability tending to $1$
as $n\to\infty$, the algorithm runs on $G$ in time $\Ostar(c^n)$. The same
holds after adding any multiset of unit clauses to the Max 2-SAT encoding,
even if they are chosen after $G$ is drawn.
\end{corollary}
\begin{proof}
For $d\le n/2$, the theorems of Cook, Goldstein, and
Johnson~\cite{CGJ18}, for $d=O(n^{2/3})$, and of Tikhomirov and
Youssef~\cite{TY19}, for $n^{1/2}\le d\le n/2$, give an absolute constant
$\lambda_*\ge1$ such that $\lambda(G)\le\lambda_*\sqrt d$ with probability
tending to $1$. If $n/2<d\le(1-\varepsilon)n$, the complement $\bar G$ is a
uniformly random $(n-1-d)$-regular graph with $n-1-d<n/2$. Its adjacency
matrix is $\ind\ind^{\mathsf T}-I-A_G$. Hence every eigenvector of $A_G$
that is orthogonal to $\ind$, with eigenvalue $\mu$, is an eigenvector of
$A_{\bar G}$ with eigenvalue $-1-\mu$, and so
$\lambda(G)\le\lambda(\bar G)+1$. With probability tending to $1$, this
gives $\lambda(G)\le\lambda_*\sqrt{n-1-d}+1\le(\lambda_*+1)\sqrt d$.

In both cases $G$ satisfies the hypotheses of Corollary~\ref{cor:maxcut}
with $\lambda_0=\lambda_*+1$ and the given $\varepsilon$. We use the
algorithm of that corollary for these parameters. It is exact on every
input, and on graphs satisfying the hypotheses it runs in time
$\Ostar(c^n)$ with $c=c(\lambda_*+1,\varepsilon)$.
\end{proof}

For the random graph $G(n,p)$, in which every pair of vertices is joined
independently with probability $p$, the pair coefficients
$J_{uv}=-2A_{uv}$ have mean $-2p$. The deviations $-2(A_{uv}-p)$ of present
and absent edges have opposite signs; for $p=1/2$ they are independent
uniform signs. Random edges can thus play the role of the random literal
signs of Section~\ref{sec:random}. The next section develops this for
graphs with communities, with $G(n,p)$ as the one-community case.

\section{Max-Cut on random graphs with communities}\label{sec:sbm}

\subsection{Intuition}\label{sec:sbm_intuition}

\paragraph{One count per community.}
In Section~\ref{sec:maxcut}, a single number $\beta$ described the common
part of all pair coefficients, and the count of ones summarized the common
term of all comparisons. With communities, one count is not enough.
Consider two communities of equal size, with edge probability $0.6$ inside
a community and $0.4$ between communities. A vertex of the first community
has, on average, $0.6$ times the number of ones in its own community plus
$0.4$ times the number of ones in the other community as neighbors among
the ones. Moving a one from the second community to the first, without
changing the total, raises the expected neighbor count of every vertex of
the first community by $0.2$ and lowers that of every vertex of the second
community by $0.2$. In general, the vertices belong to $k$ communities, and
the probability $\pi_{ab}$ of an edge depends on the communities $a$ and $b$
of its endpoints. A vertex $t$ of community $a$ is adjacent, on average, to
a fraction $\pi_{ab}$ of the ones of community $b$, so the common term of
its neighbors among the ones depends on the $k$ counts of ones in the
communities. If the communities were known, we could guess all $k$ counts.
There are at most $(n+1)^k$ guesses, polynomially many for fixed $k$. In
the run with the right guess, each comparison is rewritten with the actual
counts replaced by the guessed ones, exactly as in
Section~\ref{sec:maxcut}.

\paragraph{The randomness of the edges.}
What remains after this rewriting is the deviation of each edge indicator
from its expected value. These deviations are independent across pairs, and
their signs are random: a present edge deviates upwards and an absent edge
downwards. They play the role that random literal signs played in
Section~\ref{sec:random}, and that the eigenvalue condition played in
Section~\ref{sec:maxcut}. With high probability, they satisfy the
conditions of the accumulation theorem of Appendix~\ref{app:accumulation},
the analogues of those of Lemma~\ref{lem:count_accumulation}: no direction
dominates the deviations, and the deviations between the two prefix rounds
create large requirements.

\paragraph{Hidden communities.}
The algorithm is not told the communities, and it does not try to recover
them. It only needs to remove the common part of the pair coefficients,
their expectation: the matrix with entries $-2\pi_{ab}$, which has rank at
most $k$. The matrix $J$ of pair coefficients, which the algorithm computes
from the input, is this low-rank matrix plus the random deviations. The
algorithm computes approximate eigenvectors $v_1,\ldots,v_r$, with
$r\le k$, for the eigenvalues of $J$ of largest absolute value, and guesses
the projections $\langle v_i,x\rangle$ of the assignment on a polynomial
grid. These projections play the role of the community counts: the
projection of an assignment onto the indicator vector of a community is
exactly its count of ones there. The eigenvectors need not reveal the
communities. Removing the directions with the largest eigenvalues leaves a
remainder whose norm is at most that of the random deviations, up to
rounding, whatever the expected matrix is (Appendix~\ref{app:discovery}).
A part of the expectation that is not removed is therefore no larger than
the random deviations and does no harm. This is why the probabilities
inside and between communities may even be equal. Each guess gives
comparisons with an explicit allowance, which covers the grid error and
the variable's own contribution to the projection, and at least one guess
retains an optimum. The search itself is the two-round search of
Section~\ref{sec:maxcut}. Unlike in Section~\ref{sec:maxcut}, a spectral
computation is part of the algorithm, not only a tool in its analysis.

\subsection{The input model}\label{sec:sbm_model}

Fix a partition $V=V_1\cup\cdots\cup V_k$ and a symmetric probability
matrix $\Pi=(\pi_{ab})_{a,b\le k}$. Include each unordered pair of distinct
vertices independently, with probability $\pi_{ab}$ when its endpoints
belong to $V_a,V_b$. This is the stochastic block model~\cite{Abbe18}.
The partition can be arbitrary and need not be balanced. Alternatively,
it can be drawn first; the result below holds uniformly conditional on it.
The number $k$ is fixed, while the probabilities may depend on $n$.

Define the smallest and largest edge variances by
\[
 v_{\min}=\min_{a,b}\pi_{ab}(1-\pi_{ab}),\qquad
 v_{\max}=\max_{a,b}\pi_{ab}(1-\pi_{ab}).
\]
Our result assumes that $v_{\min}>0$, that $v_{\max}\le\chi v_{\min}$ for a
fixed $\chi\ge1$, and that $nv_{\min}\ge C\log n$ for a sufficiently large
constant $C$. These conditions measure the amount of random variation, not
how well the communities can be distinguished.

\subsection{The algorithm for hidden communities}

For a graph input, first form its Max-Cut encoding; additional units, if
present, are retained. The following procedure takes that formula $F$, a
fixed bound $k$ on the number of communities, and a fixed bound $\chi\ge1$
on the ratio of the largest to the smallest edge variance; $\chi$ only sets
the grid spacing in Appendix~\ref{app:discovery}. The algorithm has three
phases.

\paragraph{Spectral step.}
If $n<4$, solve the instance by exhaustive search. Otherwise, compute the
coefficients of the objective in~\eqref{eq:quadratic}; the matrix
$J=-2A_G$ records the graph edges, and units affect only the constant and
linear coefficients. Compute a rational approximate diagonalization of $J$
(Appendix~\ref{app:discovery}), and retain the $r\le k$ approximate
eigenvalues $\lambda_i$ of largest absolute value, with their approximate
eigenvectors $v_i$, which are rational and orthonormal. This gives
\[
 P=\sum_{i=1}^{r}\lambda_i v_i v_i^{\mathsf T},
 \qquad N=\operatorname{offdiag}(J-P),
\]
where $\operatorname{offdiag}$ sets the diagonal entries to $0$. Compute
$D_N=n^{-1}\sum_{u\ne v}N_{uv}^2$. If $D_N=0$, use exhaustive search.

\paragraph{Projection guesses.}
Construct rational grids $\mathcal Y_i$ with spacings
$\delta_i\ge n^{-2}$ (Appendix~\ref{app:discovery}) such that every
projection $\langle v_i,x\rangle=\sum_t v_i(t)x_t$ of a Boolean assignment
is within $\delta_i$ of a point of $\mathcal Y_i$. The lower bound $n^{-2}$
on the spacing keeps the number of grid tuples polynomial for fixed $k$.
Compute the comparison allowances
\begin{equation}
 s_t=\sum_{i=1}^{r}|\lambda_i|\,|v_i(t)|
                       \bigl(\delta_i+|v_i(t)|\bigr),
                       \label{eq:community_slack}
\end{equation}
which remain fixed across all guesses.

\paragraph{Comparison search.}
Let $B_1$ contain the first $\lfloor n/4\rfloor$ variables in the input
order and $B_2$ the next $T-\lfloor n/4\rfloor$, where
$T=\lfloor n/2\rfloor$; the remaining variables form the suffix. For each
projection tuple $\hat y\in\mathcal Y_1\times\cdots\times\mathcal Y_r$, use
the adjusted gain
\begin{equation}
 H_t^{\hat y}(x)
 =a_t+\sum_{i=1}^{r}\lambda_i\hat y_i v_i(t)
                 +\sum_{w\ne t}N_{tw}x_w,
                 \label{eq:community_gain}
\end{equation}
the projection counterpart of~\eqref{eq:countgain}, and run the search of
\textnormal{\textsc{CountComparisons}} with the rounds $B_1,B_2$, the
adjusted gain $H_t^{\hat y}$, and the allowance $s_t$ in place of $|\beta|$.
Every candidate is scored in the original formula, and the best one over
all tuples and prefixes is returned.

\begin{algorithm}[tp]
\caption{\textnormal{\textsc{CommunityComparisons}}$(F,k,\chi)$}
\label{alg:communities}
\begin{algorithmic}[1]
\STATE $M\gets-\infty$; $\alpha^\star\gets\emptyset$.
\IF{$n<4$}
  \STATE \textbf{return} the best score and assignment found by exhaustive search.
\ENDIF
\STATE Compute $a,J$ in~\eqref{eq:quadratic}.
\STATE Construct $P,N$ by the approximate diagonalization in Appendix~\ref{app:discovery}.
\COMMENT{This polynomial computation uses only $J$ and $k$.}
\STATE Compute $D_N=n^{-1}\sum_{u\ne v}N_{uv}^2$.
\IF{$D_N=0$}
  \STATE \textbf{return} the best score and assignment found by exhaustive search.
\ENDIF
\STATE Construct the projection grids and allowances by Appendix~\ref{app:discovery}.
\STATE Fix the two prefix rounds $B_1,B_2$ described above.
\FOR{each projection tuple $\hat y\in\mathcal Y_1\times\cdots\times\mathcal Y_r$}
  \FOR{each prefix assignment $\alpha$ on $B_1\cup B_2$}
    \STATE $L(x)\gets0$.
    \FOR{each round $B_j$ and each $t\in B_j$}
      \STATE Compute the average $\Delta_t$ of $H_t^{\hat y}$ using $\alpha$ on earlier rounds and $1/2$ for all other variables.
      \STATE Set $p_t\gets1$ if $\Delta_t\ge0$, and $p_t\gets0$ otherwise.
      \IF{$\alpha_t\ne p_t$}
        \STATE $L(x)\gets L(x)+(2p_t-1)H_t^{\hat y}(x)-s_t$.
      \ENDIF
    \ENDFOR
    \FOR{each suffix $y$ satisfying $L(\alpha,y)\le0$, by Lemma~\ref{lem:enumeration}}
      \STATE Evaluate $F(\alpha,y)$; update $M,\alpha^\star$ if this score is larger.
    \ENDFOR
  \ENDFOR
\ENDFOR
\STATE \textbf{return} $(M,\alpha^\star)$.
\end{algorithmic}
\end{algorithm}

The comparisons use the rational directions actually computed, so no
exact-real arithmetic is assumed.

\subsection{Correctness}

\begin{lemma}[Correctness with hidden communities]\label{lem:community_correctness}
\textnormal{\textsc{CommunityComparisons}} returns an optimum on every
input, independently of whether it was generated by the stated model.
\end{lemma}
\begin{proof}
The exhaustive-search cases are exact. Otherwise fix an optimal
assignment $x^*$ and choose a tuple $\hat y$ within the grid errors of
its projections. For any variable $t$, direct subtraction of the
original and adjusted gains gives
\[
 \mathrm{gain}_t(x^*)-H_t^{\hat y}(x^*)
 =\sum_i\lambda_i v_i(t)
       \bigl(\langle v_i,x^*\rangle-\hat y_i-v_i(t)x_t^*\bigr).
\]
Its absolute value is at most $s_t$ by~\eqref{eq:community_slack}.
If the optimum takes the nonpreferred value at $t$, changing it to
$p_t$ cannot improve its score, so
$(2p_t-1)\mathrm{gain}_t(x^*)\le0$. Therefore
$(2p_t-1)H_t^{\hat y}(x^*)\le s_t$.
Every comparison at its prefix, and hence their sum, is satisfied.
The suffix of $x^*$ is enumerated in this run. Every reported score
is evaluated in the original formula, which proves the claim.
\end{proof}

\subsection{Running-time guarantee}

A binary probability matrix with singleton communities can represent any
graph, but it violates the assumptions below: the number of communities
must be fixed, and the edge variances must be nonzero and comparable.

\begin{theorem}[Exact Max-Cut on the stochastic block model]\label{thm:sbm}
Fix $k\ge1$ and $\chi\ge1$. There are constants $C=C(k,\chi)$ and
$\gamma=\gamma(k,\chi)>0$ and a deterministic polynomial-space algorithm
that receives a graph together with $k$ and $\chi$, but neither the
partition nor the probabilities $\pi_{ab}$, and solves Max-Cut exactly on
every graph, with the following property. Whenever
\begin{equation}
 0<v_{\min},\qquad v_{\max}\le\chi v_{\min},\qquad
 n v_{\min}\ge C\log n,                       \label{eq:sbmvariance}
\end{equation}
the algorithm runs on a graph from this model in time
$\Ostar(2^{(1-\gamma)n})$ with probability $1-o(1)$. The same holds with
arbitrary additional unit clauses, even chosen after the graph is drawn.
\end{theorem}

\subsection{Proof of the running-time bound}

The algorithm is exact by Lemma~\ref{lem:community_correctness}. Its
running time follows from the discovery lemma of
Appendix~\ref{app:discovery} (Lemma~\ref{cor:planted}), which concerns
independent pair coefficients whose expected values form a matrix of
bounded rank. In the Max 2-SAT encoding of a block-model graph, the
coefficient of a pair is $-2$ times its edge indicator. Given the
partition, these coefficients are independent, their means $-2\pi_{ab}$
form a matrix of rank at most $k$, and their variances are
$4\pi_{ab}(1-\pi_{ab})$. Appendix~\ref{app:sbm_proof} applies the lemma with
$r=k$, $\nu=4v_{\min}$, and $B=B_0=\mu=2$; condition~\eqref{eq:sbmvariance}
gives its variance hypotheses. The probability event depends only on the
graph, and the bound is uniform over all linear coefficients. This is why
additional unit clauses may be chosen after the graph is known.

\subsection{The usual within-community and between-community probabilities}
When $\pi_{aa}=p$ and $\pi_{ab}=q$ for $a\ne b$, the assumptions reduce to
\[
 \max\{p(1-p),q(1-q)\}\le\chi\min\{p(1-p),q(1-q)\},\qquad
 n\min\{p(1-p),q(1-q)\}\ge C\log n.
\]
No separation between $p$ and $q$ is required; in particular $p=q$ is
allowed. Concrete regimes are:
\begin{itemize}
\item \emph{Dense blocks:} all $\pi_{ab}\in[\varepsilon,1-\varepsilon]$ for
a fixed $0<\varepsilon\le1/2$. The hypotheses hold for sufficiently large
$n$, with constants depending only on $k,\varepsilon$.
\item \emph{Sparse blocks:} $\pi_{ab}=c_{ab}\rho$ with fixed
$0<c_-\le c_{ab}\le c_+$, $\rho=o(1)$, and $n\rho\ge C'\log n$.
\item \emph{Almost complete blocks:} $\pi_{ab}=1-c_{ab}\rho$ with the
same bounds on $c_{ab}$ and $\rho$.
\end{itemize}
In the last two regimes, $c_-\rho/2\le\pi_{ab}(1-\pi_{ab})\le c_+\rho$ for
large $n$, so $\chi=2c_+/c_-$ and $C'=2C/c_-$ suffice. The saving is fixed
throughout each regime.

\section{Formal verification}\label{sec:verification}

An accompanying Lean~4 artifact formalizes the main arguments of this
paper. For all four classes, it checks that the algorithms return an
optimum on every input. It also checks the bounds on the number of search
nodes behind Theorems~\ref{thm:bias}, \ref{thm:random_main}, and~\ref{thm:sbm}
and Corollaries~\ref{cor:lineargap} and~\ref{cor:maxcut}, together with the
supporting lemmas. Two published estimates on random matrices, used in
Appendices~\ref{app:membership} and~\ref{app:discovery}, enter as explicit
hypotheses. The polynomial bounds on arithmetic cost and space are proved
on paper only. Corollary~\ref{cor:random_regular},
which combines Corollary~\ref{cor:maxcut} with published spectral-gap
theorems, is not formalized. Appendix~\ref{app:verification} describes the
scope of the formalization in detail.

\section{Conclusions and future directions}

All our algorithms use the same mechanism. When the search backtracks over
a variable, the second branch only needs the assignments of the neighbors
for which its value is strictly better, and these local constraints, weak
one by one, add up along a branch. Our four classes show different ways in
which they accumulate. Biased variables give comparisons with large
requirements. In semirandom formulas, the requirements are smaller, but
after weighting they still add up, because a shared variable with random signs helps some comparisons
and hurts others. In the two Max-Cut classes, guessing the number of ones
(Section~\ref{sec:maxcut}) or the projections of the assignment
(Section~\ref{sec:sbm}) removes the common part of the pair coefficients,
which would otherwise move all comparisons together. In each case, a bound
on how much the remaining values can compensate together turns the
accumulated requirements into a fixed exponential saving in polynomial
space.

Finding a polynomial-space algorithm with running time
$\Ostar(2^{(1-\gamma)n})$ for all dense \mbox{Max 2-SAT} instances remains an
important open theoretical problem. A natural next case is formulas whose
pair coefficients keep a dominant direction even after a few statistics of
the assignment have been guessed; the accumulation theorem of
Appendix~\ref{app:accumulation} does not apply to them. 

In future work,
however, our focus will be on practical applications.
The main objective of the paper is to present techniques of practical
importance that yield theoretically provable gains. The pruning constraints are local, cheap to compute, and can be
added to any branching algorithm for Max 2-SAT. The same reasoning can be extended to yield weaker hard
constraints for general Max-SAT, since a longer clause has more literals that need to be falsified in order to create a constraint. 
However, formulas with many short clauses appear in various applications and should still profit from pruning constraints, the more
so as clauses become shorter during branching.

Furthermore, Max-SAT has applications in areas like statistical physics
(spin glasses)~\cite{barahona1988application,liers2004computing}, materials
science~\cite{huang2016finding}, biology~\cite{malikic2019phiscs}, and
protein design~\cite{allouche2014computational}. Public instances exist for
many of these problems~\cite{jabs2026maxsat,wiegele2007biqmac}.

We plan to combine our pruning
constraints with suitable heuristics in a branch-and-bound solver and to
evaluate them on these instances, where we expect savings much larger than
the bounds proven here. These experiments are not restricted to dense
instances; spin glasses on two- and three-dimensional grids, for example,
have bounded degree. We expect the constraints to be even more useful on
sparse instances, where each constraint involves only a few variables.

\section*{Use of AI tools}

This paper was developed together with the large language models Claude
Opus~5.5 and GPT-6 Astra over several weeks. The author set the research
goal and proposed the approach: pruning constraints obtained by comparing
the two values of a variable during branching. The author also provided the
pruning lemma (Lemma~\ref{lem:pruning}) and a simpler version of
Theorem~\ref{thm:bias}, essentially Corollary~\ref{cor:lineargap}, together
with its proof. The other results were developed in discussions with both
models, which also contributed technical ideas within this approach. Almost
all proofs and most of the text were first written by the models; the
author directed the work, improved the presentation, added the intuition,
and consolidated the results. The models were never asked to design an
algorithm for Max 2-SAT from a general prompt. The main arguments are
formalized in Lean~4; Section~\ref{sec:verification} lists what is not
formalized. The author takes full responsibility for the paper.

{\small
	\bibliographystyle{plain}
	\bibliography{pruning_2_sat}
}

\newpage
\appendix

\section{Proofs for the literal-gap analysis}\label{app:bias}

This appendix contains the proofs omitted from Section~\ref{sec:bias}.

\subsection{The combined comparison statistic}\label{app:bias_statistics_proof}

\begin{proof}[Proof of Lemma~\ref{lem:bias_statistics}]
\emph{The expectation.}
Fix $\mathbf v$ and the values of every other variable. The score gain of
its preferred value over its nonpreferred value is independent of
$x_{\mathbf v}$. Averaging over the two equally likely values of
$x_{\mathbf v}$, $H_{\mathbf v}$ is $0$ with
probability $1/2$ and equals that gain with probability $1/2$. Averaging next over the
other variables, the expected preferred-over-nonpreferred gain is
$|g_{\mathbf v}|/2$, by~\eqref{eq:effectivegap} and the preference rule.
Hence $\E H_{\mathbf v}=|g_{\mathbf v}|/4$. Add these expectations;
no independence between different comparisons is needed.

\emph{The effect of one variable.}
Let $x'$ differ from $x$ only at $\mathbf w$. In the summand
$H_{\mathbf w}$, if it is present, the assignments with $\mathbf w$
reset to its preferred value coincide. Consequently
\[
 H_{\mathbf w}(x')-H_{\mathbf w}(x)=F_0(x)-F_0(x').
\]
Only clauses containing $\mathbf w$ can change their satisfaction, each
by at most one. The absolute value of this difference is at most
$B_{\mathbf w}$.

For $\mathbf v\ne\mathbf w$, changing $\mathbf w$ does not change whether
$x_{\mathbf v}$ is preferred. If it is preferred, both values of
$H_{\mathbf v}$ are zero. Otherwise, the gain identity in
Lemma~\ref{lem:pruning} shows that only occurrences of literals of
$\mathbf w$ in $C_{\mathbf v}$ can affect $H_{\mathbf v}$, by at most
one per occurrence. Across the comparisons for all variables other than
$\mathbf w$, these occurrences number exactly $D_{\mathbf w}$: each
binary clause containing $\mathbf w$ contributes once, to its partner's
comparison. Therefore the sum over the subset $P\setminus\{\mathbf w\}$
changes by at most $D_{\mathbf w}$. The triangle inequality gives the
claimed bound. Repeated or overlapping neighborhoods require no separate
assumption.

\emph{The consequence of satisfying the sum.}
Write $C_{\mathbf v}$ as
$\operatorname{lhs}(C_{\mathbf v},x)\le b_{\mathbf v}$, and put
$I_P(x)=\{\mathbf v\in P:x_{\mathbf v}\ne p_{\mathbf v}\}$.
For every $\mathbf v\in I_P(x)$, the offset of minus one in the strict constraint~\eqref{eq:gap_pruning} gives
\[
 \operatorname{lhs}(C_{\mathbf v},x)-b_{\mathbf v}
       =H_{\mathbf v}(x)+1.
\]
All other summands in $L_P$ are zero. Adding the displayed identities yields
\begin{equation}
 \sum_{\mathbf v\in I_P(x)}
       \bigl(\operatorname{lhs}(C_{\mathbf v},x)-b_{\mathbf v}\bigr)
       =L_P(x)+|I_P(x)|.
 \label{eq:bias_surplus}
\end{equation}
Satisfaction of the sum makes the left-hand side nonpositive. Thus
$L_P(x)\le-|I_P(x)|\le0$, as required. Substitution of the prefix does
not change this identity on its completions.
\end{proof}

\subsection{A bounded-differences inequality}\label{app:bias_concentration}

The randomness here is a uniform Boolean assignment used to count
candidates of a deterministic algorithm; the input formula is fixed.
We give a conservative bounded-differences estimate with a short finite
proof. It is the only probability inequality used in
Lemma~\ref{lem:bias_count}.

\begin{lemma}[Bounded differences on uniform Boolean assignments]
\label{lem:bias_concentration}
Let $f:\bits^n\to\mathbb R$, and let $c_1,\ldots,c_n\ge0$ be constants.
Suppose that changing only coordinate $i$ changes $f$ by at most $c_i$.
Put $V=\sum_{i=1}^n c_i^2$, and assume $V>0$. For a uniform Boolean
assignment $x$ and every $a\ge0$,
\[
 \Prb(f(x)\le\E f-a)\le\exp(-a^2/(2V)).
\]
\end{lemma}
\begin{proof}
We first prove by induction on $n$ that, for every real $t$,
\begin{equation}
 \E e^{t(f-\E f)}\le e^{t^2V/2}.
 \label{eq:bias_mgf}
\end{equation}
The claim also permits $V=0$ at this intermediate step. For $n=0$,
$f$ is constant and both sides are one.

For the induction step, write $f_0,f_1$ for the restrictions obtained by
fixing the first bit to $0$ and $1$. Let $\mu_b=\E f_b$ and
$\mu=(\mu_0+\mu_1)/2=\E f$. Since $|f_0(y)-f_1(y)|\le c_1$ for
every assignment $y$ of the remaining bits, averaging gives
$|\mu_0-\mu_1|\le c_1$. Each restriction satisfies the same coordinate
bounds $c_2,\ldots,c_n$. Applying the induction hypothesis to both gives
\begin{align*}
 \E e^{t(f-\mu)}
 &=\frac12\sum_{b\in\bits}e^{t(\mu_b-\mu)}
                   \E e^{t(f_b-\mu_b)}\\
 &\le\exp\left(\frac{t^2}{2}\sum_{i=2}^n c_i^2\right)
       \frac{e^{t(\mu_0-\mu_1)/2}+e^{-t(\mu_0-\mu_1)/2}}2\\
 &\le\exp\left(\frac{t^2}{2}\sum_{i=2}^n c_i^2+
                       \frac{t^2c_1^2}{8}\right)
 \le e^{t^2V/2}.
\end{align*}
For the second inequality we used
$(e^z+e^{-z})/2\le e^{z^2/2}$. This follows by comparing their power
series, since $(2j)!\ge2^j j!$ for every integer $j\ge0$.
This proves~\eqref{eq:bias_mgf}.

Finally let $t=a/V\ge0$. On the event $f\le\E f-a$,
$e^{-t(f-\E f)}\ge e^{ta}$. Averaging this pointwise bound and applying
\eqref{eq:bias_mgf} with $-t$ yields
\[
 \Prb(f\le\E f-a)e^{ta}
 \le\E e^{-t(f-\E f)}\le e^{t^2V/2}.
\]
Divide by $e^{ta}$ and substitute $t=a/V$ to obtain
$e^{-ta+t^2V/2}=e^{-a^2/(2V)}$. The calculation includes $a=0$.
\end{proof}

\subsection{Implementation costs}\label{app:bias_implementation}

We check that \textnormal{\textsc{LiteralGap}} uses polynomial time per
search node and polynomial space, as used in the proofs of
Theorem~\ref{thm:bias} and Corollary~\ref{cor:lineargap}.
Let $m$ be the number
of clauses of $F_0$, with multiplicity. A pruning constraint has at most
$m$ occurrences and bound magnitude at most $m+1$. A sum of at most
$n$ such constraints therefore has at most $nm$ occurrences and bound
magnitude at most $n(m+1)$. Substitution and complementary cancellation
cannot increase the occurrence list. The magnitude of the resulting bound
is at most its original magnitude plus twice its original occurrence
count. These are polynomial-size integer data, so constructing and testing
each sum takes polynomial time. Every candidate score is also computed
in polynomial time.

Preprocessing, computing the gaps, and choosing $P$ take polynomial time.
The depth-first prefix and suffix traversal has at most $n+1$ levels;
its active states and the cancellation procedure use polynomial space.
Only one best complete assignment is retained. Multiplying the node bound
by the polynomial work per node proves the stated time bound.

\subsection{Linear gaps in simple formulas}\label{app:lineargap_proof}

\begin{proof}[Proof of Corollary~\ref{cor:lineargap}]
The case $n=0$ is trivial. For $n\ge1$, each variable initially has at
most $4(n-1)$ binary occurrences and at most two unit occurrences.
The pair reduction preserves every effective gap. At its shared literal,
two signed binary occurrences are replaced by one unit with coefficient
two; at its other variable, opposite signed binary occurrences cancel.
Opposite-unit cancellation also preserves the gap. Neither reduction
increases a variable's total number of binary and unit occurrences.
Thus, in $F_0$,
\[
 B_{\mathbf v}\le4(n-1)+2\le4n,\qquad Q\le16n^3,
 \qquad G\ge\varepsilon\eta n^2>0.
\]
In particular $Q>0$. Lemma~\ref{lem:bias_count} gives
\[
 W(F_0)\le2^n\exp\left(-\frac{\varepsilon^2\eta^2 n}{8192}\right).
\]
Combining this count with $|P|\le\lceil n/2\rceil$ as in the theorem's
proof bounds the nodes by
\[
 (4+2n)\exp\bigl((\log2-\varepsilon^2\eta^2/8192)n\bigr).
\]
Since $\log2/48000\le1/8192$, this implies the claimed conservative
exponent. The arithmetic and space bounds are the same as for
Theorem~\ref{thm:bias}; see Appendix~\ref{app:bias_implementation}. The
argument does not require a positive residual binary degree on each
biased variable; bias converted to units is included in $G$.
\end{proof}

\section{Proofs for semirandom formulas}\label{app:random_signs}

The main text states the result for two semirandom models.
For the proof, it is convenient to work with a single independent-pair
condition encompassing both. The conditioning, moment, and concentration
lemmas in Appendix~\ref{app:random_estimates} supply the estimates used in
Theorem~\ref{thm:random_main}, Appendix~\ref{app:random_main_completion}
completes its proof, and Appendix~\ref{app:random_implementation} gives
implementation details of the algorithm.

\subsection{A general pair model and supporting estimates}\label{app:random_estimates}

\begin{definition}[General semirandom model]\label{def:invariant}
Fix an integer constant $m_0\ge1$ and a real constant $0<\kappa\le m_0^2$, a set of $n$ input variables,
and any set $E$ of unordered pairs of distinct input variables.
Let $V=\bigcup_{e\in E}e$ and $G=(V,E)$.
For each edge $e=\{u,v\}$ draw a multiset $M_e$ of binary clauses on
$u,v$, independently over the edges. The binary part of $F$ is the
multiset union of these $M_e$; pairs outside $E$ contribute no binary
clauses. A sampled $M_e$ may be empty, so the observed graph $G_F$
in~\eqref{eq:observed_support} may have fewer edges than $G$. Assume that
\begin{enumerate}
\item $|M_e|\le m_0$;
\item complementing the literal of either endpoint in every clause of
$M_e$ leaves its distribution unchanged;
\item for each endpoint $u$,
\[
 \E\left[\left(\sum_{c\in M_e}s_u(c)\right)^2\right]\ge\kappa,
\]
where $s_u(c)=1$ if the literal of $u$ in $c$ is positive and
$s_u(c)=-1$ otherwise.
\end{enumerate}
Add any multiset of unit clauses on the input variables, fixed
independently of these draws.
\end{definition}

We call $s_u(c)$ the \emph{endpoint sign} of $c$ at $u$. The last
condition requires some variation in the difference between positive
and negative occurrences at an endpoint. One uniformly signed
clause per edge has $m_0=\kappa=1$. Including each of the four clauses
independently with probability $q\in(0,1)$ has $m_0=4$ and
$\kappa=4q(1-q)$. In particular, $q=1/4$ gives the second model of
Section~\ref{sec:random_intuition}.

\begin{theorem}[Semirandom formulas in the general model]\label{thm:weighted}
Fix an integer $m_0\ge1$ and a real $0<\kappa\le m_0^2$. For a formula $F$ on $n$ variables drawn
according to Definition~\ref{def:invariant} with any fixed generating
graph $G$, \textnormal{\textsc{RandomFormula}}, applied to $F$ alone, uses
polynomial space and has expected running time
\[
 \Ostar(2^{n/2}+2^n e^{-c_0n}),
 \qquad
 c_0=\frac{c_1^2}{10816m_0^2},\qquad
 c_1=\frac{\sqrt\kappa}{2\sqrt{1+4m_0^4/\kappa^2}}.
\]
With probability at least $1-3e^{-c_0n/2}$, its running time is
$\Ostar(2^ne^{-c_0n/2})$.
It returns an optimum on every input; the expectations and probabilities
refer only to the random input formula.
\end{theorem}

\begin{proof}
Use the proof of Theorem~\ref{thm:random_main} with the present values of
$m_0,\kappa,c_1,c_0$. The lemmas below are stated and proved for this
general model. They give the same three exceptional-event estimates
and hence $\Prb(A_T+Q_T\le0\mid G_F=H)\le3e^{-c_0|V(H)|}$.
The enumeration count, averaging over the observed graph, and the
final Markov bound are unchanged. The prefix work remains
$\Ostar(2^{n/2})$. All pair-law assumptions used by that argument are
precisely the three conditions in Definition~\ref{def:invariant}.
\end{proof}

The algorithm does not know which pairs were available for sampling. The
next lemma justifies analyzing it after fixing the graph it actually
observes. Crucially, whether a pair is empty does not distinguish between
the literal polarities.

\begin{lemma}[Conditioning on the observed pairs]\label{lem:observed_pairs}
Suppose $F$ is generated according to Definition~\ref{def:invariant}.
Conditional on any observed graph $G_F=H$ of positive probability, the
clause multisets on its edges are independent. Each has its original
distribution conditioned on being nonempty, and satisfies all three
conditions of that definition with the same $m_0$ and $\kappa$.
\end{lemma}
\begin{proof}
The event $G_F=H$ requires $M_e\ne\emptyset$ for each edge of $H$ and
$M_e=\emptyset$ for every other edge of the generating graph. These events
concern separate independent draws. Conditioning on them therefore
preserves independence, and conditions each retained pair only on its
own nonemptiness. The clause-count bound remains $m_0$.

Complementing either endpoint does not change whether a multiset is
empty. Its conditional distribution is consequently still invariant
under either complementation. Finally, put
$Z_{u,e}=\sum_{c\in M_e}s_u(c)$ and
$p_e=\Prb(M_e\ne\emptyset)$. An empty multiset contributes $0$ to
$Z_{u,e}^2$, so for every retained edge
\[
 \E[Z_{u,e}^2\mid M_e\ne\emptyset]
 =\frac{\E Z_{u,e}^2}{p_e}\ge\kappa.
\]
The unit clauses were fixed independently of the pair draws and remain
fixed under this conditioning. This proves the claim.
\end{proof}

For the probability estimates below, fix a nonempty observed graph
$G_F=H=(V_H,E_H)$ and put $r=|V_H|$. All probabilities and expectations
refer to the conditional model in Lemma~\ref{lem:observed_pairs}.
We identify the variables with their positions $1,\ldots,r$ in $U$
when writing edges and neighbors, and write $v\sim t$ when
$\{t,v\}\in E_H$.

The next lemma identifies the randomness available when a constraint is
created and when a later variable is assigned. It is the reason that
overlapping constraints can be analyzed without treating them as
independent. We use the original sampled clauses in this analysis, since
preprocessing preserves all gains.

Let $b_t$ be the number of original units $(x_t)$ minus the number of
units $(\neg x_t)$. For an edge $\{t,v\}$, define its contribution to
the gain at $t$ by
\[
 G_{tv}(x_v)=
 \sum_{\substack{c\in M_{\{t,v\}}\\
                 \text{the literal on }v\text{ is false}}}s_t(c).
\]
A clause contributes its sign at $t$ exactly when its partner literal
is false. Thus
\begin{align*}
 \mathrm{gain}_t(x)&=b_t+\sum_{v\sim t}G_{tv}(x_v),\\
 \Delta_t&=b_t+\sum_{\substack{v<t\\v\sim t}}G_{tv}(x_v)
       +\frac12\sum_{\substack{v>t\\v\sim t}}\sum_{c\in M_{\{t,v\}}}s_t(c).
\end{align*}
Write $d_t=\deg_H(\mathbf u_t)$.

\begin{lemma}[Revealing the endpoint signs]\label{lem:random_exposure}
The random formula and the independent uniform assignment can be exposed
in variable order so that, just before the signs at variable $t$ are
revealed,
\[
 \Delta_t=b_t+\sum_{v\sim t}\phi_{t,\{t,v\}}a_{tv},
 \qquad Y_t=\sum_{v\sim t}a_{tv}^2\le m_0^2d_t.
\]
The coefficients $a_{tv}$ are already known, and the signs
$\phi_{t,\{t,v\}}$ are independent and uniform on $\{-1,1\}$. After they
are revealed, $x_t$ is still uniform on $\{0,1\}$ and independent of
$\Delta_t$.
For each edge $\{t,v\}$ with $t<v$, the later deviation has the form
\[
 G_{tv}(x_v)-\frac12\sum_{c\in M_{\{t,v\}}}s_t(c)
   =\phi_{t,\{t,v\}}\eta_{tv}\theta_{v,\{t,v\}},
 \qquad |\eta_{tv}|\le m_0/2.
\]
Here $\eta_{tv}$ is known from the start. The signs
$\theta_{v,\{t,v\}}$ are independent and uniform on $\{-1,1\}$ when
endpoint $v$ is reached, conditional on all earlier exposed information.
\end{lemma}
\begin{proof}
For each edge $e$, first draw an independent base multiset $K_e$ with
the prescribed law of $M_e$. Write $\sigma_u(c)$ for an endpoint sign
in this base multiset. Independently draw a uniform sign $\phi_{u,e}\in\{-1,1\}$ for
each endpoint of each edge. Obtain $M_e$ by replacing its endpoint signs
by $\phi_{u,e}\sigma_u(c)$. Complementation invariance implies that the
result has the required distribution, independently on different edges.

Draw the uniform assignment independently and put
$\theta_{u,e}=\phi_{u,e}(2x_u-1)$. Conditional on any assignment $x$,
the transformation from the $\phi$ signs to the $\theta$ signs preserves
their independent uniform distribution. Consequently all the $\theta$ signs
are mutually independent and independent of $x$ and of the base multisets.
A partner literal on $v$ is false precisely when
$\sigma_v(c)\theta_{v,e}=-1$.

Reveal all $K_e$ first. At step $t$, reveal the signs $\phi_{t,e}$ on
edges incident to $t$, followed by $x_t$. Let $\mathcal H_t$ be the
information after this step. For $e=\{t,v\}$ set
\[
 a_{tv}=
 \begin{cases}
 \displaystyle\sum_{\substack{c\in K_e\\
                    \sigma_v(c)\theta_{v,e}=-1}}\sigma_t(c),&v<t,\\[6pt]
 \displaystyle\frac12\sum_{c\in K_e}\sigma_t(c),&v>t.
 \end{cases}
\]
These coefficients are determined by $\mathcal H_{t-1}$. Multiplying
them by $\phi_{t,e}$ gives precisely the earlier-pair contributions or
the average later-pair contributions to $\Delta_t$. This proves its
representation. The conditional variance of the signed sum is $Y_t$,
because the new signs are independent and have mean $0$. Also
$|a_{tv}|\le m_0$, proving the upper bound on $Y_t$.

For $t<v$ define
\[
 \eta_{tv}=
 \sum_{\substack{c\in K_e\\\sigma_v(c)=-1}}\sigma_t(c)
       -\frac12\sum_{c\in K_e}\sigma_t(c).
\]
This is half the difference of the signed sums over the two partner
polarities, so $|\eta_{tv}|\le |K_e|/2\le m_0/2$.
If $\theta_{v,e}=1$, the first sum selects the false partner literals;
if $\theta_{v,e}=-1$, their complementary set is selected.
Subtracting the average gives the asserted deviation with opposite signs
in these two cases. The $\theta$ signs at $v$ are jointly independent
and uniform conditional on $\mathcal H_{v-1}$, as required.
\end{proof}

\begin{samepage}
We will apply the following probability estimate both to sequential
choices and to functions of the independent edge data. Its interval
form gives the constants used in Step~2 of the main proof. All sampling
experiments to which we apply it have finite outcome spaces.

\begin{lemma}[Conditional ranges and bounded differences]\label{lem:random_ranges}
On a finite probability space, let $X_1,\ldots,X_k$ be exposed in order,
and let $\mathcal H_j$ denote the information after step $j$. Conditional on the earlier information,
suppose $X_j$ lies in an interval of length $\ell_j$,
where the interval is known before $X_j$ is revealed. If
$\sum_j\ell_j^2\le B$ on every path for a deterministic $B>0$, then,
for $a>0$,
\[
 \Prb\left(\sum_{j=1}^k
       \bigl(X_j-\E[X_j\mid\mathcal H_{j-1}]\bigr)<-a\right)
 \le e^{-2a^2/B}.
\]
In particular, if $f$ is a function of independent inputs and changing
input $j$ changes $f$ by at most $c_j$, with $\sum_jc_j^2>0$, then
\[
\Prb(f<\E f-a)\le e^{-2a^2/\sum_jc_j^2}.
\]
If all $c_j$ are $0$, $f$ is constant and this probability is $0$.
\end{lemma}
\end{samepage}
\begin{proof}
A centered random variable $Z$ supported on an interval of length $\ell$
satisfies
$\E e^{\lambda Z}\le e^{\lambda^2\ell^2/8}$.
To see this, let $\psi(\lambda)=\log\E e^{\lambda Z}$.
Then $\psi(0)=\psi'(0)=0$, and $\psi''(\lambda)$ is the variance under
the distribution that reweights each outcome by $e^{\lambda Z}$.
This distribution has the same support. Every random variable in an
interval of length $\ell$ has variance at most $\ell^2/4$: subtract the
midpoint, bound the second moment by $\ell^2/4$, and use that variance
is no larger than this second moment. Integrating the bound on
$\psi''$ twice proves the exponential estimate.

Write $Z_j=X_j-\E[X_j\mid\mathcal H_{j-1}]$. The interval length
$\ell_j$ may depend on the earlier observations, so we retain its actual
value along each path. For a fixed real $\lambda$, put
\[
 M_j=\exp\left(\lambda\sum_{i=1}^j Z_i
                    -\frac{\lambda^2}{8}\sum_{i=1}^j\ell_i^2\right),
 \qquad M_0=1.
\]
Both $M_{j-1}$ and $\ell_j$ are known before step $j$. Applying the
one-step estimate conditionally therefore gives
\[
 \E[M_j\mid\mathcal H_{j-1}]
 =M_{j-1}e^{-\lambda^2\ell_j^2/8}
       \E[e^{\lambda Z_j}\mid\mathcal H_{j-1}]
 \le M_{j-1}.
\]
Taking expectations successively yields $\E M_k\le1$. The pathwise
bound $\sum_j\ell_j^2\le B$ now implies
\[
 \E\exp\left(\lambda\sum_jZ_j\right)
 \le e^{\lambda^2B/8}\E M_k\le e^{\lambda^2B/8}.
\]
Use this estimate with $-\lambda$, where $\lambda>0$. Markov's
inequality gives
\[
 \Prb\left(\sum_jZ_j<-a\right)
 \le \exp(-\lambda a+\lambda^2B/8).
\]
Taking $\lambda=4a/B$ proves the first assertion.
For the second assertion, reveal the independent
inputs one at a time and consider the successive conditional expectations
of $f$. Their increments have conditional mean zero and conditional
range at most $c_j$. Indeed, two possible values of input $j$ change
$f$ by at most $c_j$ for every fixed choice of the remaining inputs,
and hence also after averaging those inputs. The first assertion applies.
\end{proof}

An average gain may include an arbitrary unit-clause contribution.
The next lemma shows that such a shift cannot remove the useful absolute
fluctuations of a symmetric signed sum. Its second assertion converts a
variance estimate into the average gap needed below.

\begin{lemma}[Absolute gains and second moments]\label{lem:random_moments}
Let $Z=\sum_i a_i\phi_i$ for independent uniform signs
$\phi_i\in\{-1,1\}$, and put
$Y=\sum_i a_i^2$. For every real $b$ and every $K>0$,
\[
 \E\min\{|b+Z|,K\}
 \ge\frac12\left(\sqrt{Y/3}-\frac{Y}{4K}\right).
\]
For any nonnegative random variable $U$ with $\E U>0$ and finite second
moment,
\[
 \E\sqrt U\ge\frac{(\E U)^{3/2}}{(\E U^2)^{1/2}}.
\]
\end{lemma}
\begin{proof}
Independence and zero means give $\E Z^2=Y$ and
\[
 \E Z^4=\sum_i a_i^4+6\sum_{i<j}a_i^2a_j^2\le3Y^2.
\]
We first explain the moment inequality that will be used twice in this
proof. Let $V$ be any nonnegative random variable with finite fourth
moment. The Cauchy--Schwarz inequality gives
\begin{align*}
 (\E V^2)^2
 &=\bigl(\E[\sqrt V\,(V\sqrt V)]\bigr)^2
 \le (\E V)(\E V^3),\\
 (\E V^3)^2
 &=\bigl(\E[V\,V^2]\bigr)^2
 \le (\E V^2)(\E V^4).
\end{align*}
Squaring the first inequality and using the second, we obtain
\[
 (\E V^2)^4\le (\E V)^2(\E V^2)(\E V^4).
\]
If $\E V^2>0$, division by $\E V^2$ gives
\[
 (\E V^2)^3\le (\E V)^2\E V^4.
\]
The same inequality holds when $\E V^2=0$, since its right-hand side
is nonnegative. Applying it with $V=|Z|$ yields
$Y^3\le3Y^2(\E|Z|)^2$. Consequently,
$\E|Z|\ge\sqrt{Y/3}$ when $Y>0$; this conclusion is also immediate
when $Y=0$.
Write $r_+=\max\{r,0\}$. Then $(z-K)_+\le z^2/(4K)$ for $z\ge0$,
because for $z>K$
the inequality is equivalent to $(z-2K)^2\ge0$.
It follows that
\[
 \E\min\{|Z|,K\}\ge\sqrt{Y/3}-Y/(4K).
\]
The triangle inequality gives $|b+z|+|b-z|\ge2|z|$.
After truncation this implies
$\min\{|b+z|,K\}+\min\{|b-z|,K\}\ge\min\{|z|,K\}$.
Symmetry of $Z$ makes the two expectations on the left equal,
proving the first assertion.

For the second assertion, apply the same moment inequality with
$V=\sqrt U$. Its fourth moment is $\E U^2$, which is finite by
assumption. We obtain
\[
 (\E U)^3\le(\E\sqrt U)^2\E U^2.
\]
Since $\E U>0$, this inequality also implies $\E U^2>0$.
Taking square roots and dividing by $\sqrt{\E U^2}$ proves the claim.
\end{proof}

We can now bound the accumulated requirements $A_T$ in the main proof.
There are two sources of variation: the requirement at an individual
branching step, and the amount of conditional variance available across
the selected variables. The two terms in the following bound account
for these separately.

\begin{lemma}[Accumulated requirements]\label{lem:random_requirements}
Use $A_T,T,r$ from the proof of Theorem~\ref{thm:random_main}, with
$c_1$ and the pair-law parameters from Theorem~\ref{thm:weighted}. Then
for all $a,b>0$,
\[
 \Prb\left(A_T<\frac{c_1T}{26}-a-b\right)
 \le e^{-2a^2/(Tm_0^2)}+e^{-169b^2/(m_0^2r)}.
\]
\end{lemma}
\begin{proof}
Use the exposure of Lemma~\ref{lem:random_exposure}.

\par\medskip\noindent\textbf{1. A lower bound at each step.}
Put
\[
 \widetilde X_t=[t\in I]\,\omega_t
            \min\{|\Delta_t|,m_0\sqrt{d_t}\}.
\]
Here $[t\in I]$ is $1$ when a constraint is created at step $t$,
and zero otherwise. Thus $0\le\widetilde X_t\le m_0$ and
$A_T\ge\sum_{t\le T}\widetilde X_t$.
Conditional on $\mathcal H_{t-1}$, the gain is the fixed shift $b_t$
plus a symmetric signed sum of variance $Y_t$.
The bit $x_t$ is independent of this gain and chooses the nonpreferred
value with probability $1/2$. Lemma~\ref{lem:random_moments} gives
\begin{align*}
 \E[\widetilde X_t\mid\mathcal H_{t-1}]
 &\ge\frac{\omega_t}{4}
       \left(\sqrt{Y_t/3}-\frac{Y_t}{4m_0\sqrt{d_t}}\right)\\
 &\ge\frac{\omega_t\sqrt{Y_t}}4
             \left(\frac1{\sqrt3}-\frac14\right)
 \ge \mu_t:=\frac{\omega_t\sqrt{Y_t}}{13}.
\end{align*}
The middle inequality uses $Y_t\le m_0^2d_t$.
Lemma~\ref{lem:random_ranges}, with interval lengths $m_0$, therefore gives
\[
 \Prb\left(A_T<\sum_{t\le T}\mu_t-a\right)
 \le e^{-2a^2/(Tm_0^2)}.
\]

\par\medskip\noindent\textbf{2. The average available variance.}
For an earlier neighbor $v<t$, condition on the base multiset and let
$q_+$ and $q_-$ be its signed sums at $t$ over the two polarities at $v$.
The coefficient $a_{tv}$ takes these two values with equal probabilities.
Thus
\[
 \E[a_{tv}^2\mid K_{\{t,v\}}]
 =\frac{q_+^2+q_-^2}{2}
 \ge\frac{(q_++q_-)^2}{4}.
\]
Averaging and applying the model's third condition gives
$\E a_{tv}^2\ge\kappa/4$. For a later neighbor, $a_{tv}$ is half
the endpoint signed sum, giving the same lower bound directly.
Summing these lower bounds gives
\[
 \E Y_t\ge\frac{\kappa d_t}{4}.
\]
We can now use the deterministic upper bound $Y_t\le m_0^2d_t$
from Lemma~\ref{lem:random_exposure}. Since $d_t\ge1$ and $m_0>0$,
\[
 Y_t=(\sqrt{Y_t})^2\le m_0\sqrt{d_t}\sqrt{Y_t},
 \qquad
 \sqrt{Y_t}\ge\frac{Y_t}{m_0\sqrt{d_t}}.
\]
Taking expectations therefore gives the direct estimate
\[
 \E\sqrt{Y_t}\ge\frac{\E Y_t}{m_0\sqrt{d_t}}
       \ge\frac{\kappa}{4m_0}\sqrt{d_t}
       \ge c_1\sqrt{d_t}.
\]
For the last inequality, the constant already stated in the theorem
satisfies
\[
 c_1^2=\frac{\kappa^3}{4(\kappa^2+4m_0^4)}
       \le\frac{\kappa^2}{16m_0^2},
\]
because $(\kappa-2m_0^2)^2\ge0$. Both sides of the resulting comparison
$c_1\le\kappa/(4m_0)$ are positive. This argument needs no variance
estimate for $Y_t$ and no independence among its summands.
The lower weight bound~\eqref{eq:random_weight_bounds} now implies
$\E\sum_{t\le T}\mu_t\ge c_1T/26$.

\par\medskip\noindent\textbf{3. Concentration over the edge data.}
Regard each base multiset and its two $\theta$ signs as one independent
edge input. Changing the input on $\{u,v\}$ affects only the terms
$\mu_u,\mu_v$ that belong to the prefix. In each affected square root,
only one squared coefficient changes, with absolute coefficient in
$[0,m_0]$. For $R\ge0$ and $a,b\in[0,m_0]$,
$|\sqrt{R+a^2}-\sqrt{R+b^2}|\le|a-b|\le m_0$.
Thus an edge changes the sum by at most
$m_0(\omega_u+\omega_v)/13$.
The sum of squared change bounds is at most
\[
 \frac{m_0^2}{169}\sum_{\{u,v\}\in E_H}(\omega_u+\omega_v)^2
 \le\frac{2m_0^2}{169}\sum_v\omega_v^2d_v
 \le\frac{2m_0^2r}{169}.
\]
Lemma~\ref{lem:random_ranges} gives
\[
 \Prb\left(\sum_{t\le T}\mu_t<c_1T/26-b\right)
 \le e^{-169b^2/(m_0^2r)}.
\]
Combining this with the bound from the first step proves the lemma.
\end{proof}

The last estimate controls how much the later values can help satisfy
the accumulated constraints. A shared variable may affect many constraints,
but the relevant endpoint signs on its different edges are still
independent when that variable is reached. This is the estimate for
$Q_T$ used in Step~2 of the main proof.

\begin{lemma}[Favorable contributions from later values]\label{lem:random_compensation}
For every $c>0$,
\[
 \Prb(Q_T<-c)\le e^{-2c^2/(m_0^2T)}.
\]
\end{lemma}
\begin{proof}
By Lemma~\ref{lem:random_exposure}, the contribution from $\{t,v\}$,
where $t\le T$ and $t<v$, is
\[
 [t\in I]\,\omega_t(2p_t-1)
      \phi_{t,\{t,v\}}\eta_{tv}\theta_{v,\{t,v\}}.
\]
Group these terms by their later endpoint $v$. Write
\[
 \beta_{tv}=[t\in I]\,\omega_t(2p_t-1)\phi_{t,\{t,v\}}\eta_{tv},
 \qquad
 R_v=\sum_{\substack{t\le T,\ t<v\\\{t,v\}\in E_H}}
          \beta_{tv}\theta_{v,\{t,v\}},
\]
so that $Q_T=\sum_v R_v$. Each coefficient $\beta_{tv}$ is determined by
$\mathcal H_{v-1}$: the comparison at $t$ has already been made and its
endpoint sign has already been exposed. Conditional on this history,
the signs $\theta_{v,\{t,v\}}$ are jointly independent and uniform.

For a fixed coefficient $a$, the variable $a\theta$ is centered and
lies in an interval of length $2|a|$. The one-step exponential estimate
in the proof of Lemma~\ref{lem:random_ranges} gives
$\E e^{\lambda a\theta}\le e^{\lambda^2a^2/2}$ for every real
$\lambda$. Multiplying these estimates for the independent signs at
$v$, and using $|\beta_{tv}|\le\omega_t m_0/2$, yields
\[
 \E[e^{\lambda R_v}\mid\mathcal H_{v-1}]
 \le \exp\left(\frac{\lambda^2}{2}\sum_t \beta_{tv}^2\right)
 \le e^{\lambda^2 B_v/8},
 \qquad
 B_v=m_0^2\sum_{\substack{t\le T,\ t<v\\\{t,v\}\in E_H}}\omega_t^2.
\]
The quantities $B_v$ are deterministic after the observed graph is fixed.
Every edge in this sum is charged once to its earlier endpoint. Hence
\[
 \sum_v B_v
 =m_0^2\sum_{t\le T}\omega_t^2
        |\{v>t:\{t,v\}\in E_H\}|
 \le m_0^2\sum_{t\le T}\omega_t^2d_t
 \le m_0^2T.
\]
Taking conditional expectations successively now gives
$\E e^{\lambda Q_T}\le e^{\lambda^2m_0^2T/8}$.
For $\lambda>0$, Markov's inequality applied with $-\lambda$ gives
\[
 \Prb(Q_T<-c)\le e^{-\lambda c+\lambda^2m_0^2T/8}.
\]
Taking $\lambda=4c/(m_0^2T)$ proves the claim. The argument does not
require independence between $A_T$ and $Q_T$.
\end{proof}

\subsection{Completing the proof for semirandom formulas}\label{app:random_main_completion}

This subsection supplies the parts of the proof of
Theorem~\ref{thm:random_main} that were only summarized in the main text.

\paragraph{The two models and the constants.}
Both models are covered by the supporting lemmas with $m_0=4$ and
$\kappa=3/4$. Here $m_0$ bounds the number of clauses on a pair, and
$\kappa$ bounds below the second moment of the signed occurrence count
at either endpoint. One uniformly signed clause has second moment one.
For independent inclusion of four clauses with probability $1/4$, it is
$4(1/4)(3/4)=3/4$. Complementing either endpoint preserves each law.
Set
\[
 c_1=\frac{\sqrt\kappa}{2\sqrt{1+4m_0^4/\kappa^2}},
 \qquad c_0=\frac{c_1^2}{10816m_0^2}.
\]
These are fixed positive constants for the two stated models.

\paragraph{The calculation in Step 2.}
Lemma~\ref{lem:random_requirements} gives, for any $a,b>0$,
\[
 \Prb\left(A_T<\frac{c_1T}{26}-a-b\right)
 \le e^{-2a^2/(Tm_0^2)}+e^{-169b^2/(m_0^2r)}.
\]
Lemma~\ref{lem:random_compensation} gives, for any $c>0$,
\[
 \Prb(Q_T<-c)\le e^{-2c^2/(m_0^2T)}.
\]
Take $a=b=c=c_1T/104$. Outside these three exceptional events,
\[
 A_T+Q_T\ge\frac{c_1T}{26}
              -3\frac{c_1T}{104}
          =\frac{c_1T}{104}>0.
\]
Thus the summed constraint fails. The first and third exponents equal
$c_1^2T/(5408m_0^2)\ge c_0r$.
The second is
$169c_1^2T^2/(10816m_0^2r)\ge c_0r$, since $T\ge r/2$.
Consequently
\[
 \Prb_{F,x}(A_T+Q_T\le0\mid G_F=H)\le3e^{-c_0r}.
\]

\paragraph{The details of Step 3.}
For a prefix $\alpha$, let $N(\alpha)$ count suffix assignments satisfying
its sum, and put $W(F)=\sum_{\alpha\in\bits^T}N(\alpha)$.
Here we define the sum by constructing the comparisons along the whole
prefix, even if the algorithm would reject that prefix earlier. Every
completion actually enumerated satisfies its whole-prefix sum, and distinct
search paths produce distinct completions. Thus $W(F)$ is an upper bound
on the number enumerated; equality is unnecessary. By
Lemma~\ref{lem:enumeration}, the number of recursive calls, including
calls that immediately reject their assignment, is at most
\[
 \operatorname{Nodes}(F)\le 2^{T+1}-1+2rW(F).
\]
The inactive variables have already been fixed during initialization.
Thus $W(F)$ counts each assignment to the $r$ active variables exactly
once. Equivalently, we may keep the original labels and sample all $n$
bits uniformly: the comparisons depend only on the active bits, and
their probability of being satisfied is unchanged. The factor converting
this probability to the candidate count is still $2^r$, not $2^n$.
Consequently
\[
 \E[W(F)\mid G_F=H]
 =2^r\Prb_{F,x}(A_T+Q_T\le0\mid G_F=H)
 \le3\cdot2^r e^{-c_0r}.
\]
In particular, on this fixed nonempty observed graph,
\[
 \E[\operatorname{Nodes}(F)\mid G_F=H]
 \le 2^{T+1}+6r\,2^r e^{-c_0r}.
\]
If $H$ is empty, the algorithm finishes in polynomial time; define
$W(F)=1$ in this case, so the same bound holds with $r=0$.
Since $\kappa\le m_0^2$, the definition of $c_1$ gives
$c_1\le\kappa/(4m_0)\le m_0/4$. Hence
$c_0\le1/173056<(\ln2)/2$.
The quantity $2^r e^{-c_0r}$ is therefore increasing in $r$.
Using $r\le n$ and averaging over $G_F$ gives
\[
 \E W(F)\le3\cdot2^n e^{-c_0n}.
\]
Lemma~\ref{lem:enumeration} bounds the search work by
$\Ostar(2^T+W(F))$. Since $2^T\le2^{\lceil n/2\rceil}$ and $c_0<(\ln2)/2$, this prefix
cost is absorbed by $\Ostar(2^ne^{-c_0n})$, proving the expectation bound.
Markov's inequality applied to the unconditional expectation also gives
\[
 \Prb_F\left(W(F)>2^n e^{-c_0n/2}\right)
 \le3e^{-c_0n/2}.
\]
The prefix cost $2^T$ is smaller than the claimed high-probability time
bound, up to a constant factor, because $c_0<\ln2$.
This completes the proof.

\subsection{Implementation details of the algorithm}\label{app:random_implementation}

This subsection specifies the bookkeeping of
Algorithms~\ref{alg:weighted} and~\ref{alg:random_search}, which
Section~\ref{sec:random_algorithm} describes only in outline.
The observed graph $G_F$ has no isolated vertices, since its vertices are
exactly the endpoints of its edges, and it remains fixed throughout the
search.

The two procedures use the following data. Shared variables are initialized
once by \textnormal{\textsc{RandomFormula}} and are available to every call
of \textnormal{\textsc{RandomSearch}}.
\begin{description}
\item[Fixed shared data.]
$U=(\mathbf u_1,\ldots,\mathbf u_r)$ is the variable order, and
$(\omega_1,\ldots,\omega_r)$ contains the constraint weights.
Neither changes during the search. All variables remain in $U$, including
variables that disappear from the residual formula.
\item[Shared best solution.]
$M$ is the largest score in $F_{\rm core}$ found so far, and
$\alpha^\star$ assigns all variables of $V_F$ and attains that score.
The outer procedure adds $s_0$ and $\alpha_0$ when returning the solution
for the original formula. Initially
$M=-\infty$ and $\alpha^\star=\emptyset$, indicating that no complete
assignment has yet been evaluated. Every recursive call may update this
same pair. An improvement is retained when the call returns and when
the search backtracks to explore another branch.
\item[Parameters of a recursive call.]
$F$ is the residual soft formula, $\mathcal C$ is the collection of
pruning constraints with their weights, and $\alpha$ is the current
partial assignment. The value $s$ is the score offset accumulated through
preprocessing and assignments: for any completion $y$, its score in
the formula $F_{\rm core}$ is $s$ plus the number of clauses it satisfies in
$F$. These four parameters belong to the current node and are copied
separately for its children. The parameter $T$ is the stopping depth;
it is passed unchanged to every recursive call.
\end{description}
At each call, $\alpha$ assigns exactly the first $|\alpha|$ variables of
$U$. In particular, $M$ and $\alpha^\star$ are shared across branches,
whereas $s$ and $\alpha$ describe only the current branch.
The identities used by \textnormal{\textsc{Preprocess}} apply to copies
of clauses as well. Its score guarantees therefore also apply to the
bounded-multiset extension in Appendix~\ref{app:random_signs}.

After initializing this shared data, apply
\textnormal{\textsc{Preprocess}} to $F_{\rm core}$ and start
\textnormal{\textsc{RandomSearch}} with the resulting formula and score
offset, an empty assignment, and an empty constraint collection.
The recursive search uses substitution, cancellation, and rejection,
but does not call the full \textnormal{\textsc{Simplify}} procedure of
Appendix~\ref{app:simplify} or
make forced assignments. Thus every call at depth $t=|\alpha|$ has
assigned exactly the first $t$ variables of $U$. Created constraints
remain in $\mathcal C$ with their updated bounds, even if they have
become automatically satisfied.

\section{Analysis of accumulated comparisons}\label{sec:counts}\label{app:accumulation}

This appendix supplies the count and projection analysis used in
Sections~\ref{sec:maxcut} and~\ref{sec:sbm}. The direct argument for
semirandom formulas in Section~\ref{sec:random} has its own estimates.
It also specifies the projection comparison routine used for hidden
communities. A few aggregate guesses account for the systematic part of the pair
coefficients; the remaining comparisons accumulate enough gap while a norm bound
controls their combined fluctuation. The matrix formulation is proof
machinery for these input families, rather than an additional class
definition required of the reader.

\subsection{One comparison algorithm}

The comparison routine receives the objective~\eqref{eq:quadratic},
a rational representation
$P=\sum_{k=1}^r\lambda_kv_kv_k^{\mathsf T}$ of polynomial encoding length,
and finite rational grids $\mathcal Y_k$ with stated errors $\delta_k\ge0$.
Each grid must approximate every $\langle v_k,x\rangle$ to within
$\delta_k$. Fix the variable order and partition its first
$T=\lfloor\tau n\rfloor$ variables into consecutive rounds
$B_1,\ldots,B_\ell$, where $0<\tau<1$. Treat the suffix as round
$\ell+1$ and write $\operatorname{rd}(t)$ for the index of the round containing $t$.

Compute $N=\operatorname{offdiag}(J-P)$ and the fixed allowances
\[
 s_t=\sum_k|\lambda_k|\,|v_k(t)|(\delta_k+|v_k(t)|),
 \qquad \mathcal Y=\prod_k\mathcal Y_k.
\]
Initialize the best original score and assignment. Process the tuples
$\hat y\in\mathcal Y$ consecutively; the grid description and each point
have polynomial size, and the tuples are streamed. For the current
tuple, define
\[
 a^{\hat y}=a+\sum_k\lambda_k\hat y_kv_k,
 \qquad F'_{\hat y}(x)=c+\langle a^{\hat y},x\rangle+
                    \sum_{u<v}N_{uv}x_ux_v.
\]
For each tuple, enumerate all assignments to the prefix and process
the rounds of the current prefix in order. For $t\in B_j$, compute
\[
 \Delta_t=a^{\hat y}_t+\sum_{w:\operatorname{rd}(w)<j}N_{tw}x_w
             +\frac12\sum_{w\ne t:\operatorname{rd}(w)\ge j}N_{tw}.
\]
Set $p_t=1$ if $\Delta_t\ge0$ and $p_t=0$ otherwise.
This preference uses only earlier rounds. Call $t$ a creator when
$x_t\ne p_t$, and let $\gamma_t(x)$ be the gain under $F'_{\hat y}$
of flipping it to $p_t$. Accumulate the comparisons
$\gamma_t(x)\le s_t$ for the creators of this prefix.
After all prefix rounds, enumerate the suffixes satisfying
\begin{equation}
 \sum_{t\text{ creator}}(\gamma_t(x)-s_t)\le0. \label{eq:counts-sum}
\end{equation}
This is one affine inequality on the suffix, so Lemma~\ref{lem:enumeration}
applies. Evaluate the original objective on every candidate. Candidates
need not have projections near the guessed grid point. Retain the best
score and assignment across all prefixes and tuples, and return them
after all runs finish.

\begin{lemma}[Sound comparisons with approximate projections]\label{lem:counts-exact}
Every optimum is enumerated in a run that approximates its projections.
The algorithm is exact for every supplied representation and round partition.
\end{lemma}
\begin{proof}
Fix an optimal complete assignment $x^*$, and put
$y_k=\langle v_k,x^*\rangle$. By the coverage requirement on the grids,
there is a tuple $\hat y$ with $|y_k-\hat y_k|\le\delta_k$ for every $k$.
Consider this run and the prefix of $x^*$.

For a variable $t$, write $G_t(x)$ for the original $0\to1$ score gain
and $H_t(x)$ for the $0\to1$ gain under the auxiliary objective
$F'_{\hat y}$. Subtracting the two expressions for the gain gives
\[
 G_t(x^*)-H_t(x^*)
 =\sum_k\lambda_kv_k(t)(y_k-\hat y_k-v_k(t)x_t^*).
\]
The term $v_k(t)x_t^*$ is necessary because a variable's gain excludes
its own contribution to the projection. Since $x_t^*\in\{0,1\}$,
the triangle inequality gives
\begin{align*}
 |G_t(x^*)-H_t(x^*)|
 &\le\sum_k|\lambda_k|\,|v_k(t)|
       \bigl(|y_k-\hat y_k|+|v_k(t)|x_t^*\bigr)\\
 &\le\sum_k|\lambda_k|\,|v_k(t)|
       (\delta_k+|v_k(t)|)=s_t.
\end{align*}

If $t$ is a creator, then $x_t^*\ne p_t$. Flipping $t$ to $p_t$
has original gain $(2p_t-1)G_t(x^*)$, which is at most $0$ by optimality.
Its auxiliary gain therefore satisfies
\[
 \gamma_t(x^*)=(2p_t-1)H_t(x^*)
 \le (2p_t-1)G_t(x^*)+|H_t(x^*)-G_t(x^*)|
 \le s_t.
\]
Thus every comparison for this prefix holds, and their sum satisfies
\eqref{eq:counts-sum}. The suffix enumeration visits $x^*$.
The argument applies to every optimum, and weak inequalities retain ties.
Finally, the algorithm scores each visited assignment in the original
objective and keeps a maximum. Since an optimum is visited, the returned
assignment and score are optimal. No test that a candidate matches the
run's guessed projections is needed.
\end{proof}

The accumulation proof below uses only nonnegative comparison allowances,
their total size, and the fact that an optimum survives in at least one
run. Thus a smaller allowance may replace $s_t$ if its validity is proved
directly. The same proof allows a joint finite list of aggregate guesses,
rather than a Cartesian product of independent grids. This includes guessing the counts of
ones in known communities, as sketched in Section~\ref{sec:sbm_intuition}; the running bound uses
the number of joint guesses.

\subsection{A general accumulation theorem}

Let $D_N=n^{-1}\sum_{u\ne v}N_{uv}^2>0$ and
$W_t=\sum_{w:\operatorname{rd}(w)<\operatorname{rd}(t)}N_{tw}^2$. The common assumptions are
\begin{equation}
 \sum_{t\le T}\sqrt{W_t}\ge4\sqrt3\,g n\sqrt{D_N},\qquad
 \sum_{t\le T}s_t\le\tfrac18g n\sqrt{D_N}.       \label{eq:profile-slack}
\end{equation}
We call the first inequality the \emph{gap bound} and the second the
\emph{allowance bound}. All constants below are independent of $n$.

\begin{theorem}[Accumulation after projection guessing]\label{res:counts}
Assume~\eqref{eq:profile-slack}, with $0<\tau<1$ and $g,K>0$.
Either of the following conditions gives a deterministic exact algorithm
in polynomial space and time
\[
 \Ostar\bigl(|\mathcal Y|(2^{\tau n}+2^ne^{-c_*n})\bigr).
\]
\begin{enumerate}
\item \textbf{Sequential comparisons.} Each prefix round is a singleton.
Let $R_{tw}=N_{tw}[w>t]$ for $t\le T$, and assume $\|R\|^2\le KD_N$.
Set $K_L=K$ and
$c_*=g^2/(64\tau\alpha^2K)$.
\item \textbf{A fixed number of rounds.} There are $\ell\ge2$ prefix
rounds, with $\ell$ fixed, and $\|N\|^2\le KD_N$.
Set $K_L=(\ell-1)K$ and
\[
 c_*=\min\left\{
 \frac{g^2}{32\tau\alpha^2K_L},\,
 \frac{g^2}{8192\ell^2\tau K}\right\}.
\]
\end{enumerate}
In both cases, $\alpha=\max(2,9\tau\sqrt{K_L}/(4g))$ and
$n\ge256\pi\tau K_L/g^2$. Linear coefficients and units are arbitrary.
For polynomially many grid points this is a fixed base below two.
\end{theorem}

Exact count guessing is a special case: for
$v_k\in\{-1,0,1\}^n$, use all integer counts and $\delta_k=0$.
Then $|\mathcal Y|\le(2n+1)^r$ and $s_t\le\sum_k|\lambda_k|$.
Empirical mean centering takes $r=1$, $v_1=\ind$,
$\lambda_1=\binom n2^{-1}\sum_{u<v}J_{uv}$, and only $n+1$ counts.
If $P=0$, there is one run and no comparison allowance.

The next two estimates supply the concentration needed in the proof.
They concern uniform assignments to a fixed instance; neither estimate
introduces random choices into the algorithm. We include their proofs to
make explicit which properties of the constraints are being used.

\begin{lemma}[Two finite-sign concentration estimates]\label{lem:counts-concentration}
Let $X$ be uniform on $\bits^m$, and let $X^{(i)}$ denote $X$ with
coordinate $i$ flipped.
\begin{enumerate}
\item If $f:\bits^m\to\mathbb R$ and $B>0$ satisfy
$\sum_i(f(x)-f(x^{(i)}))^2\le B$ for every $x$, then, for $a>0$,
\[
 \Prb(f(X)<\E f(X)-a)\le e^{-a^2/(2B)}.
\]
\item Let $W$ be a real $m\times m$ matrix with zero diagonal,
$\|W\|^2\le V$, and $V\ge0$. If $B>0$ and $mV\le B$, then for
independent uniform signs $u_i\in\{-1,1\}$ and $e>0$,
\[
 \Prb\left(-\tfrac14u^{\mathsf T}Wu<-e\right)
 \le e^{-e^2/(2B)}.
\]
\end{enumerate}
\end{lemma}
\begin{proof}
\par\medskip\noindent\textbf{A bound from the changes under single-coordinate flips.}
For a positive function $h$, write
$\operatorname{Ent}(h)=\E(h\log h)-(\E h)\log(\E h)$.
We first prove
\begin{equation}
 \operatorname{Ent}(h^2)
 \le\E\sum_i(h(X)-h(X^{(i)}))^2.             \label{eq:cube-entropy}
\end{equation}
For one bit, put $a=h(0)$, $b=h(1)$, and $q=(a^2+b^2)/2$.
Using $\log z\le z-1$ in the two terms of the entropy gives
\[
 \operatorname{Ent}(h^2)
 \le\frac{(a^2-b^2)^2}{4q}\le(a-b)^2.
\]
The last inequality follows from $(a+b)^2\le2(a^2+b^2)=4q$.
For the induction step, separate the first bit and write
$h_0(y)=h(0,y)$, $h_1(y)=h(1,y)$, and
$r(y)=\sqrt{(h_0(y)^2+h_1(y)^2)/2}$.
Expanding the entropy gives the exact identity
\[
 \operatorname{Ent}(h^2)
 =\operatorname{Ent}(r^2)+
   \E_y\operatorname{Ent}_{b\in\bits}(h_b(y)^2).
\]
Apply the one-bit bound to the second term and the induction hypothesis
to the first. The reverse triangle inequality for the Euclidean norm gives
\[
 (r(y)-r(y^{(i)}))^2\le
 \tfrac12\sum_{b\in\bits}(h_b(y)-h_b(y^{(i)}))^2.
\]
Adding these bounds proves~\eqref{eq:cube-entropy} in every dimension.

Apply~\eqref{eq:cube-entropy} to $h=e^{tf/2}$. The elementary inequality
$(e^a-e^b)^2\le(a-b)^2(e^{2a}+e^{2b})$, together with invariance of the
uniform distribution under each flip, gives
\[
 \operatorname{Ent}(e^{tf})
 \le\frac{t^2}{2}\E\left[e^{tf(X)}
             \sum_i(f(X)-f(X^{(i)}))^2\right]
 \le\frac{t^2B}{2}\E e^{tf(X)}.
\]
Let $\psi(t)=\log\E e^{tf(X)}$. These are finite sums, so differentiation
is valid and the preceding inequality says
$t\psi'(t)-\psi(t)\le t^2B/2$.
For $t>0$, integrate $(\psi(t)/t)'\le B/2$ from $0$, using
$\lim_{t\downarrow0}\psi(t)/t=\E f(X)$. It follows that
$\E e^{t(f-\E f)}\le e^{t^2B/2}$.
Applying this also to $-f$, followed by the exponential Markov inequality
with $t=a/B$, proves the first assertion.

\par\medskip\noindent\textbf{A quadratic bound from a random partition.}
Choose a subset $S\subseteq[m]$ by including each index independently
with probability $1/2$, solely for this proof. For fixed $u$, let
$Z_S(u)=\sum_{i\in S,j\notin S}W_{ij}u_i u_j$.
The zero diagonal ensures
$\E_S Z_S(u)=u^{\mathsf T}Wu/4$. Convexity of the exponential therefore gives
\[
 e^{t u^{\mathsf T}Wu/4}\le\E_S e^{tZ_S(u)}.
\]
For fixed $S$, condition on the signs in $S$. The remaining expression
is a linear sum of independent signs with coefficients
$a_j=\sum_{i\in S}W_{ij}u_i$ for $j\notin S$.
Its squared coefficient sum is at most
$\|W^{\mathsf T}z\|_2^2\le V\|z\|_2^2\le Vm$, where
$z_i=u_i[i\in S]$. The usual $\cosh$ bound thus gives
$\E e^{tZ_S(u)}\le e^{t^2Vm/2}$.
Average over $S$ and use $Vm\le B$. Exponential Markov with $t=e/B$
now proves the second assertion. Symmetry of $W$ is not needed.
\end{proof}

We next apply the first estimate to the accumulated gaps. Absolute values
have corners, so we replace them temporarily by smooth upper approximations.
This modification is only an analysis device: the algorithm still uses
the original comparisons and their original gaps.

\begin{lemma}[A common gap estimate]\label{lem:counts-profile}
Let $L$ be the $T\times T$ matrix with $L_{tw}=N_{tw}$ when
$\operatorname{rd}(w)<\operatorname{rd}(t)$ and $0$ otherwise. Suppose $\|L\|^2\le K_LD_N$ and
\eqref{eq:profile-slack}'s gap bound holds. With
$M=\alpha\sqrt{K_LD_N}$ and the parameters of Theorem~\ref{res:counts},
\[
 \Prb\left(\sum_{t\le T}\min\{|\Delta_t|,M\}
               <\tfrac34gn\sqrt{D_N}\right)
 \le e^{-g^2n/(8\tau\alpha^2K_L)}.
\]
The probability is over a uniform Boolean assignment, with the formula fixed.
\end{lemma}
\begin{proof}
Put $A=gn\sqrt{D_N}$, $V=K_LD_N$, $u=2x-\ind$, and
$H_t=a^{\hat y}_t+\frac12\sum_{w\ne t}N_{tw}$.
Write $z=Lu$, with $u$ restricted to the prefix, so
$\Delta_t=H_t+z_t/2$. The norm bound gives $\sum_tz_t^2\le VT$.
Partition the rows into $\mathcal S_0=\{t:|H_t|\le M/2\}$ and its complement $\mathcal S_1$.
For $t\in \mathcal S_0$, the triangle inequality and $(|z_t|-2M)^2\ge0$ give
\[
 |H_t+z_t/2|-M\le |z_t|/2-M/2\le z_t^2/(8M),
\]
so clipping at $M$ loses at most $z_t^2/(8M)$.
For $t\in \mathcal S_1$, if $|z_t|\le M/2$ then $|\Delta_t|\ge M/4$;
otherwise $M/4-z_t^2/M<0$. In either case,
$\min\{|\Delta_t|,M\}\ge M/4-z_t^2/M$.
Consequently, with
\[
 G(u)=\sum_{t\in \mathcal S_0}|H_t+z_t/2|+\frac M4|\mathcal S_1|,
\]
we have
\begin{equation}
 \sum_{t\le T}\min\{|\Delta_t|,M\}\ge G(u)-VT/M.
                                                        \label{eq:gap-clipping}
\end{equation}
A symmetric sign sum of variance $W_t$ has first absolute moment at least
$\sqrt{W_t/3}$, and a deterministic shift cannot decrease its expected
absolute value. Thus $\E|\Delta_t|\ge\frac12\sqrt{W_t/3}$.
Each $W_t\le V$, and $\alpha\ge2$ gives
$M/4\ge\frac12\sqrt{W_t/3}$. The gap hypothesis therefore implies
$\E G\ge2A$.

Set $\delta=\sqrt V/\alpha$, and define
\[
 h_\delta(z)=
 \begin{cases}(z^2+\delta^2)/(2\delta),&|z|\le\delta,\\
 |z|,&|z|>\delta.
 \end{cases}
\]
Its derivative lies in $[-1,1]$, is $1/\delta$-Lipschitz, and satisfies
\[
 |z|\le h_\delta(z)\le |z|+\delta/2,\qquad
 0\le h_\delta(y)-h_\delta(z)-h_\delta'(z)(y-z)
       \le (y-z)^2/(2\delta).
\]
Replace the absolute values in $G$ by $h_\delta$, obtaining a function $f$.
Then $G\le f\le G+T\delta/2$ and $\E f\ge2A$.

We bound the total squared change of $f$ when individual bits are flipped.
At a fixed assignment, let $d_t=h_\delta'(\Delta_t)$ for $t\in \mathcal S_0$
and $d_t=0$ otherwise, and put $C_i=\sum_tL_{ti}^2$.
Flipping bit $i$ changes $\Delta_t$ by $-u_iL_{ti}$, so the preceding
remainder bound gives
\[
 f(x^{(i)})-f(x)=-u_i(L^{\mathsf T}d)_i+r_i,
 \qquad 0\le r_i\le C_i/(2\delta).
\]
Since $\sum_td_t^2\le T$, we have
$\|L^{\mathsf T}d\|_2^2\le VT$.
Also $C_i\le V$ and $\sum_iC_i\le TV$, so
$\sum_iC_i^2\le TV^2$.
Using $(a+b)^2\le2a^2+2b^2$ now gives
\[
 \sum_i(f(x^{(i)})-f(x))^2
 \le2TV+\frac{TV^2}{2\delta^2}
 =TV(2+\alpha^2/2)\le\alpha^2TV.
\]
The definition of $\alpha$ and $T\le\tau n$ imply
$2T\sqrt V\le\alpha A$. Hence $VT/M\le A/2$ and
$T\delta/2\le A/4$.
By~\eqref{eq:gap-clipping}, an assignment with clipped gap below $3A/4$
has $f<3A/2\le\E f-A/2$.
The first part of Lemma~\ref{lem:counts-concentration} therefore bounds
its probability by
\[
 \exp\left(-\frac{A^2}{8\alpha^2TV}\right)
 \le\exp\left(-\frac{g^2n}{8\tau\alpha^2K_L}\right),
\]
as required.
\end{proof}

\begin{proof}[Proof of Theorem~\ref{res:counts}]
Correctness is Lemma~\ref{lem:counts-exact}. Fix a run and a uniform $x$,
and write $u=2x-\ind$ and $I$ for its creators. In the sequential case,
$L$ is the transpose of the leading square block of $R$, so $\|L\|\le\|R\|$.
In the round case, each row block is $N_{B_j\times B_{<j}}$, whence
$\|L\|^2\le(\ell-1)\|N\|^2$. The common gap estimate applies in both cases.

\paragraph{Accumulated gap.}
Given the earlier variables, each creator indicator is uniform on $\{0,1\}$ and
$\Delta_t$ is fixed; within a round all preferences were fixed before it
started. For $Y_t=[t\in I]\min(|\Delta_t|,M)$, Azuma--Hoeffding gives
\[
 \Prb\left(\sum_tY_t<\tfrac12\sum_{t\le T}\min(|\Delta_t|,M)-a\right)
 \le e^{-2a^2/(TM^2)}.
\]
Use $a=gn\sqrt{D_N}/8$. The gap lemma and the allowance bound imply
$S'=\sum_{t\in I}(|\Delta_t|-s_t)\ge gn\sqrt{D_N}/8$
outside these two exceptional events.

\paragraph{Sequential comparisons.}
Here the sum tested by the algorithm is
\begin{equation}
 S'+Q,\qquad Q=-\tfrac12\sum_wu_w\sum_{t\in I,t<w}u_tN_{tw}.
                                                        \label{eq:counts-race}
\end{equation}
The inner coefficient is predictable, and its squared sum is at most
$\|R^{\mathsf T}(u_t[t\in I])_{t\le T}\|^2\le KD_NT$.
The conditional $\cosh$ bound gives
$\Prb(Q\le-b)\le e^{-2b^2/(KD_NT)}$.
For $b=gn\sqrt{D_N}/16$, a passing assignment belongs to an exceptional
event. Their total probability is at most $3e^{-c_*n}$ with the stated
sequential constant.

\paragraph{Comparisons within rounds.}
The tested sum is $S'+Q_{\rm later}+\sum_{j\le\ell}Q_j$, where
\[
 Q_{\rm later}=-\tfrac12\sum_wu_w
        \sum_{t\in I:\operatorname{rd}(t)<\operatorname{rd}(w)}u_tN_{tw},\qquad
 Q_j=-\tfrac12\sum_{t\in I\cap B_j}\sum_{w\in B_j\setminus\{t\}}u_tN_{tw}u_w.
\]
To bound $Q_{\rm later}$, expose the variables one at a time in their
fixed order. Before exposing $u_w$, its inner coefficient is already
known: it involves only creators in strictly earlier rounds, whose
preferences were fixed before those rounds began. Thus the same
predictable-sign estimate used for sequential comparisons applies.
For a destination block $B_j$, let $z_t=u_t$ when $t$ is a creator in an
earlier round, and let $z_t=0$ otherwise. Its vector of incoming
coefficients is the restriction of $N^{\mathsf T}z$ to $B_j$. Since
$\|z\|^2\le T$, its squared norm is at most $KD_NT$, even though $z$
depends on the assignment. The first prefix round has no incoming
coefficients. There are therefore at most $\ell$ nonzero blocks of
coefficients, including the suffix, and the total squared coefficient
budget is at most $\ell KD_NT$. The conditional $\cosh$ estimate gives
\[
 \Prb(Q_{\rm later}\le-b)\le e^{-2b^2/(\ell KD_NT)}.
\]
Take $b=gn\sqrt{D_N}/32$.

Condition on the rounds preceding $B_j$. Put $c_t=2p_t-1$ and
$\beta_t=-c_tu_t$. These are independent and uniform on $\{-1,1\}$, and
$t\in I$ precisely when $\beta_t=1$. The matrix
$M^{(j)}_{tw}=c_tc_wN_{tw}$ on $B_j$ has zero diagonal and norm at most
$\sqrt{KD_N}$. Direct expansion gives
\[
 Q_j=-\tfrac14\beta^{\mathsf T}M^{(j)}\beta
          -\tfrac14\langle M^{(j)}\ind,\beta\rangle.
\]
The vector $M^{(j)}\ind$ has squared norm at most $|B_j|KD_N$.
Put $B=TKD_N>0$ to use a common bound for every block, including empty
blocks. For $e=gn\sqrt{D_N}/(64\ell)$, the linear term is below $-e$ with
probability at most $e^{-8e^2/B}$. The second part of
Lemma~\ref{lem:counts-concentration} bounds the corresponding quadratic
probability by $e^{-e^2/(2B)}$. These bounds are uniform over the fixed
preferences, so averaging over the earlier-round history preserves them.
A union bound therefore gives $\sum_jQ_j\ge-gn\sqrt{D_N}/32$ outside
these events. On the remaining assignments the tested sum is at least
$(1/8-1/32-1/32)gn\sqrt{D_N}>0$.
Using $|B_j|\le T\le\tau n$, the total passing probability is at most
$(3+2\ell)e^{-c_*n}$. Lemma~\ref{lem:enumeration} finishes both time
bounds. Runs and candidates are streamed, so space is polynomial.
\end{proof}

\subsection{Projection grids and discovering the low-rank part}

\begin{corollary}[Polynomial projection grids]\label{cor:spectral}
Suppose $\|v_k\|_2\le1$ and use
$\mathcal Y_k=\delta\mathbb Z\cap[-\|v_k\|_1-\delta,\|v_k\|_1+\delta]$,
where $\delta>0$ is rational. Then
\[
 |\mathcal Y|\le(2\sqrt n/\delta+3)^r,\qquad
 \sum_{t\le T}s_t\le\sum_k|\lambda_k|(\delta\sqrt n+1).
\]
Thus Theorem~\ref{res:counts} applies whenever its norm and gap conditions
hold and the last expression is at most $gn\sqrt{D_N}/8$.
\end{corollary}
\begin{proof}
The Cauchy--Schwarz inequality gives
$\|v_k\|_1=\sum_t|v_k(t)|\le\sqrt{\sum_t1^2}\sqrt{\sum_tv_k(t)^2}
\le\sqrt n$. Also
$\sum_ts_t=\sum_k|\lambda_k|(\delta\|v_k\|_1+\|v_k\|_2^2)$.
\end{proof}

The vectors need not be exact eigenvectors. One deterministic construction
uses a rational, approximately diagonalizing orthogonal matrix, takes its
$r$ largest diagonal entries in absolute value, and uses the corresponding
columns. Appendix~\ref{app:discovery} gives a polynomial-bit construction
and a random-model guarantee. All pruning comparisons use the rational
vectors actually computed, so numerical approximation affects the running
bound, not exactness.

\section{Proofs for the Max-Cut analysis}\label{app:maxcut}

This appendix contains the proof omitted from Section~\ref{sec:maxcut}.

\subsection{Regular expanders satisfy the accumulation conditions}\label{app:expander_conditions}

\begin{proof}[Proof of Lemma~\ref{lem:expander_conditions}]
Choose an integer $n_0$ with
\[
 n_0\ge80,\qquad \varepsilon n_0\ge2,\qquad g^2\varepsilon n_0\ge128;
\]
it depends only on $\varepsilon$, since $g=\varepsilon/(320\sqrt3)$.
Let $n\ge n_0$, and put $p=d/(n-1)$, so that $\beta=-2p$.

\emph{The norm condition.} The Max-Cut encoding has $J=-2A_G$; unit clauses
change only the linear coefficients. Hence
\[
 N=-2A_G-\beta(\ind\ind^{\mathsf T}-I).
\]
The matrix $N$ annihilates $\ind$, and its other eigenvalues are
$-2\lambda_i(A_G)+\beta$ for $i\ge2$. Since $|\beta|\le2$, and
$2\le\lambda_0\sqrt d/4$ because $d\ge64\lambda_0^2$,
\[
 \|N\|\le2\lambda_0\sqrt d+2\le3\lambda_0\sqrt d.
\]
A row of $N$ has $d$ entries $-2(1-p)$ and $n-1-d$ entries $2p$. Its
squared Euclidean length is therefore $4d(1-p)$ for every row, and this
also equals their average $D_N$. Since $n\ge2/\varepsilon$ and
$d\le(1-\varepsilon)n$, we have $1-p\ge\varepsilon/2$, and so
\[
 D_N=4d(1-p)\ge2\varepsilon d>0,\qquad
 \|N\|^2\le9\lambda_0^2d\le\frac{9\lambda_0^2}{2\varepsilon}D_N=KD_N.
\]

\emph{The deviations between the rounds.} An edge between the two prefix rounds contributes at least
$\varepsilon^2$ to $W_t$. Indeed, its coefficient has absolute value
$2(1-d/(n-1))\ge\varepsilon$ when $n\ge2/\varepsilon$.
Write $e(B_1,B_2)$ for the number of edges between the rounds, and
$\deg_{B_1}(t)$ for the number of neighbors of $t$ in $B_1$.
The expander mixing inequality gives, for $n\ge80$,
\[
 e(B_1,B_2)\ge \frac{d|B_1||B_2|}{n}
           -\lambda(G)\sqrt{|B_1||B_2|}
 \ge\frac{dn}{32}-\frac d2\ge\frac{dn}{40}.
\]
Indeed $|B_1||B_2|\ge(n/4-1)^2$,
$\sqrt{|B_1||B_2|}\le n/4$, and $\lambda(G)\le d/8$.
The mixing inequality itself follows by writing each indicator vector
as its constant component plus a vector perpendicular to $\ind$.
Specifically, if
$z_i=\mathbf 1_{B_i}-(|B_i|/n)\ind$, then
\[
 e(B_1,B_2)-\frac{d|B_1||B_2|}{n}=z_1^{\mathsf T}A_Gz_2,
 \qquad
 |z_1^{\mathsf T}A_Gz_2|
 \le\lambda(G)\|z_1\|_2\|z_2\|_2
 \le\lambda(G)\sqrt{|B_1||B_2|}.
\]
The middle inequality uses the spectral bound on the subspace
perpendicular to $\ind$ and the Cauchy--Schwarz inequality.
Consequently
\[
 \sum_{t\in B_2}\sqrt{W_t}
 \ge\varepsilon\sum_{t\in B_2}\sqrt{\deg_{B_1}(t)}
 \ge\frac{\varepsilon e(B_1,B_2)}{\sqrt d}
 \ge\frac{\varepsilon n\sqrt d}{40}.
\]
The middle step uses $\sqrt a\ge a/\sqrt d$ for $0\le a\le d$
and $\sum_{t\in B_2}\deg_{B_1}(t)=e(B_1,B_2)$.
Since $\sqrt{D_N}\le2\sqrt d$, this gives
$\sum_{t\in B_2}\sqrt{W_t}\ge4\sqrt3\,gn\sqrt{D_N}$ for
$g=\varepsilon/(320\sqrt3)$.

\emph{The allowance.} There are $T=\lfloor n/2\rfloor$ comparisons, each
with allowance $|\beta|$, and $T|\beta|\le2d$. For $n\ge n_0$, the
inequalities $D_N\ge2\varepsilon d$ and $d\le n$ give
\[
 (16d)^2\le256nd\le2g^2\varepsilon n^2d\le g^2n^2D_N.
\]
Taking square roots yields $T|\beta|\le2d\le gn\sqrt{D_N}/8$, the third
condition.
\end{proof}

\section{Estimates for independent pair coefficients}\label{app:membership}

For this appendix, $P_{uv}=\sum_k\beta_kt_k(u)t_k(v)$ with vectors
$t_k\in\{-1,0,1\}^n$, $N=\operatorname{offdiag}(J-P)$,
$\sigma=\sum_k|\beta_k|$,
$T=\lfloor\tau n\rfloor$, $R_{tw}=N_{tw}[w>t]$ for $t\le T$, and
$W_t=\sum_{j<t}N_{tj}^2$. Exact count grids have $s_t\le\sigma$.
Thus the following imply the sequential case of Theorem~\ref{res:counts}:
\begin{equation}
 \|R\|^2\le KD_N,\qquad
 \sum_{t\le T}\sqrt{W_t}\ge4\sqrt3\,gn\sqrt{D_N},\qquad
 \sigma\le g\sqrt{D_N}/(8\tau).\label{eq:remainder}
\end{equation}

The following criterion verifies the sequential accumulation conditions
for independent pair coefficients. We state it for integer coefficients,
as they arise from Max 2-SAT formulas; the marginal
probabilities may be arbitrary real numbers. Its random-matrix estimate and
its estimates for sums of squared entries are also used in
Appendix~\ref{app:discovery}. We write $\chi$
for the variance ratio to distinguish it from clause density. The probability event concerns $J$ alone, which is
why the linear coefficients and unit clauses can be arbitrary.

\begin{proposition}[Independent-pair membership]\label{prop:pairmembership}
Fix $0<\tau<1$, $\chi\ge1$, $C_0=10^9\chi/\tau$, and
$n\ge10^9\chi^2/\tau^2$. Suppose the integer coefficients $J_{uv}$,
$u<v$, are independent random variables on a probability space, with
\[
 |J_{uv}-\E J_{uv}|\le B,\qquad
 0<\nu\le\operatorname{Var}J_{uv}\le\chi\nu.
\]
Choose either of the following decompositions:
\begin{enumerate}
\item All means equal $\mu$, and $P_{uv}=\beta$, where $\beta$ is the
empirical mean over pairs. Put $\sigma_0=|\mu|$.
\item A fixed supplied decomposition of the form above has
$P_{uv}=\E J_{uv}$ for $u\ne v$.
Put $\sigma_0=\sum_k|\beta_k|$. In particular, if all means vanish,
take $P=0$ and $\sigma_0=0$.
\end{enumerate}
If $n\nu\ge C_0(B^2\log n+\chi\sigma_0^2)$, then with probability at least
$1-1/n$, conditions~\eqref{eq:remainder} hold with
$K=9\chi$ and $g=\tau^{3/2}/(12\sqrt\chi)$.
The lower bound on $n$ in Theorem~\ref{res:counts} also holds.
\end{proposition}

\begin{proof}
Each bounded integer coefficient takes only finitely many values with
probability one. Independence identifies their joint distribution with
the product of their marginal distributions. We may therefore compute
events and moments using this finite product, even when the original
probability space is infinite. This replacement preserves the original
expectations and variances exactly.

Put $X_{uv}=J_{uv}-\E J_{uv}$, and let $\delta=\beta-\mu$ in the first
case and $\delta=0$ in the second. Then $N_{uv}=X_{uv}-\delta$ off the
diagonal. Set $\eta=1/100$.

The scalar concentration steps use the elementary variance-sensitive
estimate proved after this proof; it is a sufficient form of Bernstein
concentration~\cite{BLM13}. For the matrix norm we use the Rademacher
extension of Bandeira and van Handel's rectangular expectation estimate
\cite[Theorem~3.1, Corollary~3.2, and the paragraph following it]{BvH16},
and the upper-tail inequality for convex Lipschitz functions in
\cite[Corollary~3.3 and its interval rescaling, pp.~165--166]{Ledoux99}.
The last estimate, for independent random signs $\xi_{tw}$ and a fixed
coefficient array $a$, is
\[
 \E\|(\xi_{tw}a_{tw})\|
 \le(1+\vartheta)\left(\sigma_1+\sigma_2+
       \frac{5\max_{t,w}|a_{tw}|\sqrt{\log n}}
            {\sqrt{\log(1+\vartheta)}}\right),\quad
 0<\vartheta\le1/2,
\]
where $\sigma_1,\sigma_2$ are the largest row and column Euclidean norms.
The published logarithmic factor is $\sqrt{\log(\min\{T,n\})}$ for a
$T\times n$ array. We may instead pad the array with zero rows and use the
$n\times n$ version displayed above; padding does not change its operator
norm. The cited expectation estimate applies to independent random signs
multiplying fixed entries. We will apply it only after symmetrization.

For a convex function $f$ that is $1$-Lipschitz in Euclidean distance,
Ledoux's inequality, applied to independent coordinates in $[-B,B]$, gives
\[
 \Prb(f>\E f+r)\le \exp\!\left(-\frac{r^2}{8B^2}\right),\qquad r>0.
\]
This is an upper tail around the expectation, so no conversion from a
median is needed. In our application $f$ is the operator norm of the
upper-triangular array. It is convex, and changing the independent pair
coordinates changes this array's Frobenius norm by at most their Euclidean
distance: each coordinate is used at most once. The triangle inequality
for the operator norm, together with its bound by the Frobenius norm,
therefore gives the required Lipschitz constant.

\paragraph{Centering and total variance.}
First observe that $\nu\le B^2$: the variance of each centered
entry is its second moment, and its square is always at most $B^2$.
The assumption $n\nu\ge C_0(B^2\log n+\chi\sigma_0^2)$ therefore gives the
following useful budgets:
\begin{equation}
 10^9\chi\log n\le\tau n,\qquad
 10^9\chi B^2\log n\le\tau n\nu,\qquad
 10^9\chi^2\sigma_0^2\le\tau n\nu.
 \label{eq:membership-budgets}
\end{equation}
The size assumption implies $n\ge100$, $\tau n\ge1000$, and $\chi\le n$.
These deliberately loose consequences suffice for all rounding and
constant estimates below.

Set $\delta_0=6\sqrt{\chi\nu\log n}/n$. The scalar estimate applies
because $\delta_0 B\le2\chi\nu$, as follows from the second budget.
In the empirical case it gives $|\delta|\le\delta_0$ except with
probability at most $n^{-3}$. Indeed, its two-sided bound is
$2\exp[-4.5(1-1/n)\log n]\le2n^{-4}\le n^{-3}$ for $n\ge100$.
The identity
\[
 n^2\delta_0^2=36\chi\nu\log n
\]
and the first budget imply
$\delta_0\le10^{-3}\sqrt\nu$ and
$n\delta_0\le0.002\sqrt{n\chi\nu}$.
In the supplied-mean case $\delta=0$, so both conclusions hold directly.

Applying the squared-entry estimate separately to the lower and upper
tails gives
\[
 (1-\eta)n(n-1)\nu\le\sum_{u\ne v}X_{uv}^2
           \le(1+\eta)n(n-1)\chi\nu.
\]
We call $\sum_{u\ne v}X_{uv}^2$ the \emph{total energy}. Each tail has
probability at most $n^{-3}$ under the second budget.
In the empirical case,
$D_N=n^{-1}\sum_{u\ne v}X_{uv}^2-(n-1)\delta^2$; in the supplied-mean
case, $D_N=n^{-1}\sum_{u\ne v}X_{uv}^2$. Consequently, whenever the
mean and energy bounds hold,
\begin{equation}
 (1-2\eta)(n-1)\nu\le D_N\le(1+\eta)(n-1)\chi\nu. \label{eq:random-D}
\end{equation}

\paragraph{The operator norm.}
Let $R_X$ have entries $X_{tw}$ for $t\le T$, $w>t$, and $0$ otherwise.
The entries in this array are independent. Symmetrize with an independent
copy $X'$ and put $\widehat X=X-X'$. Then
$\E\|R_X\|\le\E\|R_{\widehat X}\|$.
Independently changing the sign of any coordinate of $\widehat X$ preserves
its joint law: exchanging $X_{uv}$ and $X'_{uv}$ negates their difference,
and the pairs of copies are independent. We may therefore introduce
independent signs, condition on the differences, and apply the displayed
Rademacher estimate with $\vartheta=\eta$. This yields
\[
 \E\|R_X\|\le
 (1+\eta)(\E\sigma_1(\widehat X)+\E\sigma_2(\widehat X))
       +102B\sqrt{\log n}.
\]
Here $|\widehat X_{uv}|\le2B$ and
$\E\widehat X_{uv}^2=2\operatorname{Var}J_{uv}\le2\chi\nu$.
To bound both maxima, consider the incident squared mass
$H_i=\sum_{j\ne i}\widehat X_{ij}^2$ at each vertex. Every row-square
sum and every column-square sum of $R_{\widehat X}$ is bounded by the
corresponding $H_i$. The scalar estimate and a union bound give
\[
 \Prb\!\left(\max_i H_i>2(1+\eta)n\chi\nu\right)
 \le q_n:=n\exp\!\left(-\frac{\eta^2n\chi\nu}{8B^2}\right).
\]
Each $H_i$ is always at most $4nB^2$. Splitting the expectation according
to this exceptional event consequently gives, for $j=1,2$,
\[
 \E\sigma_j(\widehat X)^2
 \le2(1+\eta)n\chi\nu+4nB^2q_n
 \le2(1+2\eta)n\chi\nu.
\]
For completeness, the sufficient budget
$n\chi\nu\ge10^9B^2\log n$, which follows from the proposition's
assumptions, implies both
$4nB^2q_n\le2\eta n\chi\nu$ and
$102B\sqrt{\log n}\le0.004\sqrt{n\chi\nu}$.
The Cauchy--Schwarz inequality gives
$\E\sigma_j\le\sqrt{\E\sigma_j^2}$, so
\[
 \E\|R_X\|
 \le\bigl(2(1.01)\sqrt{2.04}+0.004\bigr)\sqrt{n\chi\nu}
 \le2.89\sqrt{n\chi\nu}.
\]
Ledoux's upper-tail inequality with $r=\eta\sqrt{n\chi\nu}$ now gives
$\|R_X\|\le2.90\sqrt{n\chi\nu}$ except with probability at most
$\exp(-n\chi\nu/(80000B^2))\le n^{-3}$.
After empirical centering, the entries need not be independent; instead
we use the deterministic perturbation bound
$\|R\|\le\|R_X\|+n|\delta|$. On the previously established event for
$\delta$, this yields
\[
 \|R\|\le2.902\sqrt{n\chi\nu},\qquad
 \|R\|^2\le8.43n\chi\nu\le9\chi D_N.
\]
In the supplied-mean case the same argument applies with $\delta=0$.

\paragraph{The earlier-neighbor mass.}
Put $a=\lceil\tau n/10\rceil$. For $t>a$, the squared-entry estimate gives
$V_t=\sum_{j<t}X_{tj}^2\ge(1-\eta)(t-1)\nu$ with failure probability at most
$e^{-\eta^2\tau n\nu/(40B^2\chi)}\le n^{-3}$.
On these events and the mean event, the triangle inequality gives
\[
 \sqrt{W_t}\ge\sqrt{V_t}-|\delta|\sqrt{t-1}
 \ge0.993\sqrt{(t-1)\nu}.
\]
We give the summation details, including the integer endpoints. Write
$H=\tau n$, so $T=\lfloor H\rfloor$ and $a=\lceil H/10\rceil$.
For every $x\ge0$,
\[
 (x+1)\sqrt{x+1}-x\sqrt x\le\tfrac32\sqrt{x+1}.
\]
For example, this follows by expanding the nonnegative expression
$(\sqrt{x+1}-\sqrt x)^2(\sqrt{x+1}+2\sqrt x)$.
Telescoping, and then bounding each of the first $a$ summands by
$\sqrt a$, yields
\[
 \sum_{j=a}^{T-1}\sqrt j
 \ge \tfrac23 T\sqrt T-\sqrt T-a\sqrt a.
\]
Since $H\ge1000$, we have $0.99H\le T\le H$,
$\sqrt T\ge0.994\sqrt H$, $a\le0.101H$, and
$\sqrt a\le0.318\sqrt H$. Substitution in the preceding display gives
\[
 \sum_{j=a}^{T-1}\sqrt j\ge0.62H\sqrt H,
 \qquad
 \sum_{t\le T}\sqrt{W_t}\ge0.61\,H\sqrt H\sqrt\nu.
\]
The shift from $t$ to $j=t-1$ is why the upper summation index is $T-1$.
To compare with the required bound, square the nonnegative expression
$4\sqrt3\,gn\sqrt{D_N}$ and use~\eqref{eq:random-D}:
\[
 (4\sqrt3\,gn\sqrt{D_N})^2
 =\frac{\tau^3n^2D_N}{3\chi}
 \le\frac{1.01}{3}\tau^3n^3\nu
 \le(0.60\,H\sqrt H\sqrt\nu)^2.
\]
This proves the gap bound in~\eqref{eq:remainder} with room to spare.

\paragraph{The comparison allowance.}
It suffices to prove $100000\chi\sigma^2\le\tau n\nu$.
Indeed, by~\eqref{eq:random-D} and $n\ge100$,
\[
 \left(\frac{g\sqrt{D_N}}{8\tau}\right)^2
 =\frac{\tau D_N}{9216\chi}
 \ge\frac{0.98\tau(n-1)\nu}{9216\chi}
 \ge\frac{\tau n\nu}{100000\chi}.
\]
The last budget in~\eqref{eq:membership-budgets} gives
$400000\chi\sigma_0^2\le\tau n\nu$.
For the empirical error, use the first budget, the identity for
$\delta_0^2$, and $\chi\le n\le n^2$ to obtain
$400000\chi\delta_0^2\le\tau n\nu$.
In either case $\sigma\le\sigma_0+\delta_0$; hence
\[
 100000\chi\sigma^2
 \le200000\chi(\sigma_0^2+\delta_0^2)
 \le\tau n\nu.
\]
All quantities being compared are nonnegative, so taking square roots
establishes the allowance condition.

There are at most $n+4$ failure events: one for the empirical mean,
two for total energy, at most $n$ for the rows, and one for the norm.
Their individual bounds are $n^{-3}$, so the union bound gives a total
failure probability at most $(n+4)n^{-3}\le1/n$. No independence among
these events is needed. Finally, using $\pi\le4$,
\[
 \frac{256\pi\tau K}{g^2}
 =331776\pi\,\frac{\chi^2}{\tau^2}
 \le10^9\frac{\chi^2}{\tau^2}\le n.
\]

\end{proof}

\paragraph{Scalar estimates used in the proof.}
Here is a direct proof of the scalar bounds, separating them from the
random-matrix estimate. Let $Z_1,\ldots,Z_r$ be independent centered random
variables with $|Z_i|\le b$ and $\sum_i\E Z_i^2\le v$, where $v>0$.
The elementary inequality $e^z\le1+z+z^2$ for $|z|\le1$ gives
\[
 \E e^{sZ_i}\le1+s^2\E Z_i^2\le e^{s^2\E Z_i^2}
 \qquad\text{when }|s|b\le1.
\]
Independence lets us multiply these inequalities, so
$\E e^{s\sum_iZ_i}\le e^{s^2v}$ on the same range of $s$.
If $a>0$ and $ab\le2v$, exponential Markov inequality with
$s=a/(2v)$ gives
\[
 \Prb\!\left(\sum_iZ_i\ge a\right)\le e^{-a^2/(4v)},\qquad
 \Prb\!\left(\sum_iZ_i\le-a\right)\le e^{-a^2/(4v)}.
\]
The second inequality follows by replacing every $Z_i$ by $-Z_i$.
The two-sided bound is at most twice this quantity. For arbitrary $a>0$,
enlarge the variance budget to $v+ba/2$; this yields the sufficient
upper-tail bound $e^{-a^2/(4v+2ba)}$.

For sums of squared entries, apply this argument to
$Y_i=X_i^2-\E X_i^2$. Since $0\le X_i^2\le B^2$, we have
\[
 |Y_i|\le B^2,\qquad
 \E Y_i^2=\E X_i^4-(\E X_i^2)^2\le B^2\E X_i^2.
\]
Thus the variance budget depends on the expected squared mass, rather
than simply on the number of terms times $B^4$. In particular, if
$\nu\le\E X_i^2\le\chi\nu$ and $0<\eta\le2$, then
\[
 \Prb\!\left(\sum_{i=1}^rX_i^2\le(1-\eta)r\nu\right)
 \le e^{-\eta^2r\nu/(4\chi B^2)},\qquad
 \Prb\!\left(\sum_{i=1}^rX_i^2\ge(1+\eta)r\chi\nu\right)
 \le e^{-\eta^2r\chi\nu/(4B^2)}.
\]
For the lower tail, use deviation $\eta r\nu$ and variance budget
$B^2r\chi\nu$; the actual expected sum is at least $r\nu$.
For the upper tail, use deviation $\eta r\chi\nu$ and the same budget;
the actual expected sum is at most $r\chi\nu$.
These give the total-energy and earlier-neighbor estimates above.
For the independent-copy differences, use
$|X_i-X_i'|\le2B$ and $\E(X_i-X_i')^2=2\E X_i^2$.
Both row and column deviations of $2\eta n\chi\nu$ can use the common
variance budget $8B^2n\chi\nu$, giving the displayed exponent with
denominator $8B^2$. The larger constants in the proposition absorb these
bounds without changing its input assumptions or conclusion.

All concentration is applied to the original independent variables.
Subtracting their empirical mean generally makes them dependent. The
subsequent centered-energy identity and row-norm inequality are
deterministic, and a union bound combines their required events; no
independence between those events is needed.

\section{Spectral preprocessing for hidden communities}\label{app:discovery}

\subsection{A discovery lemma for the stochastic block model}

\begin{lemma}[Discovering a bounded-rank mean]\label{cor:planted}
Fix a rank bound $r$ and variance ratio $\chi\ge1$. Suppose the integer
coefficients $J_{uv}$ are independent over $u<v$,
\[
 |J_{uv}-\E J_{uv}|\le B,\quad
 0<\nu\le\operatorname{Var}J_{uv}\le\chi\nu,
\]
where we take $B\ge1$. Suppose their means are the off-diagonal entries
of a symmetric matrix $\bar A$ with
$\rank\bar A\le r$, $|\bar A_{uu}|\le B_0$ and $\|\bar A\|\le\mu n$.
For a sufficiently large $C=C(r,\chi)$, if
\[
 n\nu\ge C(B^2\log n+B_0^2+\mu^2),
\]
there is a deterministic exact polynomial-space algorithm, given the fixed
class parameters $r,\chi$ but not $\bar A$ or $\nu,B,B_0,\mu$, taking
$\Ostar(2^{(1-\gamma)n})$ time with probability $1-o(1)$, for a fixed
$\gamma>0$. Arbitrary units, even chosen after the binary draw, are allowed.
\end{lemma}

We give the computational detail separately so that the general theorem
requires only rational comparisons, with no exact-real arithmetic oracle.

\paragraph{Rational approximate diagonalization.}
We only need accuracy $\zeta=n^{-4}$ for the matrix $J$ of pair
coefficients of an explicitly listed formula. The following finite grid construction suffices.
Its polynomial factors are deliberately generous; they do not affect the
exponential saving in the theorem.

Suppose first that $n\ge1$. For the rational symmetric input matrix $J$, set
\[
 S_0=\sum_{u,v}J_{uv}^2,\qquad
 M_0=\lceil S_0\rceil+1,\qquad H=32n^5M_0.
\]
Initialize $Q=I$ and $A=J$, and maintain $A=Q^{\mathsf T}JQ$.
At each step compute the off-diagonal energy
$E=\sum_{i\ne j}A_{ij}^2$. If $E\le n^{-8}$, stop.
Otherwise choose a largest off-diagonal entry in absolute value, say
$A_{ij}$, and write $a=A_{ii}$, $b=A_{ij}$ and $d=A_{jj}$.
For each $t\in\{0,1/H,\ldots,1\}$, consider the plane rotation whose
block on coordinates $i,j$ is
\[
 R(t)=\begin{pmatrix}c(t)&-s(t)\\s(t)&c(t)\end{pmatrix},\qquad
 c(t)=\frac{1-t^2}{1+t^2},\quad s(t)=\frac{2t}{1+t^2}.
\]
Here and below this block is extended by the identity on all other
coordinates. Its entries are rational and $c(t)^2+s(t)^2=1$, so it is
exactly orthogonal. The updated pivot in $R(t)^{\mathsf T}AR(t)$ is
\[
 f(t)=b\bigl(c(t)^2-s(t)^2\bigr)+(d-a)c(t)s(t).
\]
Choose a grid point minimizing $|f(t)|$, breaking ties by a fixed order,
and replace $Q$ by $QR(t)$ and $A$ by $R(t)^{\mathsf T}AR(t)$.
Allow at most $M_0n^{10}$ such steps. We next show that this number
always suffices to reach the stopping condition.

Orthogonal changes of coordinates preserve the sum of squared entries.
Consequently every intermediate matrix has squared Frobenius norm $S_0$,
and $|a|,|b|,|d|\le M_0$, since $S_0\le M_0\le M_0^2$.
The continuous function $f$ satisfies $f(0)=b$ and $f(1)=-b$.
The intermediate value theorem gives a real $t_*\in[0,1]$ with
$f(t_*)=0$. This real number is used only in the analysis.
To see why the rational grid approximates it sufficiently well, note that
\[
 (c(t)-c(u))^2+(s(t)-s(u))^2
 =\frac{4(t-u)^2}{(1+t^2)(1+u^2)}\le4(t-u)^2.
\]
Thus each coordinate changes by at most $2|t-u|$. Since $|c|,|s|\le1$,
the differences of $c^2$, $s^2$ and $cs$ are each at most $4|t-u|$.
Using $|b|\le M_0$ and $|d-a|\le2M_0$, we obtain
\[
 |f(t)-f(u)|\le16M_0|t-u|.
\]
A grid point lies within $1/H$ of $t_*$. Therefore the point selected
by the algorithm satisfies
\[
 |f(t)|\le\frac{16M_0}{H}=\frac{1}{2n^5}.
\]

Only two diagonal entries change. A direct calculation in their symmetric
$2$-by-$2$ block, together with preservation of the squared Frobenius norm,
gives the exact energy identity
\[
 E_{\mathrm{new}}=E-2b^2+2f(t)^2.
\]
The largest-pivot choice gives $E\le n^2b^2$. Hence whenever $E>n^{-8}$,
\[
 E-E_{\mathrm{new}}
 =2b^2-2f(t)^2
 \ge\frac{2E}{n^2}-\frac{1}{2n^{10}}
 >\frac{1}{n^{10}}.
\]
Initially $E\le S_0\le M_0$, and energy is always nonnegative.
It follows that at most $M_0n^{10}$ updates suffice. At termination, put
$D=\operatorname{diag}(A)$. Orthogonal invariance now gives
\[
 \|J-QDQ^{\mathsf T}\|_F
 =\|A-D\|_F=\sqrt E\le n^{-4}=\zeta.
\]
For $n=0$ the empty identity matrix needs no updates.

For a formula with $m$ explicitly listed clauses, $|J_{uv}|\le2m$:
changing another variable changes a score gain by at most $2m$.
Thus $M_0\le4n^2m^2+1$, so both the grid size and the number of updates
are polynomial in the input size. Each rotation at $t=j/H$ has common
denominator $H^2+j^2\le2H^2$. An update of $A$ multiplies its common denominator by at most
$(H^2+j^2)^2$, and an update of $Q$ multiplies its common denominator by
at most $H^2+j^2$. Numerator lengths also remain polynomial, since
$\|A\|_F^2=S_0$ and every entry of $Q$ has absolute value at most $1$.
The bit lengths therefore grow by only $O(\log H)$ per update, in addition
to the input bit lengths. Evaluating each finite grid and updating the
matrices takes polynomial time and space. This establishes the numerical
preprocessing bound needed here for explicitly listed unweighted formulas.

\paragraph{Retaining the largest diagonal entries.}
For $n<r$, use the zero representation $P=0$; these bounded dimensions
do not affect the asymptotic estimate for fixed $r$. Suppose now that
$n\ge r$.
Let $P$ retain the $r$ entries of $D$ of largest absolute value and the
corresponding columns of $Q$. It has rank
at most $r$ and a rational orthonormal representation. We explain directly
why this truncation preserves a good approximation. Let $A_0$ be any matrix
of rank at most $r$, and put $A'=Q^{\mathsf T}A_0Q$ and
$e=\|D-A'\|$. There can be at most $r$ diagonal entries of $D$ with
absolute value greater than $e$. Indeed, if there were more, the restriction
of the linear map $A'$ to the span of those coordinates would have rank
at most $r$, and hence a nonzero vector $x$ with $A'x=0$. Since $x$ is supported on these coordinates, we would have
\[
 \|(D-A')x\|_2=\|Dx\|_2>e\|x\|_2,
\]
contradicting the definition of $e$. Therefore every discarded entry has
absolute value at most $e$. Orthogonal changes of coordinates preserve the
operator norm, and the operator norm is at most the Frobenius norm. Thus
\[
 e=\|QDQ^{\mathsf T}-A_0\|
 \le\zeta+\|J-A_0\|,
\]
and the triangle inequality gives
\begin{equation}
 \|J-P\|\le\|J-A_0\|+2\zeta
 \qquad\text{whenever }\rank A_0\le r.\label{eq:spectral-tail}
\end{equation}
This argument needs no estimate for individual eigenvectors or separation
between eigenvalues. The retained coefficients are exactly the selected
diagonal entries of $Q^{\mathsf T}JQ$, and the retained directions are the
corresponding rational unit columns of $Q$.

\paragraph{Setup for the comparison search.}
Given the fixed class parameters $r,\chi$, choose a positive rational
$g$ between $1/(512\sqrt\chi)$ and $1/(256\sqrt\chi)$.
Fix the original variable order and two prefix rounds of sizes
$\lfloor n/4\rfloor$ and
$\lfloor n/2\rfloor-\lfloor n/4\rfloor$. If $n<4$, use exhaustive
search. Otherwise, compute $P$ by the preceding construction with
$\zeta=n^{-4}$, then form
$N=\operatorname{offdiag}(J-P)$ and
$D_N=n^{-1}\sum_{u\ne v}N_{uv}^2$.
If $D_N=0$, use exhaustive search.

When $D_N>0$, let $\Lambda=\sum_{k\le r}|\lambda_k|$ for the retained
directions. If $\Lambda=0$, take $P=0$ and use one run with no projection
guesses. Otherwise, the required spacing can be computed using only an
integer square root. Write the positive rational number
\[
 q=\frac{g^2nD_N}{(64\Lambda)^2}=\frac{u}{v},\qquad u,v\in\mathbb N_{>0},
\]
and put $\delta=\lfloor\sqrt{uv}\rfloor/v$. For every positive integer
$m$, we have $\sqrt m/2\le\lfloor\sqrt m\rfloor\le\sqrt m$. Thus
\begin{equation}
 \frac{g\sqrt{nD_N}}{128\Lambda}\le\delta
 \le\frac{g\sqrt{nD_N}}{64\Lambda}.\label{eq:discovery-spacing}
\end{equation}
Replace $\delta$ by $\max\{\delta,n^{-2}\}$. Construct the grids of
Corollary~\ref{cor:spectral} with this spacing and their comparison
allowances. Now execute the comparison routine of
Appendix~\ref{app:accumulation} with these grids and rounds.
Since each direction has Euclidean norm one, the cap bounds the total
number of tuples by $(2n^3+3)^r$ on arbitrary inputs. Each grid still
covers every Boolean projection, so
the routine remains exact whether or not the input satisfies the
probabilistic guarantees.

\begin{proof}[Proof of Lemma~\ref{cor:planted}]
As in Proposition~\ref{prop:pairmembership}, bounded integer coefficients
can be represented exactly by their finite marginal distributions on any
original probability space. All probability estimates below concern only
these binary coefficients. Once the resulting event holds, the deterministic
node bound holds for every choice of the unit clauses, including choices
made after observing the draw.

We verify the running-time guarantee for this construction. We may assume
$n\ge\max\{4,r\}$, since the remaining input sizes are bounded in terms
of the fixed class parameter $r$.
Put $X=\operatorname{offdiag}(J-\bar A)$. With $\zeta=n^{-4}$, write
$E_0=\|X\|+B_0$, so $J=\bar A+X-\operatorname{diag}(\bar A)$ and
$\|J-\bar A\|\le E_0$. Apply~\eqref{eq:spectral-tail} with
$A_0=\bar A$ to obtain
\[
 \|J-P\|\le E_0+2\zeta,\qquad
 \|N\|\le2E_0+4\zeta,
 \quad N=\operatorname{offdiag}(J-P).
\]
For $H=\bar A-P$, subadditivity of rank gives $\rank H\le2r$,
and the triangle inequality gives $\|H\|\le2E_0+2\zeta$.
To pass to the Frobenius norm, let $\theta_1,\ldots,\theta_n$ be the
eigenvalues of the symmetric matrix $H$, counted with multiplicity.
At most $2r$ are nonzero, and each has absolute value at most $\|H\|$.
The spectral theorem therefore gives
\[
 \|H\|_F^2=\sum_{i=1}^n\theta_i^2\le2r\|H\|^2,
 \qquad
 \|H\|_F\le\sqrt{2r}(2E_0+2\zeta).
\]

We use the already established upper-triangular matrix estimate from
Appendix~\ref{app:membership}. If $U$ is the strict upper triangle of
$X$, then $X=U+U^{\mathsf T}$, and consequently
\[
 \|X\|\le2\|U\|\le5.8\sqrt{\chi n\nu}
\]
except with probability at most $n^{-3}$, for sufficiently large $C$.
This uses the same applications of Bandeira and van Handel and Ledoux
as the membership proof. Bernstein's inequality and a union bound also
give, simultaneously with probability $1-o(1)$,
\[
 .99n(n-1)\nu\le\|X\|_F^2\le1.01\chi n(n-1)\nu,
 \qquad
 \sum_{s\in B_1}X_{ts}^2\ge.99|B_1|\nu\quad(t\in B_2).
\]
To make the probability accounting explicit, it suffices to require
$n\ge20$ and $n\nu\ge10^9\chi B^2\log n$. Each cross-round row
contains $|B_1|\ge n/5$ distinct independent pair entries. The scalar
lower-tail bound with relative deviation $.01$ gives failure probability
at most $n^{-3}$ for each row. Each of the two total-energy tails also
costs at most $n^{-3}$. Together with the matrix event, a union bound gives
at most $(|B_2|+3)n^{-3}\le1/n$ failure probability. No independence
between these events is needed. The displayed density requirement is
absorbed by the choice of $C=C(r,\chi)$.

Here $B_0^2\le n\nu/C$, and $B\ge1$ ensures that the numerical error
$\zeta$ is negligible on this scale. In particular, for sufficiently
large $n$ we can use the convenient bounds
\[
 E_0\le7\sqrt{\chi n\nu},\qquad
 \|N\|\le18\sqrt{\chi n\nu},\qquad
 \|H\|_F\le16\sqrt{2r\chi n\nu}.
\]
For $n\ge5{,}120{,}000r\chi$, the last quantity is at most
$.01n\sqrt\nu$, as follows by squaring both nonnegative sides.
This is a fixed size cutoff because $r$ and $\chi$ are fixed.

We next explain why this error is small enough for the search analysis.
Deleting diagonal entries cannot increase the Frobenius norm. Hence
$N=X+\operatorname{offdiag}H$ gives
\[
 \bigl|\|N\|_F-\|X\|_F\bigr|\le\|H\|_F\le.01n\sqrt\nu.
\]
Squaring the resulting lower and upper bounds, and using $n\ge100$,
yields
\[
 .9n\nu\le D_N\le1.1\chi n\nu,\qquad \|N\|^2\le400\chi D_N.
\]
The same error estimate controls the sum of cross-round row norms.
By the triangle inequality in each row and the Cauchy--Schwarz inequality,
its decrease is at most
\[
 \sum_{t\in B_2}\left(\sum_{s\in B_1}H_{ts}^2\right)^{1/2}
 \le\sqrt{|B_2|}\left(\sum_{t\in B_2}\sum_{s\in B_1}H_{ts}^2\right)^{1/2}
 \le\sqrt{|B_2|}\,\|H\|_F.
\]
For $n\ge20$, both rounds have at least $n/5$ variables. Their raw
row sum is at least $|B_2|\sqrt{.99|B_1|\nu}$, whereas the displayed
loss is at most $.01n\sqrt{n\nu}$. Therefore
\[
 \sum_{t\in B_2}\left(\sum_{s\in B_1}N_{ts}^2\right)^{1/2}
 \ge .07n\sqrt{n\nu}
 \ge4\sqrt3\,gn\sqrt{D_N}.
\]
The last inequality uses $D_N\le1.1\chi n\nu$ and
$g\le1/(256\sqrt\chi)$. This is precisely the gap bound
in~\eqref{eq:profile-slack} for two rounds: the first round has no
earlier-round entries, while the second round sees the first. We can consequently take $K=400\chi$.

The retained diagonal entries are Rayleigh quotients of $J$ on unit
vectors. The Cauchy--Schwarz inequality gives
$|v^{\mathsf T}Jv|\le\|v\|_2\|Jv\|_2\le\|J\|$ for each such
vector, even if it is not an exact eigenvector. Therefore
\[
 \Lambda\le r\|J\|\le r(\mu n+E_0)
             \le r\bigl(\mu n+7\sqrt{\chi n\nu}\bigr).
\]
For clarity, sufficient conditions for the desired budget are
\[
 (256r\mu)^2\le g^2n\nu,\qquad
 1792r\sqrt\chi\le gn.
\]
The first follows by increasing $C$ in the assumed bound on $\mu^2$;
the second holds for sufficiently large $n$, since $r,\chi,g$ are fixed.
They bound each of the two terms in the last display by
$gn\sqrt{n\nu}/256$. Using $D_N\ge.9n\nu$, we obtain
\[
 \Lambda\le\frac{gn\sqrt{n\nu}}{128}
          \le\frac{gn\sqrt{D_N}}{64}.
\]
If $\Lambda=0$, the comparison allowance is $0$ and the algorithm
executes one unadjusted comparison run. The same two-round norm and
gap estimates apply to this run, so the round case of
Theorem~\ref{res:counts} bounds its nodes by $(2+14n)c^n$ for the same
base $c<2$ used below. Otherwise, before the lower cap is applied, the chosen
spacing satisfies
$\delta\ge g\sqrt{nD_N}/(128\Lambda)\ge1/(2\sqrt n)$.
Thus the cap is inactive for sufficiently large $n$ on this event.
Corollary~\ref{cor:spectral} bounds the total comparison allowance
by $gn\sqrt{D_N}/32$ and gives
\[
 |\mathcal Y|\le(4n+3)^r=n^{O(r)}.
\]
This proves the promised bound by the round case of
Theorem~\ref{res:counts}. The constant $K$ can be replaced by a
slightly larger rational bound without affecting the conclusion.

We finish by choosing constants that work for the entire family, including
bounded input sizes. With the rational $g$ fixed as above, take
\[
 C_0=10^9\chi+1+(256r/g)^2
\]
and choose an integer $N_0$ at least
\[
 \max\{100,r,5{,}120{,}000r\chi,1792r\sqrt\chi/g\}.
\]
Suppose $n\ge N_0$ and the lemma's assumption
$n\nu\ge C(B^2\log n+B_0^2+\mu^2)$ holds with $C=C_0$.
Since $C_0\ge10^9\chi$, it implies the required noise density bound.
Since $C_0\ge1$, it gives $B_0^2\le n\nu\le\chi n\nu$.
Moreover, $\log n\ge1/2$, $B\ge1$, and $C_0\ge2$ imply $n\nu\ge1$;
hence $\zeta=n^{-4}\le1\le\sqrt{\chi n\nu}$.
Finally, $C_0\ge(256r/g)^2$ gives
$(256r\mu)^2\le g^2n\nu$. The choice of $N_0$ supplies the two
remaining size conditions. Thus none of these conditions depends on an
additional assumption on the input.

To make the bounded-size argument explicit, put
\[
 p_r(n)=(2n^3+4n+3)^r(2+14n).
\]
On arbitrary inputs, the capped grids and the enumeration bound give at
most $p_r(n)2^n$ search nodes, including the exhaustive-search and
zero-coefficient branches. On the useful event the stronger bound is
$(4n+3)^r(2+14n)c^n\le p_r(n)c^n$, where the round analysis supplies
one fixed $1\le c<2$. If $n<N_0$, the arbitrary-input bound is at most
$2^{N_0}p_r(n)c^n$. Therefore the same constants give
\[
 \Prb\!\left(\text{search nodes}\le2^{N_0}p_r(n)c^n\right)
 \ge1-1/n\qquad(n\ge1).
\]
The polynomial preprocessing and arithmetic overhead preserve the claimed
$\Ostar(c^n)$ bound, and all phases use polynomial space.

The algorithm uses only $J$, the fixed class parameters $r,\chi$, and
its computed $D_N,\Lambda$. No planted mean, variance parameter, or
spectral cutoff has to be supplied.
\end{proof}

\subsection{Proof of the block-model theorem}\label{app:sbm_proof}

\begin{proof}[Proof of Theorem~\ref{thm:sbm}]
Use Algorithm~\ref{alg:communities}, which is exact by
Lemma~\ref{lem:community_correctness}. Its running time follows once we
verify the hypotheses of the discovery lemma, Lemma~\ref{cor:planted}. In the Max 2-SAT encoding,
$J_{uv}=-2A_{uv}$ for the adjacency indicators $A_{uv}$. Conditional on
the partition, these coefficients are independent, with
\[
 \E J_{uv}=-2\pi_{ab},\qquad
 \operatorname{Var}J_{uv}=4\pi_{ab}(1-\pi_{ab}),\qquad
 |J_{uv}-\E J_{uv}|\le2.
\]
Let $Z_{ua}=[u\in V_a]$. The matrix
$\bar A=-2Z\Pi Z^{\mathsf T}$ agrees with the means off the diagonal,
has rank at most $k$: the map factors through the $k$ columns of $Z$.
Every entry has absolute value at most $2$, including the diagonal.
Consequently $\|\bar A\|_F^2\le4n^2$ and
$\|\bar A\|\le\|\bar A\|_F\le2n$.
In Lemma~\ref{cor:planted}, take
\[
 r=k,\quad \nu=4v_{\min},\quad B=B_0=\mu=2.
\]
Its threshold becomes
$4n v_{\min}\ge C_0(4\log n+8)$. For $n\ge100$ we have
$\log n\ge1/2$, so it suffices to take $C=5C_0$ in
\eqref{eq:sbmvariance}. The bounded sizes are covered by the uniform
multiplier in the discovery proof.
The lemma's algorithm does not use $\bar A$ or its probabilistic
parameters. Its good event depends only on the pair coefficients and its
bound is uniform over linear fields, proving the assertion about units.
\end{proof}

\section{Additional simplification rules}\label{app:simplify}

The algorithms analyzed in this paper use only preprocessing, substitution,
cancellation, and the sum test of Section~\ref{sec:framework}. Branching
algorithms for Max-SAT commonly use further rules, such as the domination
rule of Lemma~\ref{lem:unit_dom} and values forced by a constraint. We
record a procedure that combines them with pruning constraints. None of our
results needs it.

When pruning constraints are present, an assignment made by domination
must also respect them. The procedure \textnormal{\textsc{Simplify}}, defined
below, uses domination only when the literal being set
to $1$ occurs in no current constraint.

The procedure \textnormal{\textsc{Simplify}} receives the residual formula
$F$ and current constraints $\mathcal C$. It records any new assignments
in $\alpha$ and their score contribution in $s$, both local to this call.
The instruction \emph{assign $v=p$} updates the formula, its unassigned
variable set, and $s$, substitutes the value in every constraint, and
records the corresponding variable value in $\alpha$.

Each iteration first cancels opposite units and complementary literal
pairs in the constraints, then removes constraints that are already
satisfied. It checks individual constraints and their sum for
infeasibility; an empty collection has sum $0\le0$. If neither check
rejects, it first looks for a value forced by a constraint. If none is
found, it looks for a dominating literal whose assignment preserves
the constraints, using the condition stated above. After either type
of assignment, the procedure repeats these checks on the updated state.
When neither assignment rule applies, it returns that state. All choices
between eligible assignments follow fixed orders.

\begin{algorithm}[tp]
\caption{\textnormal{\textsc{Simplify}}$(F,\mathcal C)$}\label{alg:simplify}
\begin{algorithmic}[1]
\STATE $s\gets0$; $\alpha\gets\emptyset$
\WHILE{true}
  \STATE Cancel opposite units by Lemma~\ref{lem:unit_canc}, adding their constant to $s$.
  \STATE Cancel complementary literal pairs in each $C\in\mathcal C$, decreasing its bound accordingly.
  \STATE Let $r_C$ be its remaining number of occurrences; remove $C$ if $b_C\ge r_C$.
  \IF{some $b_C<0$, or Lemma~\ref{lem:constraint_sum} proves infeasibility}
    \STATE \textbf{return} \textsc{Infeasible}
  \ENDIF
  \IF{some literal $v$ occurs more than $b_C$ times in some $C\in\mathcal C$}
    \STATE Assign $v=0$, updating $F,\mathcal C,s,\alpha$.
  \ELSIF{an unassigned literal $v$ satisfies $d_1(\neg v)=0$, $d_1(v)\ge d_2(\neg v)$,
          and occurs in no $C\in\mathcal C$}
    \STATE Assign $v=1$, updating $F,\mathcal C,s,\alpha$.
  \ELSE
    \STATE \textbf{return} $(F,\mathcal C,s,\alpha)$
  \ENDIF
\ENDWHILE
\end{algorithmic}
\end{algorithm}

The last assignment rule preserves feasibility of the current constraints:
setting $v=1$ cannot increase any left-hand side, because only $\neg v$ can
occur there. Lemma~\ref{lem:unit_dom} shows that the soft score does not
decrease. Thus this rule preserves the best score subject to the current
constraints.

\section{Details of the formal verification}\label{app:verification}

This appendix describes the formalization summarized in
Section~\ref{sec:verification}.

The accompanying Lean~4 artifact formalizes the common score identities,
pruning constraints, their arithmetic checker, and the literal-gap
algorithm of Section~\ref{sec:bias}. The paper and the formalization use
the same fixed initial preferences and the same choice of the branching
set. In particular, the verified algorithm enumerates completions of the
sum, including completions that violate individual constraints.

The checked results include the simultaneous-optimum argument in
Lemma~\ref{lem:gap_correctness}, the expectation and influence estimates
in Lemma~\ref{lem:bias_statistics}, the candidate count in
Lemma~\ref{lem:bias_count}, and the concentration inequality in
Appendix~\ref{app:bias_concentration}. For Theorem~\ref{thm:bias}, Lean
proves correctness on the original input after preprocessing and the
explicit bound $(4+2n)2^{(1-\gamma)n}$ on the search nodes. It also proves
bounds on recursion depth and the sizes of the constraint representations.
For Corollary~\ref{cor:lineargap}, the formalization additionally proves
that preprocessing preserves effective gaps and does not increase
occurrence counts, that simplicity bounds each initial occurrence count
by $4(n-1)+2$, and that these facts imply the stated candidate and
search-node bounds. Simplicity identifies binary clauses whose literals
are written in opposite orders; repeated units created by preprocessing
are permitted in the reduced formula.
The final polynomial bit-overhead and space analysis uses the ordinary
implementation argument given in the proof; it is not a formal machine-cost
or compiler verification.

The deterministic correctness statement in Lemma~\ref{lem:exactA} is also
formalized. The verified procedure computes its observed graph, variable
order, rational weights, and stopping depth from the formula alone.
The proof includes independent optimization of variables occurring only
in unit clauses, score offsets under substitution and preprocessing,
the earliest-optimum argument with changing preferences, and exact
enumeration of the weighted sum. It covers repeated clauses, empty
support, and the case of zero variables.

In this formal implementation, a bit vector and a list of unassigned
variables represent a partial assignment. Residual formulas are
recomputed from the fixed preprocessed formula when needed; inherited
constraints retain the bounds they had when created, and assigned values
are substituted when the weighted sum is tested. The proved score and
minimum identities justify these representations. The executable search
streams the candidates and keeps a best assignment; the full decision
tree and output lists used in the proofs are not constructed by it.

For the random-input analysis, the artifact defines the finite sampling
models and verifies their endpoint symmetries and second moments. It also
formalizes Lemma~\ref{lem:observed_pairs}: conditioning on the observed
support preserves independence and symmetry and cannot decrease the
second moments. The sampled pair multisets are explicitly converted to
the formula representation, and their nonempty support is proved to agree
with the neighbor relation used by the algorithm.
Lemma~\ref{lem:random_exposure} is also formalized. The auxiliary experiment
has exactly the prescribed joint distribution of the formula's pair
multisets and an independent uniform assignment. The proof identifies
$\Delta_t$ both with the average score difference and with half the
residual effective gap used by the algorithm, including its initialization
and preprocessing. It checks which coefficients are determined by each
history, the conditional independence of the fresh signs and assignment
bit, the conditional mean and variance, and the bound $Y_t\le m_0^2d_t$
using the algorithm's observed degree. The later-deviation identity and
the joint conditional law of its signs are checked as well. These results
do not assume independence of the accumulated pruning constraints.

Lemma~\ref{lem:random_moments} is also formalized. The checked
proof includes the exact fourth moment of a signed sum, its expected
absolute value, the truncation loss, and the symmetry argument for an
arbitrary unit-clause shift. The second-moment assertion is proved on
an arbitrary probability space, with the stated nonnegativity and
integrability assumptions. A proved connection between the finite
sampling laws and probability measures makes this result available to
the random-input analysis.

Lemma~\ref{lem:random_ranges} is formalized on finite probability spaces,
which suffice for all applications here. This includes history-dependent
interval lengths with a pathwise budget, bounded differences for
independent inputs, and the case in which every influence is $0$.
The formal proof constructs the conditional expectations from the exact
history classes and proves the exponential bound for the entire sum.

Lemmas~\ref{lem:random_requirements} and~\ref{lem:random_compensation}
are also formalized, including their applications over the complete
prefix. For accumulated requirements, the checked proof uses the direct
bound $\sqrt{Y_t}\ge Y_t/(m_0\sqrt{d_t})$, the expected squared
coefficients, the exact independent-edge representation, and its
coordinate-sensitivity budget. This simpler proof preserves the stated
constant $c_1$. For compensation, the actual score-gain differences are
identified with the incoming contributions at each later variable.
Their conditional exponential estimates are combined using a checked
bound on the total degree budget.

The formalization also combines these lemmas into the joint feasibility
bound $3e^{-c_0r}$, proves the identity between the summed comparison
and $A_T+Q_T$, and transfers the event from the auxiliary experiment to
the original pair draws and an independent uniform assignment.
The expected number of satisfying assignments and its high-probability
bound are checked on a fixed observed graph. The formalization also
bounds all recursive calls of the initialized search, including rejected
calls, by the whole-prefix candidate count. It proves that the rational
constraints actually constructed by the algorithm have exactly the score
interpretation used in the analysis, including preprocessing, tie choices,
and the strict-comparison offset.

The transfer to arbitrary original labels and the averaging over observed
supports are also checked. The formal proof keeps the original
labels throughout: inactive variables contribute $0$ to the influence
budget, and the candidate count fixes their initialized values. Deleting
empty pair draws preserves the score comparisons and the candidate count.
The conditional product law on the retained pairs is derived from the
original input law. This yields the bound in terms of the number $r$ of
active variables; the empty support is handled separately with one
terminal call. Finally, monotonicity in $r\le n$ and averaging give
\[
 \E[\operatorname{Nodes}(F)]
 \le (4+6n)\,2^n e^{-c_0n},
\]
and
\[
 \Prb\!\left(\operatorname{Nodes}(F)>(4+2n)\,2^n e^{-c_0n/2}\right)
 \le 3e^{-c_0n/2}.
\]
These are bounds on the actual initialized search, including rejected
calls. No concentration or candidate-count estimate is left as an
assumption. Both concrete clause distributions are instantiated with
$m_0=4$ and $\kappa=3/4$. Theorems~\ref{thm:random_main} and~\ref{thm:weighted}
therefore have checked correctness and explicit node bounds; polynomial
arithmetic overhead and streamed space remain conventional implementation
arguments, as for the literal-gap and count-guessing results.

For count guessing, Lemma~\ref{lem:slice} is checked as an exact gain
identity and a comparison inequality for either polarity, including its
diagonal correction and allowance. Lemmas~\ref{lem:count_correctness} and~\ref{lem:counts-exact}
are formalized against the actual clause scores. The checked procedures
use rational comparisons, retain the original variable labels through
preprocessing, and stream suffixes through a common best-assignment
accumulator. The proof checks the diagonal correction in both exact
counts and approximate projections, the dependence of preferences only
on earlier rounds, and the survival of every optimum in an approximating
run. The projection grids are constructed explicitly and their coverage
is checked.

The accumulation analysis for count and projection guessing is checked
through the final explicit recursion-node bound. This includes the actual
comparison averages, first absolute moments and clipping, creator concentration,
and sequential compensation under the Euclidean operator-norm hypothesis.
For fixed rounds, Lean checks the decomposition into later-round and
within-round contributions, the conditional joint laws of current-round bits,
and the norm bounds under restriction and sign conjugation.

Both concentration estimates in Lemma~\ref{lem:counts-concentration} have
complete Lean proofs. The first is derived from a two-point entropy inequality,
induction over the Boolean coordinates, and a finite-sum exponential-moment
argument. A smooth approximation to the absolute value then proves
Lemma~\ref{lem:counts-profile}. The second uses a random partition and the
linear sign-sum bound. No concentration inequality remains an assumed input
of the final count/projection accumulation bound.

Lean checks the norm reductions, parameter substitutions, and exceptional-event
bounds in both cases of Theorem~\ref{res:counts}. It connects them to the nodes
visited by the rational comparison procedures across all guesses and proves
that the resulting exponential base is strictly below two. The count
specialization includes preprocessing, its at most $n+1$ guesses, and the
prescribed two-round prefix. The projection specialization includes the
constructed grids, their cardinalities, and their total comparison allowances.
Thus Lemma~\ref{lem:count_accumulation}, Theorem~\ref{res:counts}, and
Corollary~\ref{cor:spectral} have checked correctness and node bounds;
the polynomial arithmetic-cost and space arguments remain the conventional
implementation arguments in this paper.

Corollary~\ref{cor:maxcut} has a checked mathematical proof and
an explicit recursion-node bound for the complete regular-graph procedure.
The formalization derives its matrix assumptions from the adjacency
matrix's eigenvalues, retaining their multiplicities. It checks the graph's
Max 2-SAT encoding and cut-score identity, the centered operator-norm bound,
row energies, the mixing inequality, rounded prefix sizes, the resulting
gap bound, and the comparison allowance. The bounded-degree branch
includes the greedy independent set, its size guarantee, optimal completion
of the remaining unit clauses, and the exact enumeration count.
The final theorem chooses integer thresholds once for the family and
bounds all three branches by a polynomial factor times $c^n$ for one
$c<2$, including arbitrary input unit clauses and small graphs.
Correctness holds even for graphs outside the stated spectral family.
As above, the polynomial bit-cost and space arguments are conventional
implementation arguments rather than formal machine-cost certificates.
The cited spectral properties of the named Paley and Ramanujan families
have not been separately formalized.

Established results used from the literature need not have their proofs
formalized in this artifact. We record their precise statements, references,
and conventions, and verify our applications to them. In Lean these results
appear as explicit theorem hypotheses, not as new axioms or as claims of
completed Lean proofs. Accordingly, deductions that use them are verified
conditional on the listed external results. The standard-axiom audit alone
does not describe these mathematical dependencies; the artifact records
them separately in \texttt{external-results.json}.

For finite input distributions, both cases of
Proposition~\ref{prop:pairmembership} have a checked membership proof,
conditional on the two explicitly stated published matrix estimates.
Lean proves the scalar concentration, empirical centering, gap-bound
summation with integer rounding, comparison allowance, and the final
probability of at least $1-1/n$. The original density and size assumptions
imply all intermediate numerical conditions; none of the needed row,
allowance, or probability bounds is an extra hypothesis of the final
formal membership theorems. The additional size requirement of
Theorem~\ref{res:counts} is checked as well.

The formal sampling model uses one independent coordinate per unordered
pair and constructs the actual symmetric matrix. The statements for the
original sampled coefficients use precisely $J-\beta$ and $|\beta|$ in the
empirical case, and $J-P$ with the supplied coefficient sum in the other
case. Instantiations cover homogeneous Bernoulli edges with empirical
centering and arbitrary independent Bernoulli edges with a supplied mean,
including the membership estimates needed for stochastic block models.

For the matrix step, Lean verifies our application of the Rademacher
expectation estimate of Bandeira and van Handel and Ledoux's convex
tail bound: symmetrization, independent sign insertion, the Euclidean
operator norm, row and column moments, and all parameter substitutions.
These two published results remain explicit hypotheses. The formalization
also connects an arbitrary original probability space to the finite product
law: bounded integer coefficients have finite ranges almost surely, and
the encoding preserves independence, event probabilities, expectations,
and the usual measure-theoretic variance. Both cases of
Proposition~\ref{prop:pairmembership}, as stated for integer coefficients,
therefore have checked proofs under the original size and density
assumptions. A statement for arbitrary continuously distributed real
coefficients is not claimed here. Membership with a supplied mean is
separate from discovery of that mean, which is verified below.

The hidden-community verification includes the exact search wrapper:
both exhaustive-search branches, the zero-coefficient branch, rational grid
spacing computed by integer square root, and the spacing cap. For any total
routine that returns rational directions and coefficients, the search returns
an optimum. Its correctness does not require an accurate approximation.
Lean also proves the Frobenius perturbation inequalities for energy and
cross-round profiles and the node bounds for both positive and zero
coefficient mass. The latter case uses one unadjusted comparison run.
The formalization identifies the matrix of coefficients of the adjusted
gains used by the search with $\operatorname{offdiag}(J-P)$ and derives the coefficient budget from the
retained Rayleigh quotients.

For the original finite product law, the cross-round row estimates, total
energy bounds, and symmetric noise estimate hold jointly with checked
probability at least $1-1/n$. The symmetric estimate uses the same two
published-result hypotheses as before. Combining this event with explicit
guarantees on the numerical approximation gives a checked probability bound
for the actual search-node count, including the zero-coefficient case.
No gap-bound or coefficient-budget estimate is assumed in this final
conditional statement.

The numerical construction is also checked. Lean defines a total
rational grid search for each plane rotation and proves exact orthogonality,
the energy decrease, termination within $M_0n^{10}$ steps, and the final
Frobenius residual $n^{-4}$. It also bounds the grid and iteration counts
polynomially in the explicitly listed formula length. The bit-operation
and space bounds use the conventional argument in Appendix~\ref{app:discovery}.
Rational sorting then selects the largest absolute diagonal entries,
with checked unit columns and exact Rayleigh coefficients. The truncation
and rank estimates derive all four approximation guarantees. The spectral
theorem used in those estimates is already formalized in mathlib.

The formalized algorithm therefore computes its representation directly
from the input formula. It returns an optimum on every well-formed input,
and its search-node bound holds with checked probability at least $1-1/n$
under explicit density, mean-rank and size conditions. No numerical routine
or approximation guarantee remains an assumed input to this result.
Lean also verifies the final parameter substitution and bounded-size
running-time argument. For fixed class parameters, one rational $g$, one
density constant, one multiplicative constant, and one base $c<2$ work
uniformly over every input size and every independent bounded-integer
input distribution satisfying the discovery assumptions. The original
probability space may be infinite.

For the stochastic block model, the original sampling law is a product of
Bernoulli distributions, one per unordered pair. Lean defines the sampled
graph and its Max-Cut formula, proves the coefficient identity, factors the
mean matrix through the community labels, and bounds its norm and diagonal.
The formal minimum and maximum block variances are exactly those in
Theorem~\ref{thm:sbm}. Substitution into the discovery bound proves the
actual algorithm's node bound with probability at least $1-1/n$.
The constants are fixed before the number of vertices, partition,
probabilities, or additional units are supplied. Units may be an arbitrary
function of the sampled graph. The algorithm receives the resulting
formula and the fixed class parameters, with no hidden partition or mean.

Thus Theorem~\ref{thm:sbm} has the same verification status as the other main
running-time results: exactness and the mathematical node bound are checked,
conditional on the two recorded published probability estimates; polynomial
bit-operation overhead and space use the conventional paper arguments.
For Lemma~\ref{cor:planted}, the probability-space translation is checked
as well. Lean constructs a measurable event, determined by the binary
coefficients, of probability at least $1-1/n$. On this event the node bound
holds for every formula with those coefficients. Consequently the units
may be chosen after the draw, including using additional randomness;
no independence assumption on them is needed. The exactness theorem
continues to apply to every well-formed input.

The artifact includes a correspondence record for every numbered result
and the proposed theorem,
with its precise verification status and Lean declarations where available.
Apart from Corollary~\ref{cor:random_regular}, every numbered proved result
has checked mathematical content, subject to the two explicitly recorded
published results where used. Lemmas~\ref{lem:expander_conditions}
and~\ref{lem:bounded_degree} restate steps of the checked proof of
Corollary~\ref{cor:maxcut}.
For running-time results, the verified bounds count actual search nodes;
the polynomial bit-operation overhead and space arguments remain
conventional proofs in the paper. No numbered proved result is recorded
as partially formalized.
The proposed extension in Appendix~\ref{app:biased_subset}
has neither a proof in this paper nor a Lean proof.
This record distinguishes proved formal statements from future targets;
it does not treat a proof of a conditional intermediate statement as a
verification of its unproved hypotheses. Supporting Lean lemmas implement
or justify statements and procedures retained in this paper.
The correspondence record and statement fingerprints track the intended
translation; they do not prove semantic equivalence between the prose
and Lean statements. Independent review of that correspondence remains
separate from checking the formal proofs.

\section{Proposed extension: a biased subset with bounded external influence}
\label{app:biased_subset}

This appendix records a possible extension of the literal-gap analysis.
The statement below is a research target: we do not provide a proof or
claim Lean verification. It is not used by any result in the main text.

The global occurrence bound in Theorem~\ref{thm:bias} can be restrictive
when a formula contains a dense part whose variables have little bias.
One possible extension is to collect pruning constraints only from a
linear-sized subset $S$ of biased variables. Clauses involving no variable
of $S$ cancel in these comparisons. This suggests bounding the influence
of a variable through its connections to $S$, rather than through all
its occurrences in the formula.

For a preprocessed formula $F_0$, let $B_v$ denote the total number of
unit and binary occurrences of variable $v$, and let
\[
 g_v=d_2^{F_0}(v)-d_2^{F_0}(\neg v)
       +2\bigl(d_1^{F_0}(v)-d_1^{F_0}(\neg v)\bigr).
\]
For distinct variables $v,w$, let $m_{vw}$ be the number of binary
clauses containing both variables, irrespective of their signs and
counting multiplicities. For $w\notin S$, write
$r_S(w)=\sum_{v\in S}m_{vw}$.

\begin{target}[A biased subset attached to an arbitrary remainder]
\label{target:biased_subset}
Fix constants $\varepsilon,\eta\in(0,1]$ and $C\ge1$.
Let $F$ be a simple unweighted Max 2-SAT formula on $n$ variables and
let $F_0$ be its preprocessed formula. Suppose a set $S$ and a number
$d\ge1$ satisfy
\[
 |S|\ge\eta n,\qquad
 B_v\le d\ \hbox{and}\ |g_v|\ge\varepsilon d\quad(v\in S),
 \qquad r_S(w)\le Cd\quad(w\notin S).
\]
Given $F$ and such a set and scale, an optimal assignment can be computed
deterministically in polynomial space and time
$\Ostar(2^{(1-\gamma)n})$, for some
$\gamma=\gamma(\varepsilon,\eta,C)>0$.
\end{target}

Here $S$ need not be a connected component, and its boundary need not
be small. Variables in $S$ may have most of their neighbors outside
$S$. The external condition limits how many of their clause occurrences
can be concentrated at a single outside variable. Clauses entirely
outside $S$ are unrestricted apart from the original simplicity
assumption. In particular, that remainder may be much denser than the
part incident to $S$. Clauses with both variables in $S$ are also allowed.

The intended algorithm would fix the preferred values on $S$, choose
the half of $S$ carrying more total absolute gap, and use the same
prefix branching and summed-constraint enumeration as in
Section~\ref{sec:bias}. The proposed analysis would count only the
occurrences that can affect these selected comparisons. We leave both
the proof and its formal verification to subsequent work.

The supplied set in the target separates the counting question from
the question of recognizing the structure. An input-only result should
also explain how to choose $S$ in polynomial time. A concrete restricted
version could take
\[
 S_d=\{v:B_v\le d,\ |g_v|\ge\varepsilon d\}
\]
and test its size and external occurrence bounds for candidate scales.
Such tests need consider only polynomially many distinct sets, since
membership changes only at the occurrence and gap thresholds. This
does not establish that the tests find every set allowed by the target.

One motivation is an existing system together with a substantial set
of new decisions whose pairwise soft requirements create a consistent
literal bias. To obtain an application, one would still need a faithful
encoding and a justification of the size, bias, and external occurrence
conditions. We leave that investigation, including any extension to
hard input requirements, to the later work on applications.

\end{document}